%% file: arxiv.tex
\documentclass[11pt,letterpaper,dvipsnames]{article}
\usepackage[margin=1in]{geometry}
\usepackage[dvipsnames]{xcolor}
\usepackage{mathtools}
\usepackage[sort&compress]{natbib}
\usepackage{hyperref}
\hypersetup{
	colorlinks,
	linkcolor={red!50!black},
	citecolor={blue!50!black},
	urlcolor={blue!80!black}
}
\usepackage{subcaption}
\usepackage{tikz}
\usepackage{enumitem}
\setlist[itemize]{nosep}
\setlist[enumerate]{nosep}
\usepackage[font={footnotesize}]{caption}

\usepackage{amsmath}
\usepackage{amsthm}
\usepackage{amssymb}
\usepackage{cleveref}
\usepackage{varwidth}
\usepackage{nicefrac}

\usepackage{tikz}
\usetikzlibrary{calc,arrows.meta,patterns}
\tikzset{>=latex}
\usepackage{subcaption}

\definecolor{color1}{HTML}{3f31d6}
\definecolor{color2}{HTML}{d63152}
\definecolor{color3}{HTML}{31d673}

\definecolor{mppetrol}{cmyk}{1,.26,.45,.16}
\definecolor{mpgreen}{cmyk}{.35,0,.99,0}
\definecolor{mpgray}{gray}{.261}
\definecolor{myblue}{cmyk}{0, 0.45, 0.90, 0.17}

\definecolor{mpblue}{HTML}{03a6ff}
\definecolor{mpred}{HTML}{ff03a6}
\tikzstyle{node}=[circle,fill=mpgreen,inner sep=.5ex,minimum width=3ex, outer sep=.5mm]

\def\sx{4}
\def\sy{.75}
\tikzstyle{objectivefunction}=[mppetrol, very thick]

\newcommand{\sctau}[2]{(\textcolor{mpblue}{#1}, \textcolor{mpred}{#2})^{\top}}

\usepackage[textwidth=3cm,textsize=scriptsize]{todonotes}

\newtheorem{theorem}{Theorem}
\newtheorem{lemma}[theorem]{Lemma}
\newtheorem{claim}[theorem]{Claim}
\newtheorem{corollary}[theorem]{Corollary}
\newtheorem{definition}[theorem]{Definition}

\crefname{claim}{Claim}{Claims}

\usepackage{bm,upgreek}
\renewcommand{\vec}[1]{
\begingroup
\let\alpha\upalpha
\let\beta\upbeta
\let\theta\uptheta
\let\lambda\uplambda
\let\mu\upmu
\let\tau\uptau
\bm{\mathbf{#1}}
\endgroup
}

\newcommand{\mfrac}[2]{#1 / #2}

\newcommand{\N}{\mathbb{N}}
\newcommand{\R}{\mathbb{R}}

\newcommand{\MS}{M}
\newcommand{\FV}{F}
\newcommand{\sinkInflow}{f^-}
\makeatletter
\newcommand{\numContributingLinks}[1][\vec{\mu}]{\@ifstar{\eta}{\eta_{s}(#1)}}
\makeatother
\newcommand{\exitTime}{\omega}

\newcommand{\eps}{\varepsilon}
\newcommand{\opt}{\textsc{Opt}}

\newcommand{\arrival}{u}
\newcommand{\s}{d}
\newcommand{\vprior}{\vec{\lambda}^*}
\newcommand{\prior}{\lambda^*}
\newcommand{\alg}{\textsc{Alg}}

\title{
Optimizing Throughput and Makespan of Queuing Systems\\ by Information Design
}
\author{Svenja M.~Griesbach \and Max Klimm \and Philipp Warode \and Theresa Ziemke}
\date{}

\author{
Svenja M.\ Griesbach\thanks{Technische Universität Berlin, Institute for Mathematics, Email: \texttt{griesbach@math.tu-berlin.de}}.\and
Max Klimm\thanks{Technische Universität Berlin, Institute for Mathematics, Email: \texttt{klimm@tu-berlin.de}} \and
Philipp Warode\thanks{Humbold-Universität Berlin, School of Business and Economics,
Email: \texttt{philipp.warode@hu-berlin.de}.} \and 
Theresa Ziemke\thanks{Technische Universität Berlin, Institute for Mathematics and Technische Universität Berlin, Transport Systems Planning and Transport Telematics, Email: \texttt{tziemke@vsp.tu-berlin.de}}
}

\begin{document}

\maketitle

\begin{abstract}
We study the optimal provision of information for two natural performance measures of queuing systems: throughput and makespan. A set of parallel links (queues) is equipped with deterministic capacities and stochastic travel times where the latter depend on a realized scenario.
A continuum of flow particles (users) arrives at the system at a constant rate. 
A system operator knows the realization of the scenario and may (partially) reveal this information via a public signaling scheme to the flow particles.
Upon arrival, the flow particles observe the signal issued by the system operator, form an updated belief about the realized scenario, and decide on a link to use.
Inflow into a link exceeding the link's capacity builds up in a queue that increases the travel time on the link. Dynamic inflow rates are in a Bayesian dynamic equilibrium when the expected travel time along all links with positive inflow is equal at every point in time.

For a given time horizon $T$, the throughput induced by a signaling scheme is the total volume of flow that leaves the links in the interval $[0,T]$.
We show that the public signaling scheme that maximizes the throughput may involve irrational numbers and provide an additive polynomial time approximation scheme (PTAS) that approximates the optimal throughput by an arbitrary additive constant $\eps>0$. The algorithm solves a Langrangian dual of the signaling problem with the Ellipsoid method whose separation oracle is implemented by a cell decomposition technique. We also provide a multiplicative fully polynomial time approximation scheme (FPTAS) that does not rely on strong duality and, thus, allows to compute also the optimal signals. It uses a different cell decomposition technique together with a piece-wise convex under-estimator of the optimal value function.

Finally, we consider the makespan of a Bayesian dynamic equilibrium which is defined as the last point in time when a total given value of flow leaves the system. Using a variational inequality argument, we show that full information revelation is a public signaling scheme that minimizes the makespan.
\end{abstract}

\thispagestyle{empty}
\newpage
\setcounter{page}{1}

\section{Introduction}

Imagine you are managing an airport with several security lanes. Should you inform the passengers about the current state of the system in order to maximize its performance? On the one hand, making passengers aware of less congested lanes may decrease their time in the queue. On the other hand, providing this information may cause inefficiencies since lanes with large delay are not used to capacity.

In this paper, we study this question using the framework of Bayesian persuasion for a dynamic queuing model that has been first studied (in a more basic and deterministic setting) by \citet{Vickrey69}.
Specifically, we consider a set $[m] \coloneqq \{1,\dots,m\}$ of parallel links. The state of each link is stochastic and depends on a realized scenario $s \in [d]$ where $d$ is a constant.
First, each link~$i$ has a scenario-independent \emph{capacity} $\nu_i \in \mathbb{Q}_{>0}$.
Second, each link~$i$ has a \emph{travel time} $\tau_{i,s} \in \mathbb{R}_{\geq 0}$ that depends on the realized scenario~$s \in [d]$.
Initially, the links have no queues.
A steady flow of infinitesimally small users arrives at the links with a constant inflow rate  of $\arrival \in \mathbb{Q}_{>0}$. Upon arrival, each of the flow particles chooses a link leading to an inflow rate of $f_i(\theta)$ into each link~$i$ for all times $\theta$.
When at some time~$\theta$ the inflow rate $f_i(\theta)$ into a link~$i$ exceeds its capacity $\nu_i$, a queue builds up at link~$i$ at a rate of $f_i(\theta) - \nu_i$.
When at some time~$\theta$ the inflow rate into a link~$i$ is less than the capacity~$\nu_i$ and the link has a queue, the queue depletes at a rate of $\nu_i - f_i(\theta)$. 
This queuing dynamics leads to  uniquely defined queue lengths $z_i(\theta)$ for link~$i$ at time~$\theta$.
The total delay of a flow particle that arrives at the system at time $\theta$ and chooses link~$i$ when the realized scenario is $s$ is given by $\tau_{i,s} + z_i(\theta) / \nu_i$.

 The flow particles, however, do not know the realized scenario~$s$ and, instead, form only stochastic beliefs about the state of the system. Formally, a belief is a vector $\smash{\vec \mu = (\mu_1,\dots,\mu_d)^\top \in [0,1]^d}$ such that $\sum_{s\in[\s]} \mu_s = 1$ and $\mu_s$ is the (anticipated) probability that state $s$ is realized. In the following, we denote by $\Delta$ the set of all such beliefs.
All flow particles have initially a true prior belief $\smash{\vec \lambda^* = (\lambda_1^*,\dots,\lambda_d^*)^\top \in \Delta}$, e.g., from previous observations of the system.
If the flow particles do not receive any further information on the state of the system, each of them chooses a link that minimizes their expected delay in the system (where the expectation is taken according to the prior belief $\vec \lambda^*$), i.e., when arriving at time $\theta$ they choose a link~$i$ that minimizes $\sum_{s\in[\s]}
 \lambda^*_s \tau_{i,s} + z_i(\theta)/\nu_i$.
We say that a vector of inflow functions $\vec f = (f_{i}(\cdot))_{i \in [m]}$ is a \emph{dynamic equilibrium} (with respect to the belief~$\vec \lambda^*$) if this property of the particles' behavior is satisfied for almost all $\theta$.

The operator of the system, however, knows the realized scenario and, hence, the travel times of the links and can determine in how far this information should be shared with the flow particles.
To this end, the system operator commits to a \emph{public signaling scheme} $\varphi$.
Such a signaling scheme consists of a finite set of signals $\Xi$, as well as probabilities $(\varphi_{s,\xi})_{s \in [d], \xi \in \Xi}$ where $\varphi_{s,\xi}$ is the probability that scenario~$s$ is realized and signal~$\xi$ is issued. Since the prior belief $\lambda^*_s$ represents the true probability that scenario~$s$ is realized, we have the constraint $\sum_{\xi \in \Xi} \varphi_{s,\xi} = \lambda^*_s$ for all $s \in [\s]$.
When arriving at the system, the flow particles observe the issued signal~$\xi$ (but not the realized scenario~$s$) and perform a Bayesian update of their belief.
In particular, after having received signal $\xi$, the new belief is given by $\vec \mu^{\xi} = (\mu^\xi_1,\dots,\mu^\xi_d)^\top \in [0,1]^d$ defined as $\mu^\xi_s =
\mfrac{\varphi_{s,\xi}}{ \sum_{j \in [d]} \varphi_{j,\xi}}$.
After this Bayesian update, the flow particles choose a link~$i$ that minimizes the updated expected delays given by $\sum_{s\in[\s]} \mu^\xi_s \tau_{i,s} + z_i(\theta)/\nu_i$.
It is a standard observation in the field of Bayesian persuasion that there is a one-to-one correspondence between public signaling schemes and convex decompositions of the prior $\vec \lambda^*$ \citep[cf.][]{Dughmi14}. Specifically, every public signaling scheme yields updated beliefs $(\vec \mu^\xi)_{\xi \in \Xi}$ as well as corresponding probabilities $\alpha_{\xi} \coloneqq \sum_{j \in [d]} \varphi_{j,\xi}$ that signal $\xi$ is issued such that $\smash{\vec \lambda^* = \sum_{\xi \in \Xi} \alpha_{\xi} \vec \mu^{\xi}}$. Conversely, for every such (finite) convex decomposition of the prior, there is a corresponding public signaling scheme with a finite set of signals such that the updated beliefs (after receiving one of the signals) correspond exactly to the beliefs in the convex decomposition.

Suppose that for each belief $\vec \mu\in \Delta$ there was a unique dynamic equilibrium $\vec f(\vec \mu)$. Further assume that the system operator is interested in maximizing a certain functional $H(\vec f)$ of the flow vector $\vec f$.
Then, a natural question is to find the optimal public signaling scheme such that the expected value of the functional of the resulting dynamic equilibrium is optimized.
Using the one-to-one correspondence between public signaling schemes and convex decompositions of the prior, this maximization problem can be phrased as
\begin{align}
\label{eq:generic-optimization}
\sup \Biggl\{ \sum_{\xi \in \Xi} \alpha_{\xi} \,H\bigl(\vec f(\vec \mu^\xi)\bigr) \;\Bigg\vert\;  |\Xi| < \infty, \alpha_{\xi} \in [0,1], \vec \mu^{\xi} \in \Delta \text{ for all }\xi \in \Xi \text{ and } \sum_{\xi \in \Xi} \alpha_{\xi} \vec \mu^{\xi} = \vec \lambda^* \Biggr\}.
\end{align}

\paragraph{Our Results and Techniques}

We study the generic optimization problem \eqref{eq:generic-optimization} for two natural objectives of a system operator. The first objective is to maximize the expected \emph{throughput} of the system.
Informally, for a given time horizon $T \in \mathbb{Q}_{>0}$, the throughput is the amount of flow that has left the links up to time~$T$.
We first give an example exhibiting that the signaling scheme maximizing the throughput may involve irrational numbers, even though all input numbers are rational (cf.~Appendix~\ref{sec:irrational-example}).
To avoid the issue of representing irrational numbers with finite precision, we resort to approximating the maximal throughput that can be achieved by public signaling schemes.
To this end, we first provide an additive polynomial-time approximation scheme (PTAS), i.e., for any $\eps > 0$, we provide an algorithm that runs in polynomial time and computes a value $p$ such that $p \in [\opt-\eps,\opt]$, where $\opt$ is the maximal throughput that can be achieved by public signals (\Cref{thm:ptas-additive}).
We stress that, unlike other PTAS for signaling \citep[e.g.][]{Bhaskar16,Cheng15}, we do \emph{not} require that the functional is normalized, i.e., that $||H||_{\infty} =1$.

To prove the result, we consider a Lagrangian dual of the primal signaling problem and show that strong duality holds.
The proof of strong duality is non-trivial since the objective of the primal is non-convex and non-concave such that standard constraint qualifications such as Slater's cannot be applied. Duality has been used before in the context of signaling by \citet{Bhaskar16}, but they use standard linear programming duality for an approximate version of the primal problem such that they can only show a (small) bound on the duality gap.
Our dual signaling problem has a finite number of variables but an uncountable number of constraints.
Yet, we are able to show that we can solve the separation problem for the dual signaling problem in polynomial time. 
To this end, we show that the separation problem for the dual can be reduced to finding the global maximum of a piecewise quadratic function whose quadratic parts have a polytopal domain. 
Using a cell decomposition technique together with the reverse search algorithm by \citet{Avis96}  allows to compute the global maximum exactly, thus, solving the separation problem.
Finally, we use the Ellipsoid method and the equivalence of optimization and separation to obtain the result.

While the \emph{additive} PTAS yields a compelling approximation of the optimal throughput achievable by signaling, it does not allow to compute the corresponding signals.
The underlying reason for this fact is that the approximately optimal \emph{dual} solution obtained by the Ellipsoid method does not seem to provide any useful information on how approximately optimal \emph{primal} solutions may look like. 
To close this gap, we propose a fully polynomial-time approximation scheme (FPTAS) with a \emph{multiplicative} approximation guarantee that allows to compute the corresponding signals (\Cref{thm:PTAS-throughput}). 
For the multiplicative FPTAS, again a main issue is that the objective is a non-convex and non-concave function on the space of beliefs $\Delta$. We propose a non-uniform discretization of $\Delta$ that leads to a piece-wise convex under-estimator of the objective. By controlling the approximation error of the under-estimator, we are able to compute signals such that the expected throughput $\alg$ achieved by the signals satisfies $\alg \geq (1-\eps)\opt$. 

The second objective that we study is the expected \emph{makespan}. For a given time horizon $T \in \mathbb{Q}_{>0}$, the makespan is the latest point in time a flow particle that departed in the time interval $[0,T]$ leaves the system.
From a mathematical point of view, makespan minimization is challenging since the makespan of the dynamic equilibrium $\vec f(\vec \mu)$ as a function of the belief $\vec \mu$ is non-continuous (see for instance the example given in Appendix~\ref{sec:makespan-example}).
In order to still be able to analyze optimal signaling for the minimization of the makespan, we show a general property for dynamic equilibria.
Suppose we are given a system with a vector of deterministic travel times $\vec \tau = (\tau_i)_{i \in [m]}$. Further let $\vec \tau' = (\tau'_i)_{i \in [m]}$ be a vector of (potentially different) deterministic travel times and let $\vec f(\vec\tau')$ be a dynamic equilibrium where particles act as if the travel times were $\vec\tau'$.
Then, we show that the makespan of the dynamic equilibrium $\vec f(\vec \tau')$ is minimized when $\vec\tau' = \vec\tau$ (up to constant shifts). This implies in particular that full information revelation is always optimal for makespan minimization, i.e., it is optimal to set $\Xi = [d]$ and have $\varphi_{s,\xi} = \lambda^*_s$ whenever $s=\xi$ and $\varphi_{s,\xi} = 0$ otherwise (\Cref{thm:makespan:full-information}).

For illustration, we provide examples for optimal signaling schemes for throughput and makespan optimization in Appendix~\ref{sec:throughput-example} and Appendix~\ref{sec:makespan-example}, respectively.

\paragraph{Related Work}

Optimal signaling for systems with congestion effects has been studied primarily in the static equilibrium model of Wardrop \cite[e.g.][]{Acemoglu18,Das17,NachbarX21,Massicot19,Vasserman15,Wu21}.
For the Wardrop model with affine costs, 
\citet{Bhaskar16} showed that it is $\mathsf{NP}$-hard to compute a public signaling scheme that approximates the total travel time better than a factor of $4/3-\eps$ for all $\eps >0$.
\citet{Griesbach22} proved that optimal information revelation is always optimal if and only if the underlying network is series-parallel and provided an algorithm computing the optimal public signal for parallel links when the number of scenarios and commodities is constant.
For atomic congestion games, \citet{Castiglioni21} studied information design, but considered a different model where players commit to following the signal before they receive it. \cite{Zhou22} computed public and private signals in singleton games with a constant number of resources.

The dynamic flow model that we use here dates back to \citet{Vickrey69}. It has been studied in more detail, e.g., by \citet{Koch11} and \citet{cominetti2015existence}.
\citeauthor{Koch11} also showed that the price of anarchy with respect to the throughput objective, i.e., the worst-case ratio of the throughput of an arbitrary dynamic flow and that of a dynamic equilibrium, is unbounded on general networks.
For the makespan objective, \citet{Bhaskar15} showed that the price of anarchy is $e/(e-1)$ when one is allowed to reduce the capacity of the links arbitrarily (but still compares to the optimal flow for the original capacities).
\citet{Correa22} showed that this bound on the price of anarchy also holds when the inflow rate at the source can be reduced. For parallel link networks (as considered in this work), they showed that the price of anarchy is $4/3$.
For both objectives, bounds on the price of anarchy are relevant for information design since price of anarchy bounds (if they exist) yield an approximation guarantee for the signaling scheme of full information revelation. 
The long-term behavior of dynamic equilibria in this model has been explored by \citet{Cominetti22,Olver21}.
Further variants of the model with multiple commodities, more complicated queuing behavior, or further side constraints have also been studied \citep[e.g.][]{Sering19,Graf20,Israel20,Sering18}, but they do not have any implications for mechanism or information design.
\citet{Oosterwijk022} study a dynamic model on parallel paths where players need to meet a certain time deadline and try to minimize the costs of the used links. They obtain tight bounds on the price of anarchy both for the makespan and the throughput objective.

\section{Preliminaries}

For an integer $m \in \mathbb{N}$, let $[m] \coloneqq \{1,\dots,m\}$.
For $x \in \mathbb{R}$, we denote by $[x]^+ \coloneqq \max\{x,0\}$ the positive part and by $[x]^- \coloneqq \min\{x,0\}$ the negative part of $x$.
We denote vectors and matrices with bold face and assume that all vectors are column vectors.
Further, $\vec e_i$ denotes the $i$-th unit vector, $\vec 1$ the all-ones vector, and $\vec 0$ the all-zeros vector (of the appropriate dimension).
We first introduce the dynamic equilibrium model with deterministic travel times and then turn to the dynamic equilibrium model with stochastic travel times.

\paragraph{Dynamic Equilibrium with Deterministic Travel Times}
Consider a set $[m]$ of links where each link $i \in [m]$ has a \emph{capacity} $\nu_i \in  \mathbb{Q}_{>0}$ and a \emph{travel time} $\tau_i \in \mathbb{Q}_{\geq 0}$.
There is continuum of flow particles arriving at the links with a constant rate of $\arrival \in \mathbb{Q}_{>0}$.
A \emph{flow} is a family of measurable functions $\vec f = (f_i)_{i \in [m]}$ with $f_i \colon \R_{\geq 0} \to \R_{\geq 0}$ satisfying
$\sum_{i \in [m]} f_i (\theta) = \arrival$ for almost all $\theta \geq 0$.
The value $f_i (\theta)$ describes the inflow into link~$i$ at time~$\theta \geq 0$.
Each link operates with the following queuing dynamics: if the inflow into a link $i \in [m]$ is higher than its capacity $\nu_i$, a queue builds up. Particles in the queue are processed with rate $\nu_i$. After passing the queue, it takes an additional amount of time~$\tau_i$ to traverse the link.
We denote by $z_i (\theta)$ the length of the queue at any given time $\theta \geq 0$. The queue dynamics is described via the differential equation
    \begin{equation}\label{eq:queue:dynamics}
        z_i' (\theta) =
        \begin{cases}
            f_i (\theta) - \nu_i &\text{if } z_i(\theta) > 0, \\
            [f_i (\theta) - \nu_i]^+ &\text{if } z_i(\theta) = 0.
        \end{cases}
    \end{equation}
A flow particle that enters link~$i$ at time~$\theta$ waits for time $z_i(\theta) / \nu_i$ in the queue and then experiences a travel time of $\tau_i$. Therefore, the exit time of a flow particle entering link~$i \in [m]$ at time $\theta \in \mathbb{R}_{\geq 0}$ is
$T_i (\theta) \coloneqq \theta + \mfrac{z_i(\theta)}{\nu_i} + \tau_i$.
A flow is a \emph{dynamic equilibrium} if almost all particles have no incentive to use another link, i.e., if
$T_i(\theta) = \min\nolimits_{j \in [m]}  T_{j} (\theta)$ for all $i \in [m]$ with $f_i(\theta) > 0$ for almost all $\theta \geq 0$.
We refer to the set of links with minimal exit time at a given point in time~$\theta$ as the \emph{support}~$S(\theta)$ of the flow and we denote it by
$S(\theta) = \bigl\{ i \in [m] \;\big\vert\; T_i(\theta) = \min_{j \in [m]} T_{j} (\theta) \bigr\}$.
We remark that in general the dynamic equilibrium may not be unique. However, the exit times as well as the supports are~\citep[cf.][]{cominetti2015existence,Olver21}.
For a given time horizon $T > 0$, the \emph{throughput} of a flow $\vec f$ is defined as
\begin{align}
\label{eq:throughput}
\FV_T(\vec f) \coloneqq
    \int_0^T \sum_{i \in [m]} \sinkInflow_i (\theta) \,\mathrm{d}\theta,
\end{align}
where $\sinkInflow_i (\theta)$ is the outflow of link~$i$ at time $\theta$ that can be computed as
\[
    f_i^- \big(\theta + \tau_i\big) =
    \begin{cases}
        \min \{f_i (\theta), \nu_i \} &\text{if } z_i(\theta) = 0, \\
        \nu_i &\text{if } z_i (\theta) > 0.
    \end{cases}
\]
For a given time horizon $T > 0$, the \emph{makespan} of a flow $\vec f$ is defined as 
\begin{align}
\label{eq:makespan}
M_T(\vec f) \coloneqq \sup \bigl\{ T_i(\theta) \;\big\vert\; \theta \in [0,T], i \in [m] \text{ with } f_i(\theta) > 0\bigr\}.
\end{align}

\paragraph{Bayesian Dynamic Equilibrium with Stochastic Travel Times}
Let $d \in \mathbb{N}$ with $d \geq 1$ and consider a set $[d]$ of scenarios. We assume that the travel time  $\tau_{i,s} \in \mathbb{Q}_{\geq 0}$ of each link $i \in [m]$ depends on the scenario~$s \in [d]$, and write $\vec \tau_i = (\tau_{i,1},\dots,\tau_{i,d})^\top$.
The capacities $\nu_i \in \mathbb{Q}_{>0}$ are independent of the scenario.
A vector $\vec{\mu} = (\mu_1, \dotsc, \mu_d)^\top \in [0,1]^d$ such that $\sum_{j \in [d]} \mu_{j} =1$ is called a \emph{belief}, and we denote by $\Delta$ the set of all beliefs. Every belief essentially induces a system with deterministic travel times by replacing the deterministic travel time $\tau_i$ with its expectation $\vec \mu^\top \vec \tau_i$.
With a slight overload of notation, we use the same notation for deterministic and stochastic travel times.
In particular, the expected exit time of a flow particle entering link $i \in [m]$ at time $\theta \in \mathbb{R}_{\geq 0}$ is $T_i(\theta) \coloneqq \theta + z_i(\theta)/\nu_i + \vec \mu^\top \vec \tau_i$. We call a flow a \emph{Bayesian dynamic equilibrium} if $T_i(\theta) = \min_{j \in [m]} T_j(\theta)$.
For a fixed time horizon $T > 0$, the throughput of a flow $\vec f$ in scenario~$s \in [d]$ is defined as $\smash{\FV_{T,s}(\vec f) \coloneqq \int_{0}^T \sum_{i \in [m]} f_{i,s}^-(\theta)\,\mathrm{d}\theta}$ where the outflow of link~$i$ at time~$\theta$ in scenario~$s$ is
\begin{align*}
f_{i,s}^-(\theta + \tau_{i,s}) =
\begin{cases}
\min\{f_i,(\theta), \nu_i\} & \text{ if $z_i(\theta) = 0$},\\
\nu_i & \text{ if $z_i(\theta) > 0$}.
\end{cases}
\end{align*}
The \emph{expected throughput} of a flow $\vec f$ (according to belief $\vec \mu$) is then given by $F_T(\vec f) \coloneqq \sum_{s \in [d]} \mu_s F_{T,s}(\vec f)$.
For a time horizon $T >0$ and a scenario~$s$, the makespan of a flow $\vec f$ in scenario~$s$ is defined as $M_{T,s}(\vec f) \coloneqq \sup \bigl\{ T_i(\theta) \;\vert\; \theta \in [0,T], i \in [m]\}$ and the expected makespan is then $M_T(\vec f) \coloneqq \sum_{s \in d} \mu_s M_{T,s}(\vec f)$.

\paragraph{Information Design}
We assume that the flow particles have a prior belief $\vec \lambda^*$.
As shown by \citet{Dughmi14}, there is a one-to-one corresponence between public signaling schemes and convex decompositions of the prior, and that---due to Caratheodory's Theorem---the convex decomposition requires at most $d$ beliefs.
For a fixed time horizon $T$, let us define $\FV \colon \Delta \to \mathbb{R}_{\geq 0}$ as the expected throughput of a Bayesian dynamic equilibrium for belief $\vec \mu$.
While the equilibrium for a given belief is not unique, the expected exit times are. Therefore, the expected throughput is unique as well and, hence, $\FV(\vec\mu)$ is well-defined. 
Then, in order to compute the public signaling scheme that maximizes the expected throughput, we are interested in solving
$$\sup \Biggl\{ \sum_{j \in [d]} \alpha_j\FV(\vec\mu^j) \;\Bigg\vert\; \vec \mu^1,\dots\vec \mu^d \in \Delta, \vec\alpha \in [0,1]^d \text{ with }  \vec 1^\top \vec \alpha =1 \text{ and } \sum_{j \in [d]} \alpha_j \vec \mu^j = \vprior \Biggr\}.$$
Similarly, for a fixed time horizon $T$, let us further define $\MS \colon \Delta \to \mathbb{R}_{\geq 0}$ as the expected makespan of a Bayesian dynamic equilibrium for belief $\vec \mu$.\footnote{Note, that if the Bayesian dynamic equilibrium for a belief~$\vec\mu$ is not unique, the makespan may also not be unique. In this case, we define $\MS (\vec\mu)$ as the worst case makespan of all equilibria.} Then, in order to compute the public signaling scheme that minimizes the expected makespan, we are interested in solving
\begin{align*}
\inf \Biggl\{ \sum_{j \in [d]} \alpha_j\MS(\vec\mu^j) \;\Bigg\vert\; \vec \mu^1,\dots\vec \mu^d \in \Delta, \vec\alpha \in [0,1]^d \text{ with }  \vec 1^\top \vec \alpha =1 \text{ and } \sum_{j \in [d]} \alpha_j \vec \mu^j = \vprior \Biggr\}.
\end{align*}

\section{Structural Results for Deterministic Travel Times}
\label{sec:structure-deterministic}
In this section, we construct a flow $\vec f$ and show that it is a dynamic equilibrium for deterministic travel times.
    In general, the dynamic equilibrium may not be unique. However, the exit times $T_i$ and, thus, also the supports~$S(\theta)$ are the same for any dynamic equilibrium \citep[cf.][]{cominetti2015existence,Olver21}.
Hence, it is sufficient to only consider the flow $\vec f$ to compute the throughput and makespan.

    We assume for the remainder of this section that the links are ordered by their travel times, i.e., that $\tau_1 \leq \dots \leq \tau_m$. We denote by $\bar{\nu} (i) \coloneqq \sum_{j \in [i]} \nu_j$ the sum of the first~$i$ capacities. In particular, we set $\bar{\nu} (0) \coloneqq 0$.
    We define points in time~$\theta^*_i$ for $i \in [m]$ recursively as $        \theta^*_1 = 0$ and
    \begin{equation} \label{eq:equilibrium:breakpoints}
        \theta_{i+1}^* =
        \begin{cases}
            \theta_{i}^* + \frac{\bar{\nu}(i)}{\arrival - \bar{\nu}(i)} (\tau_{i+1} - \tau_i) & \text{if } \bar{\nu}(i) < \arrival, \\
            \infty & \text{if } \bar{\nu}(i) \geq \arrival,
        \end{cases}
        \qquad \text{for } i = 1, \dotsc, m-1
        .
    \end{equation}
As we will show, the times~$\theta^*_i$ are the points in time when the link~$i$ is first used in the dynamic equilibrium.
Let $k \coloneqq \max \bigl\{j \;\big\vert\; \bar{\nu} (j) < \arrival \bigr\}$ be the maximum index of a link such that the total capacities up to that link are strictly less than the inflow.
    For every $i \in [m]$, we define an inflow function $f_i \colon \R_{\geq 0} \to \R_{\geq 0}$ by
    \begin{equation}\label{eq:equilibrium:flow}
        f_i (\theta) \coloneqq
        \begin{cases}
            \frac{\arrival}{\bar{\nu}(j)} \cdot \nu_i
                &\text{if } i \leq k \text{ and } \theta \in [\theta^*_j, \theta^*_{j+1}), j=i, \dotsc, k \\
            \nu_i
                &\text{if } i \leq k \text{ and } \theta \geq \theta^*_{k+1}, \\
            \arrival - \bar{\nu}_k
                &\text{if } i = k + 1 \text{ and } \theta \geq \theta^*_{k+1}, \\
            0 &
            \text{otherwise,}
        \end{cases}
    \end{equation}
    where $\theta^*_{m+1} = \infty$ for the case when $k = m$.
    Note that the inflow defined in~\eqref{eq:equilibrium:flow} is completely determined by the values~$\theta^*_i$, which, in turn, solely depend on the travel times~$\tau_i$.
    \Cref{lem:equilibrium} which is stated and proven in Appendix~\ref{app:structure-deterministic} shows that the inflow defined in \eqref{eq:equilibrium:flow} is indeed a dynamic equilibrium and gives further properties of the queues.

    \newcommand{\nLinksEQ}[1][\theta]{\ell(#1)}

\section{Structural Results for Stochastic Travel Times}
\label{sec:structure-stochastic}

In this section, we obtain structural results for the dynamic equilibrium $\vec f$ as a function of the belief $\vec \mu \in \Delta$. 
For a belief $\vec \mu \in \Delta$, let $\pi(\cdot\, ; \vec \mu) : [m] \to [m]$ be a permutation of the links that orders them non-decreasingly by their expected travel times (according to the belief), i.e.,  $\vec \mu^\top \vec \tau_{\pi(1; \vec \mu)} \leq \dots \leq \vec \mu^\top \vec \tau_{\pi(m;\vec \mu)}$. Analogously to the deterministic case \eqref{eq:equilibrium:breakpoints}, the entry times $\theta^*_{\pi(i; \vec \mu)}$ in dependence of $\vec{\mu}$ can be computed as $\theta_{\pi(1; \vec \mu)}^* (\vec{\mu}) = 0$ and
\begin{equation} \label{eq:equilibrium:breakpoints:mu}
        \theta_{\pi(i+1; \vec \mu)}^* (\vec{\mu}) =
        \begin{cases}
            \theta_{\pi(i; \vec \mu)}^* + \dfrac{\bar{\nu}(i; \vec \mu)}{\arrival - \bar{\nu}(i; \vec \mu)} \,\vec \mu^\top \bigl(\vec \tau_{\pi(i+1; \vec \mu)} - \vec \tau_{\pi(i; \vec \mu)}\bigr) & \text{if } \bar{\nu}(i; \vec \mu) < \arrival, \\
            \infty & \text{if } \bar{\nu}(i; \vec \mu) \geq \arrival,
        \end{cases}
    \end{equation}
for  all $i = 1, \dotsc, m-1$, where $\bar{\nu}(i;\vec \mu) = \sum_{j \in [i]} \nu_{\pi(j; \vec \mu)}$ is the total capacity of the $i$ links used first, depending on the belief~$\vec \mu$.

We are interested in partitioning the set of beliefs $\Delta$ into subsets such that the ordering of the links by expected travel time is fixed within each subset. A naive bound on the number of these sets is $m!$, but this bound is not sufficient for us since our algorithms will iterate over these sets. In order to obtain a better bound on the number of these sets, we will resort to the theory of \emph{hyperplane arrangements}.
To this end, for every pair of links $i,j \in [m]$ with $i < j$, we define by
    $        H_{i,j} \coloneqq \bigl\{ \vec{\mu} \in \R^d \;\big\vert\; \vec \mu^\top \vec{\tau}_i = \vec{\mu}^\top\vec{\tau}_j \bigr\}$
    the (possibly empty) hyperplane containing all $\vec{\mu}$ such that the expected travel times on links~$i$ and~$j$ are the same.
    Then, $\mathcal{H} \coloneqq \bigl\{ H_{i,j} \;\big\vert\; i,j \in [m] \text{ with } i < j\bigr\}$ is an arrangement of $|\mathcal{H}| = \mfrac{m(m-1)}{2}$ linear hyperplanes in $\mathbb{R}^d$ (where we allow that one or more of the hyperplanes are empty).
        The hyperplanes of the arrangement $\mathcal{H}$ partition $\Delta$ into a number of open regions whose closures are called the \emph{$(d-1)$-cells} of the arrangement, i.e., the $(d-1)$-cells are the closures of the maximal connected subsets of $\smash{\Delta \setminus \bigcup_{i,j \in [m], i < j} H_{i,j}}$. Every $(d-1)$-cell is a polyhedron in $\Delta$. For $k \in \{0,\dots,d-1\}$, a $k$-cell of the arrangement is a $k$-dimensional face of one of its $(d-1)$-cells.
        The following theorem of \cite{Buck43} bounds the number of $k$-cells of a hyperplane arrangement.

    \begin{theorem}[\cite{Buck43}]
    \label{thm:buck}
    For any hyperplane arrangement of $n$ hyperplanes in $\mathbb{R}^{d-1}$ and any $k \in \{0,\dots,d-1\}$, the number of $k$-cells is not larger than
    $\binom{n}{d-1-k} \sum_{i=0}^{k} \binom{n-(d-1)+k}{i}$.
    \end{theorem}
    Within each $(d-1)$-cell $P$ of $\mathcal{H}$, the ordering of the links $\pi(\cdot; \vec \mu)$ with respect to their expected travel times is constant and, therefore, every function $\vec{\mu} \mapsto \theta_{i}^* (\vec{\mu})$ is affine within $P$ by formula~\eqref{eq:equilibrium:breakpoints:mu}.
    From this construction, we obtain the following immediate result.
    \begin{lemma}
    \label{lem:theta-piece-wise-linear}
        For every link $i \in [m]$, the function
                    $\theta_i^* \colon \Delta \to \mathbb{R}_{\geq 0}$ is a continuous, piecewise linear function on~$\R^d$. In particular, the function is affine on every $(d-1)$-cell of $\mathcal{H}$.
    \end{lemma}
    
For a link~$i \in [m]$, $\theta_i^*(\vec \mu)$ is the first point in time at which flow enters link~$i$.
Given belief~$\vec{\mu}$ and $s \in [d]$, we define for every link~$i \in [m]$ its \emph{first exit time} as
    $\exitTime_{i, s} (\vec{\mu}) \coloneqq \theta^*_{i} (\vec{\mu}) + \tau_{i, s}$.
We obtain as an immediate corollary of  \Cref{lem:theta-piece-wise-linear} 
the following result.

\begin{corollary}
For every link~$i \in [m]$ and every scenario $s \in [d]$, the function $\exitTime_{i,s} \colon \Delta \to \mathbb{R}_{\geq 0}$ is a continuous, piecewise linear function on $\mathbb{R}^d$. In particular, the function is affine on every $(d-1)$-cell of $\mathcal{H}$.
\end{corollary}

We proceed to introduce another permutation of the links. For a given scenario $s \in [d]$, let $\sigma_s(\cdot\,;\vec \mu) : [m] \to [m]$ be a permutation of the links that orders them non-decreasingly with respect to their exit times in scenario~$s$, i.e., $\sigma_{s} (i; \vec \mu) \leq \sigma_{s} (j; \vec \mu)$ implies $\exitTime_{i, s}(\vec{\mu}) \leq \exitTime_{j, s}(\vec{\mu})$.
Since the inflow functions $f_i$ of the dynamic equilibrium $\vec f$ are piecewise constant (by \eqref{eq:equilibrium:flow}), so are the outflows $f^-_{i, s}$ in scenario~$s$. We define
    $\bar{\nu}_{s} (i; \vec \mu) \coloneqq
    \sum_{j \in [i]} \nu_{\sigma_{s} (j; \vec \mu)}$ for  $i \in [m] \cup\{0\}$
    and set $\sigma_{s} (0; \vec \mu) \coloneqq 0$, $\sigma_{s} (m+1; \vec \mu) \coloneqq m+1$ and $\exitTime_{0, s} (\vec{\mu}) \coloneqq 0$, $\exitTime_{m+1, s} (\vec{\mu}) \coloneqq \infty$.
For $s \in [d]$ and $\vec \mu \in \Delta$, let
    $$        \numContributingLinks \coloneqq
        \min \bigl\{
            \max \bigl\{ i \in [m] \;\vert\; \exitTime_{\sigma_{s}(i; \vec \mu), s} \leq T \bigr\},
            \min \bigl\{ i \in [m] \;\vert\; \bar{\nu}_s (i; \vec{\mu}) \geq u \bigr\}
        \bigr\}$$
    be the number of links that contribute to the throughput.
We proceed to compute explicit formulas for the out-flows and the throughput.

\begin{lemma}\label{lem:flow-value:function}
    For a fixed scenario~$s \in [d]$ and a given belief~$\vec{\mu} \in \Delta$, the total outflow at time~$\theta$ is
    $\sum_{j \in [m]} \sinkInflow_{j,s} (\theta) =
            \min\bigl\{\bar{\nu}_{s} (i; \vec \mu), \arrival \bigr\}$  whenever $\exitTime_{\sigma_{s}(i; \vec \mu), s} (\vec{\mu}) \leq \theta < \exitTime_{\sigma_{s}(i+1; \vec \mu), s} (\vec{\mu})$ for some $i \in [m] \cup \{0\}$.
    The throughput is given by the equation
    \[
        F_s (\vec{\mu}) =
        \arrival T \!+\!
        T
        \big[ \bar{\nu}_s (\numContributingLinks; \vec{\mu}) \!-\! \arrival \big]^- 
        \!\!+\! 
        \exitTime_{\sigma_s (\numContributingLinks; \vec{\mu}), s}(\vec \mu)
        \big[
            \bar{\nu}_s (\numContributingLinks; \vec{\mu}) - \arrival
        \big]^+
        \!-\!\!\! \sum_{i \in [\numContributingLinks]} \nu_{\sigma_s(i; \vec{\mu})} \exitTime_{\sigma_{s}(i; \vec \mu), s}(\vec \mu)
        .
    \]
\end{lemma}

\begin{proof}
    Let $s \in [d]$ and $\vec{\mu}$ be fixed. For ease of notation, let $\sigma \coloneqq \sigma_s(\,\cdot\, ; \vec{\mu})$ and $\pi \coloneqq \pi (\, \cdot \,, \vec{\mu}$) as well as $\exitTime_{\sigma(i)} \coloneqq \exitTime_{\sigma_s (i; \vec{\mu}) , s} (\vec{\mu})$, and $\bar{\nu} (i) \coloneqq \bar{\nu}_s (i; \vec{\mu})$, and $\numContributingLinks* \coloneqq \numContributingLinks$.

    Let $\vec{f}$ be the inflows of the equilibrium with respect to the given belief~$\vec{\mu}$, i.e., $\vec{f}$ is the inflow defined in~\eqref{eq:equilibrium:flow} with respect to the travel times $\vec{\mu}^{\top} \vec{\tau}_i$. The inflows are piecewise constant.
    Let $k \coloneqq \max \bigl\{j \;\big\vert\; \bar{\nu} (i) \leq \arrival \bigr\}$. Then, in particular the inflow into every link $i$ with $\pi(i) \leq k$, the inflow is either $0$ or greater equal than $\nu_i$ and the link $k$ has inflow $0$ or $\arrival - \bar{\nu} (k-1)$. Therefore the outflows of the links are either $0$ or $\nu_i$ (for $i$ with $\pi(i) \leq k$) and $0$ or $\arrival - \bar{\nu} (k-1)$ (for $i = k$). All other outflows are $0$. With the definition of the first exit times, we get
    \[
        \sinkInflow_{i, s} (\theta) =
        \begin{cases}
            \nu_i &\text{if } \theta \geq \exitTime (i) \text{ and } i < k \\
            \arrival - \bar{\nu} (k-1) &\text{if } \theta \geq \exitTime (i) \text{ and } i = k \\
            0 &\text{otherwise.}
        \end{cases}
    \]
    Summation over all links yields the formula for $\sum_{i=1}^m \sinkInflow_{i,s} (\theta)$ as stated.

    We proceed by computing the throughput $F_s (\vec{\mu})$ in scenario~$s$.
    To this end, let $L \coloneqq \{ i \in [m] \mid \exitTime_{\sigma_{s}(i)} \leq T \}$ be the number of links that have an exit time smaller than~$T$. With the value $\numContributingLinks* \coloneqq \numContributingLinks$, we can rewrite the throughput as
    \begin{align*}
        F_s (\vec{\mu}) &= \int_0^T \sum_{i=1}^m f^-_{i, s} (\theta) \,\mathrm{d} \theta
        = \sum_{i = 0}^{m} \int_{\min \{\exitTime_{\sigma (i)}, T\}}^{\min \{ \exitTime_{\sigma (i+1)}, T \}}  \min\bigl\{\bar{\nu} (i), \arrival \bigr\} \,\mathrm{d} \theta \\
        &=
        \sum_{i = 0}^{\numContributingLinks* - 1} \int_{\exitTime_{\sigma (i)}}^{\exitTime_{\sigma (i+1)}}  \bar{\nu} (i) \,\mathrm{d} \theta
        +
        \sum_{i = \numContributingLinks*}^{L-1} \int_{\exitTime_{\sigma (i)}}^{\exitTime_{\sigma (i+1)}} \arrival \,\mathrm{d} \theta
        +
        \int_{\exitTime_{\sigma (L)}}^{T} \min\bigl\{ \bar{\nu} (L), u \bigr\}  \,\mathrm{d} \theta \\
        &=
        \sum_{i = 0}^{\numContributingLinks* - 1} \bar{\nu} (i) (\exitTime_{\sigma (i + 1)} - \exitTime_{\sigma (i)})
        +
        \arrival \cdot \max\bigl\{ \exitTime_{\sigma (L)} - \exitTime_{\sigma (\numContributingLinks*)}, 0 \bigr\}
        +
        \min\bigl\{ \bar{\nu} (L), u \bigr\} \cdot (T - \exitTime_{\sigma (L)})
        \\
        &=
        \sum_{i = 0}^{\numContributingLinks* - 1}
        \bigl( \bar{\nu} (i-1) - \bar{\nu} (i) \bigr) \exitTime_{\sigma(i)}
        +
        \bar{\nu} (\numContributingLinks* - 1) \exitTime_{\sigma(\numContributingLinks*)}
        +
        \arrival \cdot \max\bigl\{ \exitTime_{\sigma (L)} - \exitTime_{\sigma (\numContributingLinks*)}, 0 \bigr\}  \\
        &\qquad
        +
        \min\bigl\{ \bar{\nu} (L), u \bigr\} \cdot (T - \exitTime_{\sigma (L)})
        \\
        &=
        - \sum_{i = 0}^{\numContributingLinks*-1}
        \nu_{\sigma(i)} \exitTime_{\sigma(L)}
        +
        \bar{\nu} (\numContributingLinks* - 1) \exitTime_{\sigma(\numContributingLinks*)}
        +
        \arrival \cdot \max\bigl\{\exitTime_{\sigma (L)} - \exitTime_{\sigma (\numContributingLinks*)}, 0 \bigr\}
        +
        \min\bigl\{ \bar{\nu} (L), \arrival \bigr\} \cdot (T - \exitTime_{\sigma (L)})
        ,
    \end{align*}
    where we used $\sigma(0) = 0$ and $\omega_0 = 0$.
    If $L = \numContributingLinks*$, then $\min\bigl\{ \bar{\nu} (L), \arrival \bigr\} = \bar{\nu} (L)$ and the throughput can be simplified to
    \[
        F_s (\vec{\mu}) = - \sum_{i = 0}^{\numContributingLinks*-1}
        \nu_{\sigma(i)} \exitTime_{\sigma(\numContributingLinks*)}
        +
        \big(\bar{\nu} (\numContributingLinks* - 1) - \bar{\nu} (\numContributingLinks*)\big) \exitTime_{\sigma(\numContributingLinks*)}
        + \bar{\nu}(\numContributingLinks*) T
        =
        - \sum_{i = 0}^{\numContributingLinks*}
        \nu_{\sigma(i)} \exitTime_{\sigma(\numContributingLinks*)}
        + \bar{\nu}(\numContributingLinks*) T
        .
    \]
    If $L > \numContributingLinks*$, then $\exitTime_{\sigma (L)} > \exitTime_{\sigma (\numContributingLinks*)}$ and $\min\bigl\{ \bar{\nu} (L), \arrival \bigr\} = \arrival$. In this case, the throughput can be simplified to
    \begin{align*}
        F_s (\vec{\mu}) &=
        - \sum_{i = 0}^{\numContributingLinks*-1}
        \nu_{\sigma(i)} \exitTime_{\sigma(\numContributingLinks*)}
        +
        (\bar{\nu} (\numContributingLinks* - 1) - \arrival ) \exitTime_{\sigma(\numContributingLinks*)} + \arrival T
        \\
        &=
        - \sum_{i = 0}^{\numContributingLinks*}
        \nu_{\sigma(i)} \exitTime_{\sigma(\numContributingLinks*)}
        +
        (\bar{\nu} (\numContributingLinks*) - \arrival) \exitTime_{\sigma(\numContributingLinks*)}
        + \arrival T
        . \qedhere
    \end{align*}
    \input{fig-flowvalue.tex}
\end{proof}

As long as the ordering $\sigma_{s} (\, \cdot \,; \vec{\mu})$ remains unchanged, $\exitTime_{\sigma_s(i; \vec{\mu}), s}$ is linear in $\vec\mu$ and the capacities $\nu_{\sigma_s (i; \vec{\mu})}$ are constant in $\vec \mu$.
As long as the same number of links has an exit time smaller or equal than~$T$, the number $\numContributingLinks$ is constant as well.
Therefore, in this case, $F_s (\vec{\mu})$ is linear in~$\vec{\mu}$.
We define new hyperplanes
\begin{align*}
    H_{i,j,s} &\coloneqq \bigl\{ \vec{\mu} \in \R^m \;\big\vert\; \exitTime_{i,s} (\vec{\mu}) = \exitTime_{j,s} (\vec{\mu}) \bigr\}
    &&\text{ for } s \in [d] \text{ and } i, j \in [m] \text{ with } i < j, \text{and} \\
    H_{i,s,T} &\coloneqq \bigl\{ \vec{\mu} \in \R^m \;\big\vert\; \exitTime_{i,s} (\vec{\mu}) = T \bigr\}
    &&\text{ for } s \in [d] \text{ and } i \in [m]
    .
\end{align*}

For every $(d-1)$-cell $P$ of the hyperplane arrangement $\mathcal{H}$, consider the hyperplane arrangement
$\mathcal{H}^* \coloneqq\bigl\{ H_{i,j,s} \;\big\vert\;  s \in [d], i,j \in [m] \text{ with } i < j\} \,\cup\,\bigl\{ H_{i,s,T} \;\big\vert\; s \in [d], i \in [m] \bigr\}$.
This hyperplane arrangement further subdivides every $(d-1)$-cell $P$ of $\mathcal{H}$. The following lemma gives the main structural insights for the behavior of the functions $\FV_s$ on the $(d-1)$-cells of the hyperplane arrangements $\mathcal{H}$ and $\mathcal{H}^*$.
While the proof of the first two properties follows from the definition of the hyperplane arrangement $\mathcal{H}^*$, the proof of convexity of the functions $\FV_s$ on the $(d-1)$-cells of $\mathcal{H}$ is highly non-trivial and uses the concrete form of the functions $\FV_s$ obtained in \Cref{lem:flow-value:function}. The proof of the lemma is quite technical and thus deferred to Appendix~\ref{app:lem:piecewise-linear}.

\begin{lemma}
\label{lem:piecewise-linear}
For every scenario~$s \in [d]$, the function $\FV_s \colon \Delta \to \mathbb{R}_{\geq 0}$ has the following properties.
    \begin{enumerate}[label=(\roman*)]
        \item \label{it:piecewise-linear-1} $\FV_s (\vec{\mu})$ is piecewise linear.
        \item \label{it:piecewise-linear-2} Let $P$ be a $(d-1)$-cell of $\mathcal{H}$ and $P^*$ be a $(d-1)$-cell of $\mathcal{H}^*$; then $\FV_s$ is affine linear on $P \cap P^*$.
        \item \label{it:piecewise-linear-3} Let $P$ be a $(d-1)$-cell of $\mathcal{H}$, then $\FV_s$ is convex on $P$.
    \end{enumerate}
\end{lemma}

\section{Additive PTAS for Throughput Maximization}
\label{sec:ptas-additive}

In this section, we give an additive PTAS for computing the optimal throughput achievable by a public signaling scheme.
For ease of notation, let us write $\vec M  \in [0,1]^{d \times d}$ for the matrix that has the vectors $\vec \mu^1, \dots \vec \mu^d$ as column vectors, and let us write $\mathcal{M}$ for the set of left-stochastic matrices whose rows sum to $1$.
The primal signaling problem is given by
\begin{align}
\opt \coloneqq \sup \Biggl\{ \sum_{j \in [d]} \alpha_j\FV(\vec\mu^j) \;\Bigg\vert\; \vec M \in \mathcal{M}, \vec\alpha \in [0,1]^d \text{ with }  \vec 1^\top \vec \alpha =1 \text{ and } \vec M \vec \alpha = \vprior \Biggr\},\label{eq:primal}\tag{$P$}
\end{align}
where $F$ is the expected throughput as a function of the prior.
The main result of this section is the following.

\begin{theorem}
\label{thm:ptas-additive}
For $d$ constant and every $\eps^* > 0$, there is a polynomial-time algorithm computing $p \in [\opt-\eps^*,\opt]$.
\end{theorem}

Instead of (approximately) solving the primal signaling problem \eqref{eq:primal} directly, our algorithm relies on the following Lagrangian dual. 

\begin{lemma}
\label{lem:lagrange-dual}
The dual signaling problem is
\begin{align}
d^*=&\;\Bigl\{ \vec w^\top \vprior \;\Big\vert\; \vec w \in \mathbb{R}^{\s} \text{ with } \vec w^\top \vec \mu \geq \FV(\vec \mu) \text{ for all } \vec \mu \in \Delta \Bigr\}.\label{eq:dual}\tag{$D$}
\end{align} 
In particular, weak duality holds.
\end{lemma}

\begin{proof}
The Lagrange function is defined as
\begin{align*}
L(\vec \alpha, \vec M,\vec w) \coloneqq \sum_{j=1}^d \alpha_j \FV(\vec \mu^j) + \vec w^\top \Biggl(\vec \lambda^* - \sum_{j=1}^d \alpha_j \vec \mu^j\Biggr).
\end{align*}
The dual function is then
\begin{align*}
q(\vec w) \coloneqq&\; \sup \bigl\{ L(\vec \alpha, \vec M, \vec \mu) \;\big\vert\; \vec M \in \mathcal{M}, \vec
\alpha \in [0,1]^d \text{ with } \vec 1^\top \vec\alpha =1 \bigr\} \\
=&\; \sup \Biggl\{ \sum_{j=1}^d \alpha_j \FV(\vec \mu^j) + \vec w^\top \Biggl(\vec \lambda^* - \sum_{j=1}^d \alpha_j \vec \mu^j\Biggr)  \;\Bigg\vert\; \vec M \in \mathcal{M}, \vec
\alpha \in [0,1]^d \text{ with } \vec 1^\top \vec\alpha = 1\Biggr\} \\
=&\; \sup \Biggl\{ \vec w^\top \vec \lambda^* + \sum_{j=1}^d \alpha_j \bigl(\FV(\vec \mu^j) - \vec w^\top \vec \mu^j\bigr)  \;\Bigg\vert\; \vec M \in \mathcal{M}, \vec
\alpha \in [0,1]^d \text{ with } \vec 1^\top \vec \alpha = 1\Biggr\}\\
=&\; \vec w^\top \vec \lambda^* + \sup \Biggl\{ \sum_{j=1}^d \alpha_j \bigl(\FV(\vec \mu^j) - \vec w^\top \vec \mu^j\bigr) \;\Bigg\vert\; \vec M \in \mathcal{M}, \vec
\alpha \in [0,1]^d \text{ with } \vec 1^\top \vec \alpha = 1 \Biggr\}\\
=&\; \vec w^\top \vec \lambda^* + \sup \bigl\{ \FV(\vec \mu) - \vec w^\top \vec \mu \;\big\vert\; \vec \mu \in \Delta \bigr\}.
\end{align*}
The dual signaling problem is
\begin{align}
d^* \coloneqq&\; \inf \Bigl\{ \vec w^\top \vec \lambda^* + \sup\bigl\{ F(\vec \mu) - \vec w^\top \vec \mu \;\big\vert\; \vec \mu \in \Delta\bigr\} \;\Big\vert\; \vec w \in \mathbb{R}^d\Bigr\}.\notag\\
\intertext{It is easy to convince ourselves that the objective is shift invariant in $\vec w$, i.e., for any $\delta \in \mathbb{R}$, the vectors $\vec w \in \mathbb{R}^d$ and $\vec w + \vec 1\delta$ yield the same objective.
To see this, note that adding a constant $\delta$ to each entry of $\vec w$, we have that $\vec w^\top \vec \lambda^*$ increases by $\delta$ since $\vec \lambda^*$ is a probability vector. On the other hand, $\sup\bigl\{F(\vec \mu) - \vec w^\top \vec \mu\;\big\vert\; \vec \mu \in \Delta\bigr\}$ decreases by $\delta$ since the supremum is taken over all probability vectors.
Thus, it is without loss of generality to assume that $\sup \bigl\{F(\vec \mu) - \vec w^\top \vec \mu\bigr\} = 0$. We then obtain that the dual problem is}
d^*=&\inf \;\Bigl\{ \vec w^\top \vprior \;\Big\vert\; \vec w \in \mathbb{R}^{\s} \text{ with } \vec w^\top \vec \mu \geq \FV(\vec \mu) \text{ for all } \vec \mu \in \Delta \Bigr\}.\notag
\end{align}
Weak duality follows from the general theory of Lagrange functions.
\end{proof}

In the following, we show that there is no duality gap, i.e., the optimal values for the primal and dual signaling problem are attained and coincide.
For the proof, we use the definition and properties of the concave envelope of $\FV$.

\begin{lemma}
\label{lem:no-duality-gap}
The optimal values of the primal signaling problem \eqref{eq:primal} and the dual signaling problem \eqref{eq:dual} are attained at finite values and coincide.
\end{lemma}

\begin{proof}
Let $\hat{\FV} : \Delta \to \mathbb{R}_{\geq 0}$ be the concave envelope defined as
\begin{align}
\hat{\FV}(\vec \mu) \coloneqq \inf \bigl\{ g(\vec \mu) \;\big\vert\; g \text{ is concave and } g \geq \FV \text{ over } \Delta\bigr\}.
\label{eq:concave-envelope}
\end{align}
Using Caratheodory's theorem, it can be shown that $\opt = \hat{\FV}(\vec \lambda^*)$; see \citet[Theorem~2.35]{Dacorogna08} for a proof.
Let further
\begin{align}
\hat{d}^* \coloneqq \inf \Bigl\{ \vec w^\top \vprior \;\Big\vert\; \vec w \in \mathbb{R}^{\s} \text{ with } \vec w^\top \vec \mu \geq \hat{\FV}(\vec \mu) \text{ for all } \vec \mu \in \Delta \Bigr\}.\label{eq:dual-hat}\tag{$\hat{D}$}
\end{align}
be the dual signaling problem where the function $\FV$ is replaced by its concave envelope $\hat{\FV}$.
We claim that $\hat{d}^* = d^*$. To this end, it suffices to show that $\vec w^\top \vec \mu \geq \FV(\vec \mu)$ for all $\vec \mu \in \Delta$ implies $\vec w^\top \vec \mu \geq \hat{\FV}(\vec \mu)$ for all $\vec \mu \in \Delta$. To see this, note that the function $\vec \mu \mapsto \vec w^\top \vec \mu$ is linear and hence concave and, thus, every such function with $\vec w^\top \vec \mu \geq \FV(\vec \mu)$ is considered in the infimum in the definition of the concave envelope \eqref{eq:concave-envelope}. This yields $\vec w^\top \vec \mu \geq \hat{\FV}(\vec \mu)$ for all $\vec \mu \in \Delta$ as claimed.
Since $\hat{\FV}$ is concave, the supergradient $\partial \hat{\FV}(\vec \lambda^*)$ at $\vec \lambda^*$ exists and is a non-empty and convex set. Every vector $\vec v \in \partial\hat{\FV}(\vec \lambda^*)$ has the properties that $\vec v^\top \vec \mu \geq \hat{\FV}(\mu)$ for all $\vec\mu \in \Delta$ and that $\vec v^\top \vec \lambda^* = \hat{\FV}(\vec \lambda^*)$. As a consequence, the vectors $\vec v \in \partial\hat{\FV}(\vec \lambda^*)$ are considered in the infimum in \eqref{eq:dual-hat} and, thus, $\hat{d}^* \leq \hat{\FV}(\vec \lambda^*)$. We have shown that
$\opt = \hat{\FV}(\lambda^*) \geq \hat{d}^* = d^*$. Since, by weak duality $\opt \leq d^*$, we have $\opt = d^*$ and the result follows.
\end{proof}

The general idea for the additive PTAS is to apply the Ellipsoid method on the dual signaling problem and to use the equivalence of optimization and separation.
To this end, note that the feasible region of the Lagrange dual is always convex, see, e.g., \cite[\S~5.2]{Boyd04}.
Formally,
for a set $K \subseteq \mathbb{R}^n$ and any $\eps > 0$, let
$S(K,\eps) \coloneqq \bigl\{\vec x \in \mathbb{R}^n \;\big\vert\; ||\vec x - \vec y||_2 \leq \eps \text{ for some } \vec y
\in K\bigr\}$
be the set of points that are within an $\eps$-distance to a point in $K$. Furthermore, let
$
S(K,-\eps) \coloneqq \bigl\{\vec x \in K \;\big\vert\; S(\{x\}, \eps) \subseteq K \bigr\}$
be the set of points that are in the $\eps$-interior of $K$.
The following definition of the \emph{weak optimization problem} is taken from \citet[Definition~2.1.10]{Groetschel88} and adapted such that it deals with minimization instead of maximization.

\begin{definition}[Weak optimization problem]
\label{def:weak-optimization}
Given a vector $\vec c \in \mathbb{Q}^n$ and a rational number $\eps > 0$, the \emph{weak optimization problem} is to compute either
\begin{enumerate}[label=(\roman*)]
\item a vector $\vec y \in \mathbb{Q}^n$ such that $\vec y \in S(K,\eps)$, and $\vec c^\top \vec y \leq \vec c^\top \vec x + \eps$ for all $\vec x \in S(K,-\eps)$, or
\item assert that $S(K,-\eps)$ is empty.
\end{enumerate}
\end{definition}

The following definition is taken from \citet[Definition~2.1.13]{Groetschel88}.

\begin{definition}[Weak separation problem]
\label{def:weak-separation}
Given a vector $\vec y \in \mathbb{Q}^n$ and a rational number $\delta > 0$, the weak separation problem is to either
\begin{enumerate}[label=(\roman*)]
    \item \label{it:weak-separation-1} assert that $\vec y \in S(K,\delta)$, or
    \item \label{it:weak-separation-2} to compute a vector $\vec c \in \mathbb{Q}^n$ with $||\vec c||_{\infty} = 1$ such that $\vec c^\top \vec x \leq \vec c^\top \vec y + \delta$ for every $\vec x \in S(K, -\delta)$.
\end{enumerate}
\end{definition}

Roughly speaking, the equivalence of optimization and separation implies that a polynomial algorithm for the weak separation problem yields a polynomial algorithm for the weak optimization problem; see \citet[Corollary~4.2.7]{Groetschel88} for a formal statement. Thus, in order to solve the weak optimization problem, it is enough to solve the weak separation problem.

The separation problem for the dual signaling problem is the following. Given a vector $\vec w \in \mathbb{R}^{\s}$, find $\vec \mu \in \Delta$ such that $\vec w^\top \vec \mu < \FV(\vec \mu)$, or decide that no such $\vec \mu \in \Delta$ exists.
We proceed to show that for the dual signaling problem we can even solve a stronger version of the weak separation problem where $\delta = 0$. To this end, we show the following lemma.
For its proof, we use \Cref{thm:buck} to show that the subdivision into sets $Q \cap Q^*$ where $Q$ is a $k$-cell of $\mathcal{H}$ and $Q^*$ is a $k^*$-cell of $\mathcal{H}$ for some $k,k^* \in \{0,\dots,d\}$ produces only a polynomial number of sets. This allows to iterate over these sets in polynomial time while checking whether there is a candidate for an extreme point of the function $\FV(\vec \mu) - \vec w^\top \vec\mu$ in the relative interior of $Q \cap Q^*$. Showing that this function is differentiable in $Q \cap Q^*$ then allows first order conditions that reduce to a linear system since the function is piece-wise quadratic by \Cref{lem:piecewise-linear}. 

\begin{lemma}
\label{lem:strong-separation}
Given $\vec w \in \mathbb{R}^n$, we can compute in polynomial time $\vec \mu \in \Delta$ such that $\vec w^\top \vec \mu < \FV(\vec \mu)$, or decide that no such $\vec \mu \in \Delta$ exists.
\end{lemma}

\begin{proof}
By \Cref{lem:piecewise-linear}\ref{it:piecewise-linear-2}, for every $s \in [d]$, the function $\FV_s$ is affine on $P \cap P^*$ where $P$ is a $(d-1)$-cell of $\mathcal{H}$ and $P^*$ is a $(d-1)$-cell of $\mathcal{H}^*$. Thus, the function
\begin{align*}
g(\vec \mu) \coloneqq \FV(\vec \mu) - \vec w^\top \vec\mu = \sum_{s \in [d]} \mu_s \FV_s(\vec \mu) - \vec w^\top \vec\mu
\end{align*}
is a quadratic function on $P \cap P^*$.

Let $\vec \mu^* \in \arg\max\{ g(\vec \mu) \mid \vec \mu \in \Delta\}$ be a belief where the global maximum of $g$ is attained.
There is a $(d-1)$-cell $P$ of $\mathcal{H}$ and a $(d-1)$-cell $P^*$ of $\mathcal{H}^*$ such that $\vec \mu^* \in P \cap P^*$.
Since $P$ and $P^*$ are polytopes, so is their intersection $P \cap P^*$, and the set of inequalities defining the polytope $P \cap P^*$ is the union of those defining $P$ and $P^*$, respectively.

This implies that either $\vec \mu^*$ is contained in a $0$-cell of $P$ or $P^*$ or there are $k, k^* \in \{1,\dots,d-1\}$, a $k$-cell $Q$ of $\mathcal{H}$ and a $k^*$-cell $Q^*$ of $\mathcal{H}^*$ such that $\vec \mu^*$ is contained in the relative interior of $Q \cap Q^*$.
Since $Q \cap Q^* \subseteq P \cap P^*$, the function $g(\vec \mu)$ is also a quadratic function on $Q \cap Q^*$.
We claim that $g$ is differentiable when restricted to $Q \cap Q^*$.
Since $Q \cap Q^* \subseteq P$ and $P$ is a $(d-1)$-cell of $\mathcal{H}$, we can assume that the ordering $\pi(\cdot; \vec \mu)$ of the links with respect to the times $\theta_i^*(\vec \mu)$ when link~$i$ is first used is constant on $Q \cap Q^*$.
In addition, since $Q \cap Q^* \subseteq P^*$ and $P^*$ is a $(d-1)$-cell of $\mathcal{H}$, for every scenario $s \in [d]$, we can assume that the ordering $\sigma_s(\cdot; \vec \mu)$ of the links with respect to the times $\omega_{i,s}(\vec \mu)$ when flow exits link~$i$ first in scenario $s$ is constant on $Q \cap Q^*$.
In turn, this implies that the value $\min \{i \in [m] \mid \bar{\nu}(i; \vec \mu) > \arrival\}$ of the number of links whose capacity is enough to support the arrival rate is constant on $Q \cap Q^*$.
With the same arguments, for every scenario~$s \in [d]$ there is a fixed set of links whose times $\omega_{i,s}(\vec\mu)$ is larger than the time horizon $T$. With the notation from \Cref{lem:flow-value:function}, we have that
the number of links that contribute to the throughput defined as
\begin{align*}
\numContributingLinks = \min \bigl\{ \max \{i \in [m] \mid \omega_{\sigma_s(i;\vec \mu,s)} \leq T\}, \min \{i \in [m] \mid \bar{\nu}_s(i; \vec \mu) > \arrival\}\bigr\}
\end{align*}
is constant on $Q \cap Q^*$.
We conclude that in the formula for $\FV_s(\vec \mu)$ given in \Cref{lem:flow-value:function}, the only two terms involving $\vec \mu$ that are not constant in $\vec \mu$ are $\omega_{\sigma_s(i;\vec \mu),s}(\vec \mu)$ and $\omega_{\sigma_s(\numContributingLinks;\vec \mu),s}(\vec \mu)$ which are linear in $\vec \mu$. Hence, $g$ is differentiable when restricted to $Q \cap Q^*$.
If $\vec \mu^*$ is in the relative interior of $Q \cap Q^*$, we have as a necessary condition that $\nabla g |_{Q \cap Q^*}(\vec \mu^*) = 0$ where $g |_{Q \cap Q^*}$ is the restriction of $g$ on $Q \cap Q^*$.

Our algorithm for the solution of the separation problem uses the reverse search algorithm \citep[Theorem~3.3]{Avis96} as a main building block. The algorithm allows to iterate over all $(d-1)$-cells of a hyperplane arrangement $\mathcal{A}$ in dimension $d-1$ in time polynomial in $d$, $|\mathcal{A}|$, the number of the $(d-1)$-cells of $\mathcal{A}$, and the time to solve a linear program with $d-1$ variables and $|\mathcal{A}|-1$ inequalities.
To also enumerate over the $k$-cells with $k \in \{1,\dots,d-2\}$, we argue as follows. Every $k$-cell lies in the intersection of $d-1-k$ hyperplanes. Hence for every choice of $d-1-k$ hyperplanes, we obtain a corresponding hyperplane arrangement of $|\mathcal{A}|-(d-1-k)$ hyperplanes in dimension $k$. With \Cref{thm:buck}, this arrangement has at most
\begin{align*}
\sum_{i=0}^k \binom{|\mathcal{A}|-d-1-k}{i}  \leq \bigl( |\mathcal{A}|-d-1-k \bigr)^{k+1}
\end{align*}
$k$-cells. This implies that the time needed to enumerate all $k$-cells is polynomial in
\begin{align*}
\binom{|\mathcal{A}|}{d-1-k} \bigl( |\mathcal{A}-d-1-k\bigr)^{k+1} \leq |\mathcal{A}|^{d-1-k} \bigl( |\mathcal{A}-d-1-k\bigr)^{k+1}.
\end{align*}
This is polynomial in the input size as long as $d$ is constant and $|\mathcal{A}|$ is bounded by a polynomial of the input size.
Using that $|\mathcal{H}| = \frac{m(m-1)}{2}$ and that $|\mathcal{H}^*| = \frac{dm(m-1)}{2} + dm$, we conclude that we can enumerate both over all $k$-cells $Q$ of $\mathcal{H}$ with $k \in \{0,\dots,d-1\}$ and over all $k^*$-cells $Q^*$ of $\mathcal{H}^*$ with $k^* \in \{0,\dots,d-1\}$ in polynomial time.

For every $Q \cap Q^*$ such that $Q$ is a $k$-cell of $\mathcal{H}$ and $Q^*$ is a $k^*$-cell of $\mathcal{H}^*$ for some $k,k^* \in \{0,\dots,d-1\}$, we compute a candidate $\vec \mu_{Q,Q^*} \in Q \cap Q^*$ for the global maximimum of $g$ on $Q \cap Q^*$.
If $Q \cap Q^* = \emptyset$, there is nothing to do. If $\dim(Q \cap Q^*) = 0$, we compute the unique point in $Q \cap Q^*$ by Gaussian elimination. Since the numbers involved in the description of the hyperplanes for the definition of $Q$ and $Q^*$ are inputs to our problem, this can be done in polynomial time.

If $\dim(Q \cap Q^*)$, we compute a solution $\vec \mu$ to the linear system $\nabla g |_{Q \cap Q^*} = 0$. The system is linear since $g$ is quadratic when restricted to $Q \cap Q^*$ and, hence, its gradient is linear. Again, a solution to this linear system can be computed with Gaussian elimination and, hence, runs in polynomial time. Since the set of all points where the gradient is zero is a affine subspace of $Q \cap Q^*$, we further see that all these points have the same value of $g$. Hence, it is enough to take an arbitrary such point as the candidate $\vec \mu_{Q, Q^*}$.

After having computed all candidates $\vec \mu_{Q,Q^*}$, we simply choose the one with the highest value of $g$. Consider a true global optimum $\vec \mu^*$ of $g$ on $\Delta$. If $\vec \mu^*$ is contained in the relative interior of $Q \cap Q^*$, we have that $\nabla g|_{Q \cap Q^*} \vec \mu^* = \nabla g|_{Q \cap Q^*} \vec \mu_{Q,Q^*} = 0$ and, hence, $g(\vec \mu^*) = \vec \mu_{Q,Q^*}$. If, on the other hand, $\dim(Q \cap Q^*)$, we have that $\vec \mu^* = \vec \mu_{Q,Q^*}$. In either case, we have computed a global maximum $\vec \mu^{**}$ of $g$ on $\Delta$.

If $g(\vec \mu^{**}) > 0$, we have that $\FV(\vec \mu^{**}) - \vec w^\top \vec \mu^{**} > 0$ and, hence, we have computed $\vec \mu^{**} \in \Delta$ with $\vec w^\top \vec \mu^{**} < \FV(\vec \mu^{**})$ as requested. If, on the other hand, $g(\vec \mu^{**}) \leq 0$, we have that $\FV(\vec \mu) - \vec w^\top \vec \mu \leq 0$ for all $\vec \mu \in \Delta$. We conclude that we solved the exact separation problem in polynomial time.
\end{proof}

As an immediate corollary, we obtain that we can also solve the weak separation problem for the dual signaling problem in polynomial time
\begin{corollary}
\label{cor:strong-separation}
For any $\delta > 0$, the weak separation problem for the dual signaling problem can be solved in polynomial time.
\end{corollary}

\begin{proof}
By \Cref{lem:strong-separation}, given $\vec w \in \mathbb{Q}^d$ we can compute in polynomial time a belief $\vec \mu^* \in \Delta$ such that $\vec w ^\top \vec \mu^* < \FV(\vec \mu^*)$ or decide that no such $\vec \mu^*$ exists.
If no such belief $\vec \mu^*$ exists, we clearly have $\vec w^\top \vec \mu \geq \FV(\vec \mu)$ for all $\vec \mu \in \Delta$ and, thus, $\vec w \in K$ where
\begin{align*}
K \coloneqq \bigl\{ \vec v \in \mathbb{R}^{d} \;\big\vert\;  \vec v^\top \vec \mu \geq \FV(\vec \mu) \text{ for all } \vec \mu \in \Delta \bigr\}.
\end{align*}
In particular, we have then also that $\vec w \in S(K,\delta)$ for all $\delta > 0$ such that \Cref{def:weak-separation} \eqref{it:weak-separation-1} is satisfied.

If, on the other hand, we computed $\vec \mu^* \in \Delta$ with $\vec w^\top \vec \mu^* < \FV(\vec \mu^*)$, we have that
\begin{align*}
(-\vec \mu^*)^\top \vec w > -\FV(\vec \mu^*)  \geq (-\vec \mu^*)^\top \vec v
\end{align*}
for all $\vec v \in K$. This implies in particular for all $\delta > 0$ that $(-\vec \mu^*)^\top \vec w + \delta \geq (-\vec \mu^*)^\top \vec v$ for all $\vec v \in S(K,-\delta)$. Since $\vec \mu^*\in \Delta$, we further have $||-\vec\mu^*||_{\infty} \leq 1$. Normalizing, we have that \Cref{def:weak-separation}~\eqref{it:weak-separation-2} is satisfied for $-\vec \mu^* / ||\vec\mu^*||_{\infty}$.
\end{proof}

We are now ready to prove the main theorem of this section (\Cref{thm:ptas-additive}). The general idea for the proof is to use \Cref{cor:strong-separation} and the Ellipsoid method. However, in order to do so, we have to show that we can fit the dual feasible region into a ball with polynomially bounded diameter.
To this end, we show an upper bound on $||\vec w||_{\infty}$ of the dual optimal vector $\vec w$ based on the supergradient of $\FV$. 
Further calculations involving the approximation error of solutions $\vec v \in S(K,-\eps)$ then yield the result.

\begin{proof}[Proof of \Cref{thm:ptas-additive}]
Consider the dual signaling problem
\begin{align*}
\inf \Bigl\{ \vec w^\top \vec \lambda^* \;\Big\vert\; \vec w \in \mathbb{R}^d \text{ with } \vec w^\top \vec \mu \geq \FV(\vec \mu) \text{ for all } \vec \mu \in \Delta \Bigr\}.
\end{align*}
Let $K = \{\vec w \in \mathbb{R}^d \text{ with } \vec w^\top \vec \mu \geq \FV(\vec \mu) \text{ for all } \vec \mu \in \Delta\}$ be the set of dual feasible solutions. It is straightforward to check that $K$ is convex.
Indeed, for arbitrary $\vec v, \vec w \in K$, we have that
\begin{align*}
\biggl(\frac{\vec v + \vec w}{2}\biggr)^{\!\top} \vec \mu  = \frac{1}{2} \bigl(\vec v^\top \vec \mu + \vec w^\top \vec \mu\bigr) \geq \FV(\vec \mu)
\end{align*}
for all $\vec \mu \in \Delta$.
By \citet[Corollary~4.2.7]{Groetschel88} there is an algorithm that solves the weak optimization problem for every \emph{circumscribed} convex body $K \subseteq \mathbb{R}^n$ (where $n$ is given) in  with a polynomially bounded number of calls to a weak separation oracle and polynomial overhead. Here, circumscribed means that we need to give as an additional input to the algorithm a value $R \in \mathbb{Q}_{>0}$ such that $K \subseteq \bigl\{\vec x \in \mathbb{R}^d \mid ||x||_2 \leq R \bigr\}$.
However, this is clearly impossible since $K$ is unbounded.
To resolve this issue, we define a bounded subset $\bar{K} \subset K$ such that $$\inf\{ \vec w^\top \vec \lambda^* \mid \vec w \in K\} = \inf \{\vec w^\top \vec \lambda^* \mid \vec w \in \bar{K}\}$$ and for which a bound on the radius can be given.
To this end, note that since $\FV(\vec \mu) \geq 0$, we have $\vec w \geq 0$ for every dual feasible $\vec w$.
Since, by Lemma~\ref{lem:no-duality-gap}, we have strong duality, the supremum in the primal signaling problem and the infimum in the dual signaling problem are attained.
This implies in particular that for an optimal solution $\vec w$ of the dual signaling problem we have $\vec w^\top \vec \mu^* = \FV(\vec \mu^*)$ for some $\vec \mu^* \in \Delta$.
This implies that $\FV$ is locally concave at $\vec \mu^*$ and, hence, the supergradient $\partial \FV(\vec \mu^*)$ exists.
Furthermore, using that $\vec w^\top \vec \mu \geq \FV(\vec \mu)$ for all $\vec \mu \in \Delta$, we further have that $\vec w \in \partial \FV(\vec \mu^*)$, i.e., $\vec w$ is contained in the supergradient of $\FV$ at $\vec \mu^*$. 
The norm of the elements in the subgradient can be bounded by the (directional) derivatives (see, e.g., \citet[Proposition 3.2.12]{lucchetti2006convexity}).
Using that $\FV$ is differentiable almost everywhere, for every $\vec{w} \in \partial \FV(\vec \mu^*)$ we can bound
\begin{align*}
||\vec{w}||_{\infty} &\leq \sup \bigl\{||\nabla \FV(\vec \mu)||_\infty \;\big\vert\; \vec\mu \in \Delta \text{ and }\nabla \FV(\vec \mu) \text{ exists} \bigr\} \\
&\leq \sup\Biggl\{ \frac{\partial}{\partial \mu_s} \FV(\vec \mu) = \sum_{t \in [d]} \frac{\partial}{\partial \mu_s} \FV_t(\vec \mu) \mu_t + \FV_s(\vec \mu) \;\Bigg\vert\; s \in [d], \vec\mu \in \Delta \text{ and }\nabla \FV(\vec \mu) \text{ exists}  \Biggr\} \\
&\leq \sup\Biggl\{ \sum_{t \in [d]} \frac{\partial}{\partial \mu_s} \FV_t(\vec \mu) \mu_t + \FV_s(\vec \mu) \;\Bigg\vert\; s \in [d], \vec\mu \in \Delta \text{ and }\nabla \FV(\vec \mu) \text{ exists} \Biggr\}.
\intertext{Using \Cref{lem:flow-value:function}, and setting $\nu^* \coloneqq \sum_{i \in [m]} \nu_i$, we obtain the bound}
||\vec{w}||_{\infty}&\leq d(m+1) \sup_{i \in [m], s\in [d]} \bigg\vert\bigg\vert\frac{\partial}{\partial \mu_s}\omega_{i,s}(\vec \mu) \bigg\vert\bigg\vert \nu^* \\
&= d(m+1) \sup_{i \in [m], s\in [d]} \bigg\vert\bigg\vert\frac{\partial}{\partial \mu_s}\theta^*_{i,s}(\vec \mu) \bigg\vert\bigg\vert \nu^*
\intertext{Let $\kappa \coloneqq \min_{S \subseteq [m] : u \neq \sum_{j \in S} \nu_j} \{u - \sum_{j \in S} \nu_j\}$. Then, we obtain}
||\vec{w}||_{\infty} &\leq d(m+1) \sup \{\tau_{i,s} \mid i \in [m], s\in [d]\}(\nu^*)^2/\kappa
\end{align*}
Hence, we can set $\bar{K} \coloneqq \{\vec w \in \mathbb{R}^d \mid ||w||_{\infty} \leq R\}$ with $R \coloneqq d(m+1) d(m+1) \sup \{\tau_{i,s} \mid i \in [m], s\in [d]\}\nu^*$. Since $\log R$ is polynomial in the input size of the problem, this does not increase the runtime of the Ellipsoid method.

Having solved the weak optimization problem, we obtain $\vec w \in \mathbb{Q}^d$ such that $\vec w \in S(K,\eps)$ and $\vec w^\top \vec \lambda^* \leq \vec v^\top \vec \lambda^* + \eps$ for all $\vec v \in S(K,-\eps)$. Since $\vec w \in S(K,\eps)$, there is a vector $\bar{\vec w} \in K$ such that $||\vec w- \bar{\vec w}||_2 \leq \eps$.
This implies in particular that $||\vec w - \bar{\vec w}||_{\infty} \leq \eps$.
We obtain
\begin{align*}
\opt \leq \bar{\vec w}^\top \vec \lambda^* \leq (\vec w + \mathbf{1}\eps)^\top \vec \lambda^* \leq \vec w^\top \vec \lambda^* + \eps.
\end{align*}
In addition, we have that
\begin{align*}
\vec w^\top \vec \lambda^* \leq \vec v^\top \vec \lambda^* + \eps
\end{align*}
for all $\vec v \in S(K,-\eps)$. Since $K \subseteq \{\vec v' \in \mathbb{R}^d \mid ||\vec v- \vec v'||_{\infty} \leq d \eps \text{ for some } v \in S(K,-\eps)\}$, this implies in particular that $\vec w^\top \vec \lambda^* \leq \vec v^\top \vec \lambda^* + (d+1)\eps$ for all $\vec v \in K$ and, hence, $\vec w^\top \vec \lambda^* \leq (d+1)\eps + \opt$. We have established
\begin{align*}
\opt - \eps \leq \vec w^\top \vec \lambda^* \leq (d+1)\eps + \opt.
\end{align*}
This implies that for $p := \vec w^\top \vec \lambda^* - \eps$, we have
\begin{align*}
\opt - (d+2)\eps \leq p   \leq \opt.  
\end{align*}
Choosing $\eps^* \coloneqq \eps/(d+2)$ yields then the claimed result.
\end{proof}

\section{Multiplicative FPTAS for Throughput Maximization}
\label{sec:ptas-multiplicative}

While the additive PTAS devised in the last section allows to approximate the optimal throughput up to an arbitrary constant, it does not yield the corresponding approximate optimal signals. In this section, we devise a multiplicative FPTAS that also yields the corresponding signals, i.e., we show the following theorem.

\begin{theorem}\label{thm:PTAS-throughput}
For constant $d$ and for any $\eps^*>0$, there exists a fully polynomial-time algorithm for computing a signaling scheme, such that for its induced throughput $\alg$, it holds that
    $\alg \geq (1-\eps^*)\opt$.
\end{theorem}

As this result trivially holds, whenever $\opt=0$, we assume for the remaining part of this section that $\opt>0$.
The following lemma shows that this assumption is in fact equivalent to the smallest travel time
being strictly smaller than the time horizon $T$. 

\begin{lemma}\label{lem:positive-throughput}
    The following are equivalent:
    \begin{enumerate}[label=(\roman*)]
        \item $\opt>0$, and
        \item $T - \min \{ \tau_{i,s} \mid i \in [m], s\in [\s] \}>0.$
    \end{enumerate}
\end{lemma}

\begin{proof}
    We first show how $(i)$ implies $(ii)$.
    By the strong duality of the primal and dual signaling problem (\Cref{lem:no-duality-gap}), the supremum in \eqref{eq:primal} is attained.
    Let $\{\check{\vec\mu}^1,\dots,\check{\vec\mu}^d\}$ be the beliefs and $\{\check{\alpha}_1,\dots,\check{\alpha}_d\}$ the coefficients for which the optimum in \eqref{eq:primal} is attained.
    This yields 
    \begin{align*}
        \opt = \sum_{j\in[\s]} \check{\alpha}_j\FV(\check{\vec\mu}^j)
        =\sum_{j\in[\s]} \check{\alpha}_j \sum_{s\in[\s]}\check{\mu}_s^j\FV_s(\check{\vec\mu}^j).
    \end{align*}
    With $\opt>0$, there exists a scenario $s$ and an optimal belief $\check{\vec\mu}^j$, such that $\FV_s(\check{\vec\mu}^j)>0$.
    Using \Cref{lem:flow-value:function}, there exists a link $i^*$, such that the exit time $\exitTime_{i^*,s}$ is strictly less than the time horizon $T$, i.e., $T-\exitTime_{i^*,s}>0$.
    With the definitions of the exit times, it thus follows that
    \begin{align*}
        0
        &<T-\exitTime_{i^*,s}\\
        &= T- (\theta^*_{i^*} (\check{\vec\mu}^j) + \tau_{i^*, s})\\
        &\leq T- \tau_{i^*, s}\\
        &\leq T- \min \{ \tau_{i,s} \mid i \in [m], s\in [\s] \}.
    \end{align*}

    Next, we show how $(ii)$ implies $(i)$.
    For this, let $i^*\in [m]$ and $s^*\in [\s]$ be such that $\tau_{i^*,s^*}=\min \{ \tau_{i,s} \mid i \in [m], s\in [\s] \}$. Then we assume that $T-\tau_{i^*,s^*}>0$.
    Full information revelation as a signaling scheme induces throughput 
    \begin{align*}
        \FV^{\text{FI}}(\vprior)=\sum_{s\in [\s]}\prior_s\FV(\vec e_s),
    \end{align*}
    where $\vec e_s$ is the unit vector with entry $1$ for scenario $s$ and $0$ else.
    The throughput $\FV(\vec e_{s^*})$ for the realized scenario $s^*$ can be bounded from below by computing the throughput for the case in which the entire inflow rate $\arrival$ uses link $i^*$.
    This yields
    \begin{align*}
        \FV(\vec e_{s^*})\geq (T-\tau_{i^*,s^*})\cdot\min\{\nu_{i^*},\arrival\}>0.
    \end{align*}
    As the optimal signaling scheme induces a throughput no lower than full information revelation, we obtain
    \begin{align*}
        \opt(\vprior)
        \geq \FV^{\text{FI}}(\vprior)
        =\sum_{s\in [\s]}\prior_s\FV(\vec e_s)
        \geq \prior_{s^*}  (T-\tau_{i^*,s^*})\cdot\min\{\nu_{i^*},\arrival\}>0,
    \end{align*}
    and the result follows.
\end{proof}

The next lemma will be used to define the algorithm and bound its running time. For the proof we bound $\opt$ in terms of the throughput achieved by full information revelation.

\begin{lemma}\label{lem:delta-opt}
    Let $\opt>0$. For any $\eps>0$ and any $\delta >0$, there exists a $\kappa\in \N$ such that
    $(1-\eps)^\kappa\s T \arrival \leq \delta \opt$ and $\kappa$ is polynomially bounded in the input size of the given instance.
\end{lemma}

\begin{proof}
    Let $i^*\in [m]$ and $s^*\in [\s]$ be such that $\tau_{i^*,s^*}=\min\{\tau_{i,s} \mid i\in [m], s\in [\s] \}$.
    With $\opt>0$, we apply \Cref{lem:positive-throughput}, which yields $(T-\tau_{i^*,s^*})>0$.
    Analogously to the proof of \Cref{lem:positive-throughput}, we obtain a lower bound on $\opt$ by
    \begin{align*}
        \opt(\vprior)
        \geq \FV^{\text{FI}}(\vprior)
        =\sum_{s\in [\s]}\prior_s\FV(\vec e_s)
        \geq \prior_{s^*}  (T-\tau_{i^*,s^*})\cdot\min\{\nu_{i^*},\arrival\}>0.
    \end{align*}
    Hence, for proving the lemma, it is sufficient to show that there exists an $\kappa\in \N$ that is polynomially bounded by the input size, such that
    \begin{align*}
        (1-\eps)^\kappa\s T \arrival
        &\leq \delta \prior_{s^*}  (T-\tau_{i^*,s^*})\cdot\min\{\nu_{i^*},\arrival\}.
    \end{align*}
    Thus, by setting $\kappa\in\N$ such that
    \begin{align}
    \label{eq:M}
        \kappa\geq
        \left\lceil \frac{\log (\s T \arrival) - \log (\delta \prior_{s^*}  (T-\tau_{i^*,s^*})\cdot\min\{\nu_{i^*},\arrival\})}{\log (1-\eps)}\right\rceil
    \end{align}
    the claim holds. Finally, note that with the assumption of $\opt>0$ and \Cref{lem:positive-throughput}, we have $\delta \prior_{s^*}  (T-\tau_{i^*,s^*})\cdot\min\{\nu_{i^*},\arrival\}>0$, and thus, definition \eqref{eq:M} is indeed well-defined.
\end{proof}

Let $i^*\in [m]$ and $s^*\in[\s]$ be such that $\tau_{i^*,s^*}=\min\{\tau_{i,s} \mid i\in [m], s\in [\s] \}$.
Let $\eps>0$ and $\delta>0$ be two arbitrary but fixed values. Further, let $\kappa$ be as in \Cref{lem:delta-opt}.
We proceed to define an algorithm towards proving \Cref{thm:PTAS-throughput}.
The main building block of the algorithm is to find a piece-wise convex underestimator $\FV_{\eps,\kappa} \colon \Delta \to \mathbb{R}_{\geq 0}$ of the function $\FV$.
To this end, we define the following function $h_{\eps,\kappa}:\R \rightarrow \R$ that rounds numbers to the next power of $(1-\eps)$, or to $0$ if the number is too small:
\begin{align*}
    h_{\eps,\kappa} (x)
    \coloneqq
    \begin{cases}
        0 & \text{if } x<(1-\eps)^{\kappa-1},\\
        \max\{ (1-\eps)^{k-1} \mid (1-\eps)^{k-1}\leq x, k\in[\kappa] \} & \text{else.}
    \end{cases}
\end{align*}
Based on $h_{\eps,\kappa}$, we define an under-estimator function of the throughput $F_{\eps,\kappa}:\R^\s \rightarrow \R$ by
\begin{align}\label{eq:F-eps-M}
    F_{\eps,\kappa}(\vec\mu)
    \coloneqq \sum_{s\in [\s]} h_{\eps,\kappa}(\mu_s) F_s(\vec\mu).
\end{align}
In order to define the regions where the under-estimator is convex, we proceed to discretize $\Delta$ by a non-uniform $\eps$-net. For this, we define for every $s\in [\s]$ and every $j\in [\kappa]$ the hyperplanes
$
        L_{s,j}\coloneqq
        \{\vec\mu \mid \mu_s=(1-\eps)^{j-1}\} $ and $
        L_{s,0}\coloneqq
        \{\vec\mu \mid \mu_s=0\}.
$
We denote by $\mathcal{L}$ the union of the arrangement of hyperplanes $\mathcal{H}$ and the $\eps$-net, i.e., $
    \mathcal{L}\coloneqq
     \mathcal{H} \cup \bigl\{L_{s,j} \;\big\vert\; s\in[d],j\in[\kappa]\cup \{0\} \bigr\}.
$
The set $\mathcal{L}$ again defines an arrangement of hyperplanes in $\Delta$. We write $\Delta_{\eps}$ for the set of $0$-cells that are determined by $\mathcal{L}$, i.e., the set of points in $\Delta$ in which $d-1$ many pairwise distinct hyperplanes of $\mathcal{L}$ intersect. More formally,
$
    \Delta_{\eps} \coloneqq
         \big\{\vec\mu\in \Delta \;\big\vert\; \{\vec\mu\}=\bigcap\nolimits_{i\in[d-1]}A_{i} \text{ with } A_i\in \mathcal{L} \text{ and } A_i\neq A_j \text{ for } i\neq j\big\}$.
Note, that the number of $0$-cells $\ell\coloneqq \vert\Delta_{\eps}\vert$ is polynomially bounded by the input size as
\begin{align*}
    \ell\coloneqq \vert\Delta_{\eps}\vert
    \leq\binom{\vert\mathcal{L}\vert}{d-1}
    =\binom{\frac{m(m-1)}{2}+d\kappa}{d-1}
    \leq\bigg(\frac{m(m-1)}{2}+d\kappa\bigg)^{\!d}.
\end{align*}
In the following, we write $\Delta_\eps =\{\tilde{\vec\mu}^1,\dots,\tilde{\vec\mu}^\ell\}$.
The algorithm solves the linear program
\begin{align}\label{eq:alg}
    \alg\coloneqq \max \Biggl\{ \sum_{j\in[\ell]} \alpha_j\FV_{\eps,\kappa}(\tilde{\vec\mu}^j) \;\Bigg\vert\;
     \alpha_1,\dots,\alpha_{\ell} \in [0,1] \text{ with }  \sum_{j\in[\ell]} \alpha_j \tilde{\vec\mu}^j = \vec\vprior \Biggr\}
\end{align}
and returns a signaling scheme that induces the conditional beliefs that appear with strictly positive probability in \eqref{eq:alg}. Having computed the conditional beliefs $\{\vec\mu^1,\dots\vec\mu^\ell\}$ and the corresponding coefficients $\{\alpha_1,\dots\alpha_\ell\}$, the signaling scheme can be recovered in polynomial time \citep[cf.][]{Dughmi14}.
The next lemma states an important property of the under-estimator function $F_{\eps,\kappa}$. The proof uses that on each $(d-1)$-cell of $\mathcal{L}$, the function is a linear combination of convex functions.

\begin{lemma}\label{lem:F-eps-M-convex}
    For all $\eps>0$ and $\kappa\in \N$, the function $F_{\eps,\kappa}$ is convex on every $(d-1)$-cell of $\mathcal{L}$.
\end{lemma}

\begin{proof}
    Let $P$ be a $(d-1)$-cell of $\mathcal{L}$.
    For every $s\in [\s]$, the function $F_s$ is convex on $P$ by \Cref{lem:piecewise-linear}~\ref{it:piecewise-linear-3}, as $\mathcal{H}$ is a sub-arrangement of $\mathcal{L}$. By the definition of $F_{\eps,\kappa}$ in \eqref{eq:F-eps-M}, $F_{\eps,\kappa}$ is a linear combination of convex functions on $P$ with non-negative coefficients and hence, the claim follows.
\end{proof}

Next, we bound the under-estimator function $F_{\eps,\kappa}$ both from below and above. The proof crucially relies on \Cref{lem:delta-opt}.

\begin{lemma}\label{lem:F-eps-M-bounds}
    For all $\eps>0, \kappa\in \N,$ and $\mu \in \Delta$, we have $(1-\eps) \FV (\mu)-d(1-\eps)^{\kappa} \FV_{\max} \leq F_{\eps,\kappa}(\mu) \leq \FV(\mu)$, where $\FV_{\max}=\max_{s\in[\s]}\sup_{\mu\in\Delta}F_s(\mu)$.
\end{lemma}

\begin{proof}
    The upper bound on $F_{\eps,\kappa}$ is immediate, as 
    \begin{align}\label{eq:h-eps-M}
        F_{\eps,\kappa}(\vec \mu)
        & = \sum_{s\in [\s]} h_{\eps,\kappa}(\mu_s) F_s(\vec \mu)
        \leq \sum_{s\in [\s]} \mu_s F_s(\vec\mu)
        =F(\vec \mu),
    \end{align}
    where the inequality follows from the definition of $h_{\eps,\kappa}$.
    Thus, we focus on the lower bound on $F_{\eps,\kappa}$.
    First, we note that
    \begin{align*}
    h_{\eps,\kappa}(x)\leq x \leq h_{\eps,\kappa}(x)/(1-\eps)
\end{align*}
    holds whenever $h_{\eps,\kappa}(x)\neq 0$.
    Hence, for some $\mu\in\Delta$ we define
    \begin{align*}
        I(\vec\mu) & \coloneqq \{s\in[\s] \mid h_{\eps,\kappa}(\mu_s)>0\} \text{ and}\\
        E(\vec\mu) & \coloneqq \{s\in[\s] \mid h_{\eps,\kappa}(\mu_s)=0\}.
    \end{align*}
    Together with \eqref{eq:h-eps-M}, this yields
    \begin{align*}
        F_{\eps,\kappa}(\vec\mu)
        &= \sum_{s\in [\s]} h_{\eps,\kappa}(\mu_s) F_s(\vec\mu)\\
        &= \sum_{s\in I(\vec\mu)} h_{\eps,\kappa}(\mu_s) F_s(\vec\mu) + \sum_{s\in E(\vec\mu)} h_{\eps,\kappa}(\mu_s) F_s(\vec\mu) \\
        &\geq \sum_{s\in I(\vec\mu)} (1-\eps)\mu_s F_s(\vec\mu) \\
        &= \sum_{s\in [\s]} (1-\eps)\mu_s F_s(\vec\mu) - \sum_{s\in E(\vec\mu)} (1-\eps)\mu_s F_s(\vec\mu)\\
        &\geq (1-\eps)F(\vec\mu)  - (1-\eps)\sum_{s\in E(\vec\mu)}\mu_s F_s(\vec\mu).
    \end{align*}
    It remains to show that $\sum_{s\in E(\vec\mu)}\mu_s F_s(\vec\mu) \leq d(1-\eps)^{\kappa-1} \FV_{\max}$.
    Recall, that for all $s\in E(\vec\mu)$ we have $h_{\eps,\kappa}(\mu_s)=0$ and hence, $\mu_s<(1-\eps)^{\kappa-1}$.
    Thus, we obtain
    \begin{align*}
        \sum_{s\in E(\vec\mu)}\mu_s F_s(\vec\mu)
        &\leq (1-\eps)^{\kappa-1} \sum_{s\in E(\vec\mu)} F_s(\vec\mu) \\
        &\leq (1-\eps)^{\kappa-1} \sum_{s\in E(\vec\mu)} \max_{s\in[\s]}\sup_{\vec\mu\in\Delta}F_s(\vec\mu) \\
        &\leq (1-\eps)^{\kappa-1} d F_{\max},
    \end{align*}
    where for the final inequality we used $\vert E(\vec\mu) \vert \leq \s$.
\end{proof}

With these lemmas at hand, we are now in position to give the proof of the main result of this section (\Cref{thm:PTAS-throughput}). For the proof, we use that optimizing over the piece-wise convex under-estimator $\FV_{\eps,\kappa}$ instead of the original function $\FV$ causes only a multiplicative error of $(1-\eps)$ in the interior and an additional additive error close to the boundary of $\Delta$. In addition, the under-estimator is piece-wise convex and, in particular, convex on every $(d-1)$-cell of $\mathcal{L}$. As a consequence, every optimal convex decomposition of the prior for the convex under-estimator only uses the $0$-cells of $\mathcal{L}$. Bounding their number by a polynomial of the encoding length of the input, we then obtain the result. 

\begin{proof}[Proof of \Cref{thm:PTAS-throughput}]
    We show in fact that the algorithm defined in \eqref{eq:alg} guarantees the desired approximation ratio and runs in time polynomial in the input size.
    We start by proving the approximation guarantee of $(1-\eps)(1-\delta)$.
    Using the definition of $\opt$ and \Cref{lem:F-eps-M-bounds}, we obtain
    \begin{align*}
        \opt
        &= \max \Biggl\{ \sum_{j\in[\s]} \alpha_j\FV(\vec\mu^j) \;\Bigg\vert\; \vec\mu^1,\dots,\vec\mu^{\s}\in \Delta, \alpha_1,\dots,\alpha_{\s} \in [0,1] \text{ with }  \sum_{j\in[\s]} \alpha_j \vec\mu^j = \vec\vprior \Biggr\}\\
        & \leq \frac{1}{1-\eps} \max \Biggl\{ \sum_{j\in[\s]} \alpha_j\FV_{\eps,\kappa}(\vec\mu^j) \;\Bigg\vert\; \vec\mu^1,\dots,\vec\mu^{\s}\in \Delta, \alpha_1,\dots,\alpha_{\s} \in [0,1] \text{ with }  \sum_{j\in[\s]} \alpha_j \vec\mu^j = \vec\vprior \Biggr\} \\
        &\quad + d(1-\eps)^\kappa F_{\max}.\\
        \intertext{
        By \Cref{lem:F-eps-M-convex}, $F_{\eps,\kappa}$ is convex on each $(d-1)$-cell of $\mathcal{L}$. Since $\Delta_\eps$ contains all $0$-cells of $\mathcal{L}$, we can write
        }
        \opt
        & \leq \frac{1}{1-\eps} \max \Biggl\{ \sum_{j\in[\ell]} \alpha_j\FV_{\eps,\kappa}(\tilde{\vec\mu}^j) \;\Bigg\vert\; \tilde{\vec\mu}^1,\dots,\tilde{\vec\mu}^{\ell}\in \Delta_\eps, \alpha_1,\dots,\alpha_{\ell} \in [0,1] \text{ with }  \sum_{j\in[\ell]} \alpha_j \tilde{\vec\mu}^j = \vec\vprior \Biggr\} \\
        &\quad + d(1-\eps)^\kappa F_{\max}.
        \intertext{
        With the definition of $\alg$ in \eqref{eq:alg}, we obtain
        }
        \opt
        &= \frac{1}{1-\eps} \alg + d(1-\eps)^\kappa F_{\max}.
        \intertext{
        To prove the approximation guarantee, it remains to give an upper bound $d(1-\eps)^\kappa F_{\max}$ by $\delta\opt$.
        For this, note that the maximal throughput $\FV_{\max}$ over all scenarios $s\in[\s]$ and all possible priors $\mu\in\Delta$ can be no larger than the inflow rate $\arrival$ times the time horizon $T$. This yields
        }
        \opt
        &\leq \frac{1}{1-\eps}\alg + d(1-\eps)^\kappa T \arrival.
        \intertext{
        Applying \Cref{lem:delta-opt} gives
        }
        \opt
        & \leq \frac{1}{1-\eps}\alg + \delta \opt,
    \end{align*}
    which is equivalent to
    \begin{align*}
        (1-\eps)(1-\delta)\opt \leq \alg.
    \end{align*}
    By setting $\eps\coloneqq\eps^*/2$ and $\delta\coloneqq\eps^*/2$, we obtain the desired approximation guarantee, as 
    \begin{align*}
        (1-\eps)(1-\delta)=(1-\eps^*/2)^2\geq (1-\eps^*).
    \end{align*}
    It remains to show that $\alg$ runs in polynomial time.
    However, this immediately follows from the definition of $\alg$ and the fact that both $\kappa$ and $\ell$ are polynomially bounded by the input size.
\end{proof}

\section{Full Information Revelation for Makespan Minimization}\label{sec:makespan}

In this section, we consider the makespan objective.
As the main result of this section, we show that full information revelation is optimal, i.e., it is optimal to choose $\Xi = [d]$ and have $\varphi_{s,\xi} = \lambda^*_s$ whenever $s=\xi$ and $\varphi_{s,\xi} = 0$ otherwise.

\begin{theorem}
\label{thm:makespan:full-information}
With respect to the makespan objective, full information revelation is an optimal signaling scheme.
\end{theorem}

To prove the result, we will consider dynamic equilibria with different \emph{deterministic} travel times.
We first note that among all dynamic equilibria the total delay on every link is unique \citep[e.g.][]{cominetti2015existence,Olver21} and, hence, the queue lengths on every link are unique as well.
Thus, every dynamic equilibrium has the explicit queue lengths computed in \Cref{lem:equilibrium} (deferred to Appendix~\ref{app:structure-deterministic}) which are non-decreasing in~$\theta$. Therefore, a particle entering the system at time $\theta=T$ has the longest expected delay. This implies the following, simpler formula for the makespan.
    \begin{lemma}\label{lem:easier:makespan}
        Let $\vec f$ be the dynamic equilibrium from~\eqref{eq:equilibrium:flow}. Then,
        $
            \MS_T (\vec f) = \max \bigl\{ T_i (T) \;\big\vert\; i \in S(T) \bigr\}
        $.
    \end{lemma}

    We proceed to investigate how changing the perceived travel times  influences the makespan.
    Formally, let $\vec{\tau} = (\tau_i)_{i \in [m]}$ be the travel times of a given instance.
    Then, we denote by $\vec{f}^{\vec{\tau}}$ the dynamic equilibrium as defined in~\eqref{eq:equilibrium:flow}.
    Additionally, we consider another vector of (arbitrary) travel times $\vec{\tau}' = (\tau'_i)_{i \in [m]}$ with $\tau'_i \geq 0$.
    (The travel times $\vec{\tau}'$ could, for instance, be the travel times expected by the particles given a certain belief~$\vec{\mu}$.)
    Assume that particles behave according to $\vec{\tau}'$ rather than $\vec{\tau}$. Then, the dynamic equilibrium and, thus, the makespan changes if $\vec{\tau}'$ is changed.
    We denote by
    \[
        \mathcal{F}(\vec{\tau}') \coloneqq \big\{ \vec{f} \;\big\vert\; \vec{f} \text{ is dynamic equilibrium with respect to the travel times } \vec{\tau}' \big\}
    \]
    all dynamic equilibria with respect to the travel times~$\vec{\tau}'$ and emphasize that these equilibria are not unique. If the flow  $\vec{f}' \in \mathcal{F}(\vec{\tau}')$ emerges in the original instance (i.e., the instance with travel times~$\vec{\tau}$) the flow particles experience the exit times
    $
       \smash{T_i (\theta; \vec{f}') = \theta + \mfrac{z_i(\theta; \vec{f}')}{\nu_i} + \tau_i}$,
    where $z_i(\theta; \vec{f}')$ are the queue lengths determined by the dynamic equilibrium~$\vec{f}'$ . 
    Then, we denote by
    \[
        \overline{\MS} (\vec{\tau}') \coloneqq
        \sup \bigl\{ M_T (\vec{f}') \;\big\vert\; \vec{f}' \in \mathcal{F}(\vec{\tau}')\bigr\}
        =
        \sup_{\vec{f}' \in \mathcal{F}(\vec{\tau}')} \sup \bigl\{ T_i(\theta; \vec{f}') \;\big\vert\;  \theta \in [0, T], i \in [m] \text{ with } f'_i (\theta) > 0\bigr\}
    \]
    the (worst-case) makespan for given travel times~$\vec{\tau}'$.
    In this setting, we are interested in finding the best possible $\vec{\tau}' \geq \vec{0}$ that minimzes the makespan, i.e., we want to compute
    \begin{align} \label{eq:makespan:minimization:problem}
        \inf \bigl\{ \overline{\MS} (\vec{\tau}') \;\big\vert\; \vec{\tau}' \geq \vec{0} \bigr\}
        .
    \end{align}
    The function $\overline{\MS} (\vec{\tau}')$ is in general not continuous (see, for example, the instance in \Cref{sec:makespan-example}). Thus, it is not clear if a minimum is actually attained.

    Given a travel time vector~$\vec{\tau}' \geq \vec{0}$, we denote by $\vec{f}^{\vec{\tau}'} \in \mathcal{F}(\vec{\tau}')$ the dynamic equilibrium defined in~\eqref{eq:equilibrium:flow} for $\tau_i = \tau'_i$.
    Note, that in order to apply formulas~\eqref{eq:equilibrium:breakpoints} and~\eqref{eq:equilibrium:flow} we need to order the links with respect to $\tau'_i$. In case of a tie (i.e., if $\tau'_i = \tau'_j$) we order the links as they were ordered originally (i.e., with respect to the travel times~$\tau_i$).
    The latter ensures that the flow $\smash{\vec{f}^{\vec{\tau}'}}$
    only uses links with smallest original travel time $\tau_i$ whenever the equilibrium is not unique and particles are indifferent between links.
    Therefore,
    $\underline{\MS} (\vec{\tau}') \coloneqq \MS_T (\vec{f}^{\vec{\tau}'})  \leq \overline{\MS} (\vec{\tau}')$.
    This tie-breaking rule in favor of the equilibrium yielding the smaller makespan ensures that the function~$\underline{\MS}$ is lower semi-continuous and is key for the proof of the following lemma. 
    \begin{lemma} \label{lem:makespan:lower-semicontinuous}
        The function $\vec{\tau}' \mapsto \underline{\MS} (\vec{\tau}')$ is lower semi-continuous.
    \end{lemma}

     \begin{proof}
        Since $\vec{f}^{\vec{\tau}'}$ the dynamic equilibrium from~\eqref{eq:equilibrium:flow} for $\tau_i = \tau'_i$, we can apply \Cref{lem:easier:makespan} and obtain
        \begin{equation}\label{eq:prf:lower-semicontinuous:1}
            \underline{\MS} (\vec{\tau}') = \MS_T (\vec f^{\vec{\tau}'}) = \max \bigl\{ T_i (T; \vec{f}^{\vec{\tau}'}) \;\big\vert\; i \in S(T; \vec{f}^{\vec{\tau}'}) \bigr\},
        \end{equation}
        where $S(T; \vec{f}')$ is the support of the flow $\vec{f}'$ at time~$T$.
        The exit times $T_i(T; \vec{f}^{\vec{\tau}'})$ depend continuously on the queue lengths $z_i (T, \vec{f}^{\vec{\tau}'})$ and, thus, also continuously on $\vec{\tau}'$. A discontinuity can only appear if the support $S(T; \vec{f}^{\vec{\tau}'})$ changes.

        We want to prove lower semi-continuity, i.e., we want to show that for all $y < \underline{\MS}(\vec{\tau}')$ there exists a neighborhood $U \ni \vec{\tau}'$ such that $\underline{\MS}(\hat{\vec{\tau}}) > y$ for all $\hat{\vec{\tau}} \in U$.
        Let $\vec{\tau}'$ and $y < \underline{\MS}(\vec{\tau}')$ be fixed.
        Let $\tilde{U} \ni \vec{\tau}'$ be a neighborhood of $\vec{\tau}'$ such that
        \[
            |T_i (T; \vec{f}^{\vec{\tau}'}) - T_i (T; \vec{f}^{\hat{\vec{\tau}}})| < \underline{\MS}(\vec{\tau}') - y
            \quad \text{for all } i \in [m] \text{ and } \hat{\vec{\tau}} \in \tilde{U}
            .
        \]
        The neighborhood $\tilde{U}$ exists since the exit times $T_i(T; \vec{f}^{\vec{\tau}'})$ are continuous.

        If, additionally, there exists a neighborhood around $\hat{U} \ni \vec{\tau}'$ such that the support is the same for all $\vec{f}^{\hat{\vec{\tau}}} \ni \hat{U}$, then $\underline{M} (\hat{\vec{\tau}}) > y$ for all $\hat{\vec{\tau}} \in U \coloneqq \hat{U} \cap \tilde{U}$.

        If no neighborhood exists where the support stays the same, then this is only possible if particles are indifferent between two (or more links), i.e., if the exit times with respect to $\vec{\tau}'$ are exactly the same. Further, these links must have empty queues, since queues are continuous and links can only leave the support if the queue is empty. Hence, the links in question have the same travel times~$\tau'_i$. Denote by $L$ the set of all links with the same travel times. By definition, the flow $\vec{f}^{\vec{\tau}'}$ only uses one link without queue and ties between links with the same travel time~$\tau'_i$ are broken in favor of smaller $\tau_i$ (and thus smaller $T_i (T; \vec{f}^{\vec{\tau}'})$). Let $i^* \in L$ be the link, that is in the support of $\vec{f}^{\vec{\tau}'}$. If $i^*$ is not a maximizer of~\eqref{eq:prf:lower-semicontinuous:1}, then change of support has no effect on $\underline{M} (\vec{\tau}')$ in a small neighborhood. Otherwise, we have $\underline{M} (\vec{\tau}') = T_{i^*} (T; \vec{f}^{\vec{\tau}'})$ and therefore
        \[
            \underline{\MS} (\hat{\vec{\tau}}) \geq T_j(T; \vec{f}^{\hat{\vec{\tau}}}) \geq \tau_j = \tau_i = T_{i^*} (T; \vec{f}^{\vec{\tau}'}) = \underline{M} (\vec{\tau}')
        \]
        for all $\hat{\vec{\tau}} \in U$, proving the lower semi-continuity.
\end{proof}

    Since $\underline{\MS} (\vec{\tau}') \geq 0$, \Cref{lem:makespan:lower-semicontinuous} ensures that the minimum of $\vec{\tau}' \mapsto \underline{\MS} (\vec{\tau}')$ is attained.
    In the following lemma, we prove that the minimizer has the property that the exit times of all used links are the same. It can be shown that if this is not the case, the makespan can be improved, contradicting that the flow is a minimizer in the first place. 

    \begin{lemma} \label{lem:makespan:minimizer}
        Let $\tilde{\vec{\tau}} \geq \vec{0}$ be a vector of travel times such that
        $
            \underline{\MS} (\tilde{\vec{\tau}}) =
            \min_{\vec{\tau}' \geq \vec{0}} \underline{\MS} (\vec{\tau}')
        $.
        Denote by $S_{\vec{f}^{\tilde{\vec{\tau}}}} (T)$ the support of the flow~$\vec{f}^{\tilde{\vec{\tau}}}$ at time~$T$.
        Then,
        either $\underline{\MS} (\tilde{\vec{\tau}}) = \underline{\MS} (\vec{\tau})$ or the exit times at time~$T$ with respect to the flow~$\tilde{\vec{\tau}}$ are the same on all links $i \in S_{\vec{f}^{\tilde{\vec{\tau}}}} (T)$ in the support, i.e.,
        \[
            T_i (T; \vec{f}^{\tilde{\vec{\tau}}})
            =
            T_j (T; \vec{f}^{\tilde{\vec{\tau}}})
            \quad \text{ for all } i,j \in S_{\vec{f}^{\tilde{\vec{\tau}}}} (T),
        \]
        where $\smash{T_i (\theta; \vec{f}^{\tilde{\vec{\tau}}}) = \theta + \mfrac{z_i(\theta; \vec{f}^{\tilde{\vec{\tau}}})}{\nu_i} + \tau_i}$ are the exit times on the links with, if the particles act according to the dynamic equilibrium~$\vec{f}^{\tilde{\vec{\tau}}}$ with respect to the travel times~$\vec{\tau}'$ but the travel times are measured with respect to the original travel times~$\vec{\tau}$.
    \end{lemma}

    \begin{proof}
    Assume that there is $i, j \in S_{\vec{f}^{\tilde{\vec{\tau}}}} (T)$ such that $T_i (T; \vec{f}^{\tilde{\vec{\tau}}}) > T_j (T; \vec{f}^{\tilde{\vec{\tau}}})$. By \Cref{lem:equilibrium}\emph{(iii)}, we know that either the queue length~$z_i(T)$ is strictly decreasing in~$\tilde{\tau}_i$ or $i > k$. In the latter case, we get
    \[
        \underline{\MS} (\tilde{\vec{\tau}}) =
        T_i (T; \vec{f}^{\tilde{\vec{\tau}}}) \geq T + \tau_i \geq T + \tau_{k+1} =
        \underline{\MS} (\vec{\tau})
        .
    \]
    In the former case, we $z_i(T)$ can be decreased by increasing $\tilde{\tau}_i$ by a small~$\eps > 0$. Iterating this ensures that we are in a situation, where there is some $\delta > 0$ such that
    \[
        T_i (T; \vec{f}^{\tilde{\vec{\tau}}}) > T_j (T; \vec{f}^{\tilde{\vec{\tau}}}) + \delta
    \]
    for all $j \in S_{\vec{f}^{\tilde{\vec{\tau}}}} (T)$ with $j \neq i$. Therefore, in particular $\underline{\MS} (\tilde{\vec{\tau}}) = T_i (T; \vec{f}^{\tilde{\vec{\tau}}})$.
    Since the queues depend continuously on $\tilde{\vec{\tau}}$ as can be seen from the formulas in~\Cref{lem:equilibrium}, there exists $\eps > 0$ such that increasing $\tilde{\tau}_i$ by $\eps$ decreases $T_i (T; \vec{f}^{\tilde{\vec{\tau}}})$ by less than $\frac{\delta}{2}$ and increases $T_j (T; \vec{f}^{\tilde{\vec{\tau}}})$ by less than $\frac{\delta}{2}$. 
    Therefore, also $\underline{\MS} (\tilde{\vec{\tau}}) = T_i (T; \vec{f}^{\tilde{\vec{\tau}}})$ is decreased, contradicting the assumption that $\underline{\MS} (\tilde{\vec{\tau}}) =
    \min_{\vec{\tau}' \geq \vec{0}} \underline{\MS} (\vec{\tau}')$.
\end{proof}
    
    By \Cref{lem:makespan:minimizer}, all exit times in the equilibrium with respect to~$\tilde{\vec\tau}$ are equal. This can be used to show that the difference $\tilde{\tau}_{i+1} - \tilde{\tau}_i= \tau_{i+1} - \tau_i$ is equal for all $i \in [m]$. This implies that $\overline{\MS}(\vec\tau) = \overline{\MS}(\tilde{\vec{\tau}})$, leading to the following main result. 

    \begin{theorem} \label{thm:makespan:minimal}
        The infimum in~\eqref{eq:makespan:minimization:problem} is attained and, further,
        \[
            \MS_T (\vec{f}) = \underline{\MS}(\vec{\tau}) = \overline{\MS} (\vec{\tau}) = \min \bigl\{ \overline{\MS} (\vec{\tau}') \;\big\vert\; \vec{\tau}' \geq \vec{0} \bigr\}
            ,
        \]
        where $\vec{f}$ is the dynamic equilibrium with respect to the original travel times~$\vec\tau$.
    \end{theorem}

    \begin{proof}
    First, we observe that for $\vec{\tau}' = \vec{\tau}$, the makespan is independent of the actual equilibrium, i.e.,
    \[
        \MS_T (\vec{f}) = \MS_T (\vec{f}^{\vec\tau})
        \quad
        \text{for all } \vec{f} \in \mathcal{F}(\vec\tau).
    \]
    This is true since there are only multiple dynamic equilibria in $\mathcal{F}(\vec\tau)$ if players are indifferent between links. But in this case, the travel times~$\tau_i$ of these links are also the same and, hence, also the makespan is the same for all flows. This implies in particular $\underline{\MS}(\vec{\tau}) = \overline{\MS} (\vec{\tau})$.

    Let $\tilde{\vec{\tau}} \geq \vec{0}$ such that $\underline{\MS} (\tilde{\vec{\tau}}) =
    \min_{\vec{\tau}' \geq \vec{0}} \underline{\MS} (\vec{\tau}')$. By \Cref{lem:makespan:minimizer}, we know that either $\underline{\MS} (\tilde{\vec{\tau}}) = \underline{\MS} (\vec{\tau}) = \min_{\vec{\tau}' \geq \vec{0}} \, \overline{\MS} (\vec{\tau}')$ or all exit times of links in the support are the same. In the former case, there is nothing to show.
    Thus, assume
    \[
        T_i (T; \vec{f}^{\tilde{\vec{\tau}}})
        =
        T_j (T; \vec{f}^{\tilde{\vec{\tau}}})
        \quad \text{ for all } i,j \in S_{\vec{f}^{\tilde{\vec{\tau}}}} (T)
        .
    \]
    The definition of the exit times $T_j(T; \vec{f}^{\tilde{\vec{\tau}}}) = T + \mfrac{z_j(\theta)}{\nu_j} + \tau_i$ (where the queue lengths~$z_j(\theta)$ depend on $\vec{f}^{\tilde{\vec{\tau}}}$ and, thus, also on $\tilde{\vec{\tau}}$) yields 
    \begin{align*}
        0 &=
        T_i(T; \vec{f}^{\tilde{\vec{\tau}}}) - T_j(T; \vec{f}^{\tilde{\vec{\tau}}})
        =
        (\tau_i - \tau_j) + \frac{z_i(\theta)}{\nu_i} - \frac{z_j(\theta)}{\nu_i}
        \\
        &= (\tau_i - \tau_j) + \big(- \tilde{\tau}_i  - (- \tilde{\tau}_j) \big),
    \end{align*}
    for all $i,j \in S_{\vec{f}^{\tilde{\vec{\tau}}}} (T)$,
    where we used that every term in the formulas of~$z_j(\theta)$ from \Cref{lem:equilibrium} are independent of $\tau_i$ except for the term $\nu_j \tau_j$.
    Hence, we have shown that
    \[
        \tau_i - \tau_j = \tilde{\tau}_i  - \tilde{\tau}_j
    \]
    for all links $i,j \in S_{\vec{f}^{\tilde{\vec{\tau}}}} (T)$.
    We argue that the queue lengths only depend on the differences of travel times. If $\theta \geq \theta^*_{k+1}$ this follows immediately from~\Cref{lem:equilibrium}. If $\theta \leq \theta^*_{k+1}$, we use~\eqref{eq:prf:equilibrium:zi1} from the proof of~\Cref{lem:equilibrium} and~\eqref{eq:equilibrium:breakpoints} and obtain
    \begin{align*}
        z_i &\stackrel{\mathclap{\eqref{eq:prf:equilibrium:zi1}}}{=}
        \nu_i \Biggl(
            \tau_{\nLinksEQ} - \tau_i +
            \frac{\arrival - \bar{\nu} (\nLinksEQ)}{\bar{\nu} (\nLinksEQ)}
            \big(\theta - \theta_{\nLinksEQ}^* \big)
        \Biggr)
        \\
        &\stackrel{\mathclap{\eqref{eq:equilibrium:breakpoints}}}{=}
        \nu_i \Biggl(
            \tau_{\nLinksEQ} - \tau_i +
            \frac{\arrival - \bar{\nu} (\nLinksEQ)}{\bar{\nu} (\nLinksEQ)}
            \bigg(\theta - \sum_{j = 1}^{\nLinksEQ - 1} \frac{\bar{\nu}(j)}{\arrival - \bar{\nu}(j)} (\tau_{j+1} - \tau_j) \bigg)
        \Biggr).
    \end{align*}
    The latter formula also only depends on the differences of the travel times. Therefore, the queue lengths of the flows $\vec{f}$ and $\vec{f}^{\tilde{\vec{\tau}}}$ have the same makespan, i.e., $\MS_T (\vec{f}) = \MS_T (\vec{f}^{\tilde{\vec{\tau}}})$.

    Overall, we conclude
    \[
        \min_{\vec{\tau}' \geq \vec{0}} \, \underline{\MS} (\vec{\tau}') \leq \underline{\MS} (\vec{\tau}) = \MS_T (\vec{f}) = \MS_T (\vec{f}^{\tilde{\vec{\tau}}}) = \min_{\vec{\tau}' \geq \vec{0}} \, \underline{\MS} (\vec{\tau}')
    \]
    and
    \[
        \inf_{\vec{\tau}' \geq \vec{0}} \, \overline{\MS} (\vec{\tau}') \leq \overline{\MS} (\vec{\tau}') = \underline{\MS} (\vec{\tau}') \leq  \min_{\vec{\tau}' \geq \vec{0}} \, \underline{\MS} (\vec{\tau}') \leq \inf_{\vec{\tau}' \geq \vec{0}} \, \overline{\MS} (\vec{\tau}')
        .
    \]
    Hence, we have shown that
    $
        \inf_{\vec{\tau}' \geq \vec{0}} \, \overline{\MS} (\vec{\tau}') = \overline{\MS} (\vec{\tau}')
    $
    and the claim follows.
\end{proof}

   \Cref{thm:makespan:minimal} implies that, if the flow particles act according to a dynamic equilibrium with respect to arbitrary travel times~$\vec{\tau}'$, the makespan of this dynamic equilibrium can never be smaller than that of a dynamic equilibrium with respect with respect to the original travel times~$\vec{\tau}$. This, however, implies full information revelation is an optimal signaling scheme, as claimed in  \Cref{thm:makespan:full-information}. 

    \begin{proof}[Proof of Theorem~\ref{thm:makespan:full-information}]
        Assume, a scenario~$s$ with travel times~$\tau_{i, s}, i \in [m]$ is realized.
        If we use full information revelation signaling scheme, then flow particles act as in a dynamic equilibrium with respect to the actual travel times~$\tau_{i, s}, i \in [m]$.
        Given an arbitrary signaling scheme, the particles act as in a dynamic equilibrium with respect to travel times $\vec{\mu}^{\top}\vec{\tau}$, where $\vec{\mu}$ is the updated belief of the players.
        \Cref{thm:makespan:minimal} implies that the makespan of the dynamic equilibrium with respect to the original travel times~$\vec{\tau}$ is better than the makespan of any other flow that is a dynamic equilibrium for arbitrary travel times~$\vec{\tau}' \geq \vec{0}$. Hence, full information revelation is optimal for every scenario~$s$.
        Observe that this result also holds for an infinite number of scenarios and arbitrary distributions.
    \end{proof}

\newpage

\appendix

\section*{Appendix}

\section{Examples}

\subsection{An Example for the Throughput Objective}
\label{sec:throughput-example}

We consider an example for the throughput objective. The instance consists of $m=2$ links, $\s=2$ scenarios, and a time horizon~$T=5$. 
The links have capacities $\nu_1 = \frac{1}{3}$ and $\nu_2 = \frac{2}{3}$ and the travel times depend on two scenarios (called blue and red) with $\vec\tau_1 = (1, 5)^\top$ and $\vec\tau = (4,3)^\top$. The instance is depicted in \Cref{subfig:throughput:instance}.

The throughput function $\FV$ is shown in \Cref{subfig:throughput:function}. It is a continuous piece-wise quadratic function over the set of beliefs $\mu\in[0,1]$, where $\mu$ indicates the probability of the red scenario being realized. The function $\FV$ has two breakpoints at $\mu\in\{\frac{1}{5},\frac{3}{5}$\}. For $\mu\in[0,\frac{1}{5})$, link $1$ has a lower expected travel time and, thus, the first flow particles only use link $1$. Since the inflow rate $\arrival$ exceeds the capacity $\nu_1$, a queue starts to build up and at some point in time, flow particles start to deviate to link $2$.
With increasing $\mu$, flow particles start to deviate earlier in time to link $2$ as the expected travel time of link $2$ decreases in $\mu$ while the travel time of link $1$ increases in $\mu$.
For $\mu=\frac{1}{5}$, it is guaranteed, that the first particle, that deviated to link $2$ will leave the queuing system before time horizon $T$ no matter which scenario is realized. Hence, for $\mu\in[\frac{1}{5},\frac{3}{5}]$, the throughput increases linearly in $\mu$. For $\mu=\frac{3}{5}$, both links have the same expected travel time and, thus, both links are used from the very first point in time. This results in the maximization of the total throughput $\FV$. For $\mu\in(\frac{3}{5},1]$, the first flow particles use only link $2$ until a sufficiently large queue has built up, such that flow particles deviate to link $1$.

Additionally, \Cref{subfig:throughput:function} depicts two convex decompositions of a prior~$\lambda^*$ with the beliefs~$\vec{\mu}^i$. The corresponding two different signaling schemes are an optimal signaling scheme with three signals (green) and the suboptimal full information revelation scheme. 
The dashed lines give the throughout that can be achieved if the respective signaling scheme is used.

\input{fig_throughput_example.tex}

\subsection{An Example for the Makespan Objective}
\label{sec:makespan-example}

We consider an example for the makespan objective. The instance consists of $m=3$ links, $\s=2$ scenarios, and a time horizon~$T=\frac{1}{2}$. 
The links have capacities $\nu_1 = \frac{1}{2}$, $\nu_2 = \frac{1}{3}$, and $\nu_3 = \frac{1}{2}$. The travel times depend on two scenarios (called blue and red) with $\vec\tau_1 = (0, 5)^\top$, $\vec\tau_2 = (1,1)^\top$, and $\vec\tau = (4,0)^\top$. The instance is depicted in \Cref{subfig:makespan:instance}.

The makespan function $\MS$ is shown in \Cref{subfig:makespan:function}. It is a piece-wise quadratic function over the set of beliefs $\mu\in[0,1]$, where $\mu$ indicates the probability of the red scenario being realized. The function $\MS$ has six breakpoints at $\mu\in\{\frac{1}{10},\frac{1}{5},\frac{2}{5},\frac{1}{2},\frac{3}{4},\frac{7}{8}\}$.
    Roughly speaking, $\MS$ is continuous at a breakpoint, if the order in which links are chosen by the flow particles change. In contrast to this, $\MS$ is not continuous at a breakpoint, if the last flow particle entering the queuing system is indifferent between two links and thus the makespan cannot be uniquely defined. In fact, the makespan depends on which link the last particle chooses.
    In more detail, for $\mu\in[0,\frac{1}{10})$, all flow particles use link~$1$.
    Since the inflow rate $\arrival=1$ into link $1$ exceeds its capacity, a queue builds up and, hence, $\MS$ grows linearly in~$\mu$.
    When $\mu=\frac{1}{10}$, the last particle entering the queuing system is indifferent between choosing link $1$ and link $2$. However, choosing link $2$ increases the makespan $\MS_{T,2}$ for the red scenario from $1$ to $3/2$, resulting in an immediate increase in the makespan $\MS_T(\vec f)$ and, thus, in a discontinuity point.
    With increasing $\mu$, the flow particles start to deviate from link $1$ to link $2$ earlier in time. When $\mu=\frac{1}{5}$, already the very first particle entering the queuing system is indifferent between links $1$ and $2$. Thus, for $\mu\in(\frac{1}{5},\frac{2}{5})$, the first flow particles use link $2$ and only later in time, particles start to deviate to link $1$.
    For $\mu=\frac{2}{5}$, the expected travel time on link $1$ is sufficiently large, such that the last particle entering the queuing system is the first particle indifferent between links $1$ and $2$. In contrast to the case when $\mu=\frac{1}{10}$, this now results in an immediate decrease of the makespan.
    For $\mu\in(\frac{2}{5},\frac{1}{2})$, all flow particles use link $2$. Since the travel time of link $2$ is deterministic, this gives a constant makespan for $\mu\in(\frac{2}{5},\frac{1}{2})$.
    We can make similar observations for the reaming values of $\mu$. In short, for $\mu\in(\frac{1}{2},\frac{3}{4})$, flow particles start by using link~$2$ and then deviate to link $3$ later in time. For $\mu\in(\frac{3}{4},\frac{7}{8})$, flow particles start by using link $3$ and then deviate to link $2$ later in time. Finally, for $\mu\in(\frac{7}{8},1]$, all flow particles only use link $3$.

Additionally, \Cref{subfig:makespan:function} depicts two convex decompositions of a prior $\lambda^*$ with the beliefs~$\mu^i$ and $\tilde{\mu}^i$, $i \in [2]$, respectively. The corresponding signaling schemes with beliefs $\mu^1$ and $\mu^2$ corresponds to full information revelation and is optimal; the other signaling scheme with beliefs $\tilde{\mu}^1$ and $\tilde{\mu}^2$ corresponds to a suboptimal scheme.

\input{fig_makespan_example.tex}

\subsection{An Example for the Throughput Objective with Irrational Optimal Beliefs}
\label{sec:irrational-example}

In this section, we give an example for the throughput objective that exhibits an optimal signaling scheme that uses a convex combination of the prior~$\vec\mu$ using irrational beliefs $\vec\mu^i$.

In particular, we consider an instance with $m=3$ links, $\s=2$ scenarios, and time horizon $T=7$. The capacities of the links are $\nu_1=\frac{1}{2}$, $\nu_2 = \frac{1}{4}$, and $\nu_3=\frac{1}{3}$. The travel times depend on two scenarios (called blue and red) and are $\vec\tau_1=(1,10)^\top$, $\vec\tau_2 =(2,8)^\top$, and $\vec\tau_3 =(3,5)^\top$, respectively.

We denote by $\mu\in[0,1]$ the probability for sthe second (red) scenario.
The expected throughput for every belief $\mu\in[0,1]$ can be explicitly calculated as 
\begin{align*}
\FV(\mu)=
\begin{cases}
    \frac{1}{2} \left(-9 \mu ^2+\mu +8\right) & \text{if } 0\leq \mu \leq \frac{2}{15}, \\[4pt]
    \frac{1}{6} \mu  (3 \mu -1)+4 & \text{if } \frac{2}{15}<\mu \leq \frac{1}{4}, \\[4pt]
     \frac{23 \mu ^2}{6}-7 \mu +\frac{11}{2} & \text{if } \frac{1}{4}<\mu \leq \frac{2}{7}, \\[4pt]
     \frac{19}{4}-\frac{1}{24} \mu  (27 \mu +71) & \text{if } \frac{2}{7}<\mu \leq \frac{1}{3}, \\[4pt]
      \frac{1}{120} \bigl(\mu  (432 \mu -1111)+759\bigr) & \text{if } \frac{1}{3} <\mu\leq \frac{39}{62},\\[4pt]
      \frac{1}{12} \mu  (37 \mu -101)+6 & \frac{39}{62}<\mu \leq 1. 
\end{cases}
\end{align*}
This function is shown in \Cref{subfig:irrational:throughput}. For every prior $\prior\in \bigl(\frac{1}{36}(9-\sqrt{42}),\frac{1}{4}\bigr)$, the optimal signaling scheme induces the optimal beliefs $\mu^1=\frac{1}{36}(9-\sqrt{42})$ and $\mu^2=\frac{1}{4}$, as shown in~\Cref{subfig:irrational:throughput:zoom}. Hence, the optimal signaling scheme involves irrational numbers, even though all input numbers are rational.

\begin{figure} \centering
\begin{subfigure}{.48\textwidth}
    \centering
    \begin{tikzpicture}[scale=4]
        \node[node] (s) at (0,0) {\clap{\footnotesize s}};
        \node[node] (t) at (1,0) {\footnotesize \clap{\smash{t}}\llap{\phantom{s}}};

        \draw[-{Stealth}, very thick, mpgray]
            (s) edge[bend left=60] 
                node[pos=.5, above] {\footnotesize $\nu_1 = \nicefrac{1}{2}$ \quad $\vec\tau_1 = \sctau{1}{10}$}
            (t)
            (s) edge 
                node[pos=.5, above] {\footnotesize $\nu_2 = \nicefrac{1}{4}$ \quad $\vec\tau_2 = \sctau{2}{8}$}
            (t)
            (s) edge[bend right=60] 
                node[pos=.5, below] {\footnotesize $\nu_3 = \nicefrac{1}{3}$ \quad $\vec\tau_3 = \sctau{3}{5}$}
            (t);

        \draw[very thick, {Stealth}-, mppetrol]
            ($(s)+(-.075,0)$) -- ++(-.35, 0) node[midway, above] {\small $u=1$};
    \end{tikzpicture}
    \caption{The instance.}
    \label{subfig:irrational:instance}
\end{subfigure}

\begin{subfigure}[c]{0.48\textwidth}
\begin{tikzpicture}[xscale=5]
\useasboundingbox (-0.1,-0.5) rectangle (1.2,6);
\draw[-latex, thick] (-0.1,0) -- (1.1,0) node[right] {$\mu$};
\draw[-latex, thick] (0,-0.5) -- (0,5) node[above] {$\FV(\vec \mu)$};
\draw[thick] (0.025,4) -- (-0.025,4) node[left] {\footnotesize$4$};
\draw[thick] (0.025,2/3) -- (-0.025,2/3) node[left] {\footnotesize$2/3$};
\draw[thick] (1/4,0.1) -- (1/4,-0.1) node[below] {\footnotesize$1/4$};
\draw[thick] (1,0.1) -- (1,-0.1) node[below] {\footnotesize$1$};
\draw[smooth,variable=\x,domain=0:2/15, objectivefunction] plot({\x},{1/2*(8+\x-9*\x*\x)});
\draw[smooth,variable=\x,domain=2/15:1/4, objectivefunction] plot({\x},{4+1/6*\x*(-1+3*\x)});
\draw[smooth,variable=\x,domain=1/4:2/7, objectivefunction] plot({\x},{11/2-7*\x+(23*\x*\x)/6});
\draw[smooth,variable=\x,domain=2/7:1/3, objectivefunction] plot({\x},{19/4-1/24*\x*(71+27*\x)});
\draw[smooth,variable=\x,domain=1/3:39/62, objectivefunction] plot({\x},{1/120*(759+\x*(-1111+432*\x))});
\draw[smooth,variable=\x,domain=39/62:1, objectivefunction] plot({\x},{6+1/12*\x*(-101+37*\x)});
\end{tikzpicture}
\caption{The throughput function~$\FV(\mu)$}
\label{subfig:irrational:throughput}
\end{subfigure}
\begin{subfigure}[c]{0.5\textwidth}
\begin{tikzpicture}[xscale=5,yscale=1]
\useasboundingbox (-0.1,-0.5) rectangle (1.2,6);

\draw[-latex,thick] (-0.1,0) -- (1.1,0) node[right] {$\mu$};
\draw[-latex, thick] (0,0.5) -- (0,5) node[above] {$\FV(\vec \mu)$};
\draw[thick] (0,-0.5) -- (0,0) {};
\draw[dotted, thick] (0,0) -- (0,0.5) {};
\draw[xscale=4,yscale=40,thick] (0.005,0.1) -- (-0.005,0.1) node[left] {\footnotesize$4.05$};
\draw[xscale=4,yscale=40,thick] (0.005,0.05) -- (-0.005,0.05) node[left] {\footnotesize$4.00$};
\draw[xscale=4,yscale=40,thick] (2/15,0.0025) -- (2/15,-0.0025) node[below] {\footnotesize$2/15$};
\draw[xscale=4,yscale=40,thick] (1/4,0.0025) -- (1/4,-0.0025) node[below] {\footnotesize$\mu^2 = 1/4$};
\draw[xscale=4,yscale=40,thick] (0.07,0.0025) -- (0.07,-0.0025) node[below] {\footnotesize$\mu^1$};
\draw[smooth,variable=\x,domain=0:2/15,xscale=4,yscale=40,objectivefunction] plot({\x},{1/2*(8+\x-9*\x*\x)-3.95});
\draw[smooth,variable=\x,domain=2/15:1/4,xscale=4,yscale=40,objectivefunction] plot({\x},{4+1/6*\x*(-1+3*\x)-3.95});
\fill[mpgreen,xscale=4,yscale=40] 
    (0.07,{(4.01296-3.95)}) ellipse (.125pt and .0625pt)
    (1/4,383/96-3.95) ellipse (.125pt and .0625pt);
\draw[mpgreen,xscale=4,yscale=40,very thick,dashed] (0.07,4.01296-3.95) -- (1/4,383/96-3.95);
\end{tikzpicture}
\caption{The throughput function~$\FV(\mu)$ for $\mu \leq \frac{1}{4}$.
}
\label{subfig:irrational:throughput:zoom}
\end{subfigure}
\caption{
(a) The instance from~\Cref{sec:irrational-example} with $n=3$~links and $m=2$~scenarios (\textcolor{mpblue}{blue} and \textcolor{mpred}{red}).
(b) The throughput $\FV(\mu)$ for $\mu\in[0,1]$.
(c) The  \textcolor{mpgreen}{optimal convex decomposition} for $\prior\in(\mu^1,1/4)$ with $\mu^1=1/36(9-\sqrt{42})$.}
\end{figure}

\section{Deferred Material from Section~\ref{sec:structure-deterministic}}
\label{app:structure-deterministic}

\begin{lemma} \label{lem:equilibrium}
    Let $\vec f = (f_i)_{e \in E}$ be the flow defined in~\eqref{eq:equilibrium:flow}. Then,
        \begin{enumerate}[label=(\roman*)]
            \item the flow $\vec f$
            is a feasible flow with queue lengths
                \[
                    z_i (\theta) =
                    \begin{cases}
                        \displaystyle
                        \nu_i \Biggl(
                            \frac{\arrival - \bar{\nu} (\nLinksEQ)}{\bar{\nu} (\nLinksEQ)}
                            \Biggl(
                                \theta +
                                \sum_{j=1}^{\nLinksEQ}
                                    \frac{\arrival \, \nu_j \, \tau_j}{(\arrival - \bar{\nu}(j)) (\arrival - \bar{\nu} (j-1))}
                            \Biggr)
                            - \tau_i
                        \Biggr)
                        &
                        \text{if } \theta^*_i \leq \theta < \theta^*_{k+1}, i \leq k,
                        \\
                        \nu_i \bigl(\tau_{k+1} - \tau_i \bigr)
                        &
                        \text{if } \theta \geq \theta^*_{k+1}, i \leq k,
                        \\
                        0 & \text{otherwise,}
                    \end{cases}
                \]
                where $S(\theta) = \bigl\{ e \in [m] \;\big\vert\; \theta^*_i \leq \theta \bigr\}$ is the support at time $\theta$ and $\nLinksEQ \coloneqq |S(\theta)|$ is the number of links in the support at time~$\theta$.
            \item
                the flow $\vec f$ is a Nash equilibrium.
            \item
                the queue lengths $z_i(\theta)$ are non-increasing in $\tau_i$ and non-decreasing in $\tau_j$ for $j \neq i$. If, additionally, $i \leq k$ and $\theta > \theta^*_i$, then $z_i(\theta)$ is strictly decreasing in $\tau_i$.
        \end{enumerate}
    \end{lemma}

\begin{proof}
        By definition, $f_i (\theta) \geq 0$ for all $e \in E$ and $\theta \geq 0$ and $\sum_{e \in E} f_i (\theta) = \arrival$. Hence, $\vec f$ is a feasible flow.

        We compute the queue length at time $\theta$ for some fixed link~$i \in [m]$.
        First, assume $\theta < \theta^*_{k+1}$.
        If, additionally, $\theta < \theta^*_i$, then $z_i (\theta) = 0$ since $f_i (s) = 0$ for all $0 \leq s \leq \theta < \theta^*_i$. Thus, also $z_i(\theta) = 0$ for all $i \in [m]$ with $\theta^*_i \geq \theta$.

        Assume that $\theta \geq \theta^*_i$. We observe, that for all $\theta_i^* \leq s \leq \theta < \theta^*_{k+1}$ we have $f_i (s) \geq \nu_i$. Therefore, we can compute the queue length by computing the integral of~\eqref{eq:queue:dynamics} as follows.
        \begin{align}
            z_i (\theta) &=
            \int_0^{\theta} z_i' (s) \mathrm{d}s
            = \int_0^{\theta^*_i} [0 - \nu_i]^+ \mathrm{d}s
            + \int_{\theta^*_i}^\theta [f_i - \nu_i]^+ \mathrm{d}s \notag\\
            &=
            \sum_{j=i}^{\nLinksEQ-1} \int_{\theta^*_j}^{\theta_{j+1}^*} \Big(\frac{\arrival}{\bar{\nu}(j)} \nu_i - \nu_i \Big) \mathrm{d}s
            +\int_{\theta^*_{\nLinksEQ}}^{\theta} \Big(\frac{\arrival}{\bar{\nu}(\nLinksEQ)} \nu_i - \nu_i \Big)  \mathrm{d}s \notag\\
            &=
            \nu_i \Biggl(
            \sum_{j=i}^{\nLinksEQ-1}
            \frac{\arrival - \bar{\nu} (j)}{\bar{\nu} (j)}
            (\theta^*_{j+1} - \theta^*_j)
            +
            \frac{\arrival - \bar{\nu} (\nLinksEQ)}{\bar{\nu} (\nLinksEQ)}
            \big(\theta - \theta_{\nLinksEQ}^* \big)
            \Biggr) \notag\\
            &\stackrel{\mathclap{\eqref{eq:equilibrium:breakpoints}}}{=}
            \nu_i \Biggl(
                \tau_{\nLinksEQ} - \tau_i +
                \frac{\arrival - \bar{\nu} (\nLinksEQ)}{\bar{\nu} (\nLinksEQ)}
                \big(\theta - \theta_{\nLinksEQ}^* \big)
            \Biggr) \label{eq:prf:equilibrium:zi1}
        \end{align}
        From~\eqref{eq:equilibrium:breakpoints}, we additionally obtain
        for $i' \leq k$
        \begin{align*}
            \theta^*_{i'} &= \sum_{j = 1}^{i' - 1} \frac{\bar{\nu}(j)}{\arrival - \bar{\nu}(j)} (\tau_{j+1} - \tau_j)
            =
            \frac{\bar{\nu}(i'-1)}{\arrival - \bar{\nu}(i' - 1)} \tau_{i'} +
            \sum_{j=1}^{i'-1}
            \biggl(
                \frac{\bar{\nu}(j-1)}{\arrival - \bar{\nu}(j-1)} - \frac{\bar{\nu}(j)}{\arrival - \bar{\nu}(j)}
            \biggr) \, \tau_j
            \\
            &=
            \frac{\bar{\nu}(i')}{\arrival - \bar{\nu}(i')} \tau_{i'} +
            \sum_{j=1}^{i'}
            \frac{\arrival (\bar{\nu} (j-1) - \bar{\nu} (j))}{(\arrival - \bar{\nu}(j-1))(\arrival - \bar{\nu}(j))} \, \tau_j
            =
            \frac{\bar{\nu}(i')}{\arrival - \bar{\nu}(i')} \tau_{i'}
            -
            \sum_{j=1}^{i'}
            \frac{\arrival \, \nu_j \, \tau_j}{(\arrival - \bar{\nu}(j-1))(\arrival - \bar{\nu}(j))}
        \end{align*}
        and compute
        \begin{align*}
            z_i (\theta) &=
            \nu_i \Biggl(
                \frac{\arrival - \bar{\nu} (\nLinksEQ)}{\bar{\nu} (\nLinksEQ)}
                \bigg(\theta +
                \sum_{j=1}^{i'}
                \frac{\arrival \, \nu_j \, \tau_j}{(\arrival - \bar{\nu}(j-1))(\arrival - \bar{\nu}(j))}
                \bigg)
                - \tau_i
            \Biggr).
        \end{align*}

        Now, consider the case $\theta > \theta^*_{k+1}$. By~\eqref{eq:equilibrium:flow}, $f_i (\theta) \leq \nu_i$ if $\theta \geq \theta^*_{k+1}$ and, thus, the queues do not grow after the time $\theta^*_{k+1}$. Therefore, we compute for every $i \leq k$
        \begin{align*}
            z_i (\theta) &=
            \int_0^{\theta} z_i' (s) \mathrm{d}s
            = \int_0^{\theta^*_i} [0 - \nu_i]^+ \mathrm{d}s
            + \int_{\theta^*_i}^\theta [f_i - \nu_i]^+ \mathrm{d}s \\
            &=
            \sum_{j=i}^{k} \int_{\theta^*_j}^{\theta_{j+1}^*} \Big(\frac{\arrival}{\bar{\nu}(j)} \nu_i - \nu_i \Big) \mathrm{d}s
            +
            \int_{\theta^*_{k}}^{\theta} \underbrace{[f_i(\theta) - \nu_i]^+}_{=0}  \mathrm{d}s
            =
            \nu_i \big( \tau_{k} - \tau_i \big),
        \end{align*}
        where, in the last step, we used~\eqref{eq:equilibrium:breakpoints} and the same computations as above.
        Since links $i > k + 1$ are never used, $z_i (\theta) = 0$ for all $i > k + 1$. This concludes the proof of statement~\emph{(i)}.

        Finally, we compute the travel times of the links using the formula
        $
            T_i (\theta) = \theta + \frac{z_i(\theta)}{\nu_i} + \tau_i.
        $.
        For $i \leq k$, we observe by the explicit formulas computed above that $T_i$ is independent of~$i$ and, therefore, all links have the same travel time.
        The link $k+1$ is used only if $\theta \geq \theta^*_{k+1}$. In this case, all links have the same travel time $\tau_{k+1}$.
        All other links $i > k+1$ have travel times $\tau_i \geq \tau_{k+1}$ since the links are ordered by increasing travel times.
        Hence, the flow $\vec{f}$ is also a Nash equilibrium, proving~\emph{(ii)}.

        For the monotonicity of the queue lengths~$z_i(\theta)$, we first observe that, if $\theta \geq \theta_{k+1}^*$, statement~\emph{(iii)} follows immediately. Further, if $\theta^*_j > \theta$, then the queue lengths~$z_i(\theta)$ is independent of the variable~$\tau_j$, and the statement~\emph{(iii)} follows in this case as well. The same holds for all queue lengths~$z_i(\theta)$ with $i > k$.
        Finally, we consider the case $\theta^*_{j} \leq \theta < \theta^*_{k+1}$ and $i \leq k$. If $j \neq i$, the queue lengths $z_i (\theta)$ depends on~$\tau_j$ via the term
        \[
            \underbrace{\nu_i \frac{\arrival - \bar{\nu}(\nLinksEQ)}{\bar{\nu} (\nLinksEQ)} \frac{\arrival \, \nu_j}{(\arrival - \bar{\nu} (j)) (\arrival - \bar{\nu} (j-1))}}_{\geq 0} \tau_j
            .
        \]
        Since the coefficient of $\tau_j$ is non-negative, the function is non-decreasing in~$\tau_j$. If $j = i$, $z_i (\theta)$ depends on~$\tau_i$ via the term
        \begin{align*}
            \nu_i \biggl( \frac{\arrival - \bar{\nu}(\nLinksEQ)}{\bar{\nu} (\nLinksEQ)} \frac{\arrival \, \nu_i}{(\arrival - \bar{\nu} (i)) (\arrival - \bar{\nu} (i-1))} - 1\biggr) \tau_i
            .
        \end{align*}
        The coefficient of~$\tau_i$ is
        \begin{align*}
            \nu_i \biggl( &\frac{\arrival - \bar{\nu}(\nLinksEQ)}{\bar{\nu} (\nLinksEQ)} \frac{\arrival \, \nu_i}{(\arrival - \bar{\nu} (i)) (\arrival - \bar{\nu} (i-1))} - 1\biggr)
            \leq
            \nu_i \biggl( \frac{\arrival \, \nu_i}{\bar{\nu}(i) (\arrival - \bar{\nu} (i-1))} - 1 \biggr) \\
            &=
            \nu_i \biggl( \frac{\arrival \, \nu_i}{\arrival \, \bar{\nu}(i) - \bar{\nu}(i) \bar{\nu} (i-1)} - 1 \biggr)
            <
            \nu_i \biggl( \frac{\arrival \, \nu_i}{\arrival \, (\bar{\nu}(i) - \bar{\nu} (i-1))} - 1 \biggr) = 0,
        \end{align*}
        where we used in the last inequality that $\bar{\nu} (i) \leq \bar{\nu} (k) < \arrival$. Since the coefficient is strictly negative, the queue lengths is strictly decreasing.
    \end{proof}

\section{Proof of \Cref{lem:piecewise-linear}}
\label{app:lem:piecewise-linear}

\begin{proof}
    First, we note that by definition and~\eqref{eq:equilibrium:breakpoints:mu}, the exit times~$\exitTime_{i,s} (\vec{\mu})$ are linear on every polytope~$P$ induced by the hyperplanes~$H_{i,j}$, since on $P$ the permuation~$\pi(\,\cdot\,; \vec{\mu})$ is constant and, by~\eqref{eq:equilibrium:breakpoints:mu}, the entry times~$\theta^*_i$ are linear.

    Every polytope~$Q$ induced by the hyperplanes $H_{i,j}, H_{i,j,s}, H_{i,s,T}$ is a subset of some polytope~$P$. Thus, the exit times $\exitTime_{i,s}$ are still linear.
    Further, the definitions of the hyperplanes $H_{i,j,s}, H_{i,s,T}$ ensure that the ordering $\sigma_{s} (\,\cdot\,; \vec{\mu})$ of the exit times is constant on~$Q$.
    Further, the number $\numContributingLinks$ of links contributing to the throughput are unchanged.
    The same holds for $\bar{\nu} (\numContributingLinks; \vec{\mu})$.
    Therefore, all terms in the explicit formula for $F_s(\vec{\mu})$ in \Cref{lem:flow-value:function} are constants, except for the exit times. Since the latter is linear in $\vec{\mu}$, the function $F_s(\vec{\mu})$ is affine linear on~$Q$ and, thus, piecewise linear on~$P$. This implies statements \emph{(i)} and \emph{(ii)}.

    It remains to show that $F_s$ is convex on every polytope~$P$ induced by the hyperplanes~$H_{i,j}$. We again use that the exit times are linear on~$P$. 
    In order to show convexity, let $\vec{\mu}^{(1)}, \vec{\mu}^{(2)} \in P$ and let
    \[
        \vec{\mu}^{(3)} \coloneqq \vec{\mu}(\xi) = \xi \vec{\mu}^{(1)} + (1 - \xi) \vec{\mu}^{(2)}
    \]
    for some $\xi \in [0, 1]$.
    All three beliefs $\vec{\mu}^{(1)},\vec{\mu}^{(2)}$ and $\vec{\mu} (\xi)$ are contained in $P$ but may lie in different polytopes induced by $H_{i,j}, H_{i,j,s}, H_{i,s,T}$.
    Let $Q_1, Q_2, Q_3$ be the three polytopes induced by~$H_{i,j}, H_{i,j,s}, H_{i,s,T}$ such that
    \[
        \vec{\mu}^{(1)} \in Q_1, \vec{\mu}^{(2)} \in Q_2 \text{ and } \vec{\mu} (\xi) \in Q_3
        .
    \]
    For ease of notation, we denote by $\sigma_{Q_j} (i) \coloneqq \sigma_{s} (i; \vec{\mu}^{(j)})$ the ordering of the exit times on $Q_1, Q_2$, and~$Q_3$, respectively.
    Similarly, we write $\bar{\nu}_{Q_j} \coloneqq \bar{\nu}_{s} (\numContributingLinks[\vec{\mu}^{(j)}], \vec{\mu}^{(j)})$ and $\numContributingLinks*_{Q_j} = \numContributingLinks[\vec{\mu}^{(j)}]$.
    Finally, we define the following index sets
    \[
        I_j \coloneqq \bigl\{ i \in [m] \mid \sigma_{Q_j} (i) \leq \numContributingLinks*_{Q_j} \bigr\}
        , \qquad i=1,2,3
    \]
    containing all indices of links contributing to the throughput in $\vec{\mu}^{(j)}$. We observe
    \begin{align}
        \exitTime_{i, s} (\vec{\mu}^{(j)}) &\leq \exitTime_{\sigma_{Q_{j}}(\numContributingLinks*_{Q_{j}}), s} (\vec{\mu}^{(j)}) \leq T
        &&\text{if}\quad
        i \in I_j
        ,
        \label{eq:prf:convexity:11}\\
        \exitTime_{i, s} (\vec{\mu}^{(j)}) &\geq \exitTime_{\sigma_{Q_{j}}(\numContributingLinks*_{Q_{j}}), s} (\vec{\mu}^{(j)})
        &&\text{if}\quad
        i \notin I_j
        ,
        \label{eq:prf:convexity:12}\\
        \exitTime_{i, s} (\vec{\mu}^{(j)}) &\geq T \geq \exitTime_{\sigma_{Q_{j}}(\numContributingLinks*_{Q_{j}}), s} (\vec{\mu}^{(j)})
        &&\text{if}\quad
        i \notin I_j \text{ and } \bar{\nu}_{Q_j} < \arrival
        \label{eq:prf:convexity:13}
        ,
        \\
        \text{and}\qquad
        \bar{\nu}_{Q_j} - \nu_{\sigma_{Q_{j}}(\numContributingLinks*_{Q_{j_1}})} &= \bar{\nu}_{s} (\numContributingLinks*_{Q_{j}} - 1; \vec{\mu}^{(j)}) < \arrival
        &&\text{if}\quad
        \bar{\nu}_{Q_j} \geq \arrival 
        \label{eq:prf:convexity:3}
        ,
    \end{align}
    where \eqref{eq:prf:convexity:11}-\eqref{eq:prf:convexity:13} follow from the definition of the index~$\sigma_{Q_{j}}(\numContributingLinks*_{Q_{j}})$ as the index of the last link that contributes to the throughput.
    Fact~\eqref{eq:prf:convexity:3} follows, since, by definition, $\sigma_{Q_{j}}(\numContributingLinks*_{Q_{j_1}})$ is the first index such that the summed capacities $\bar{\nu}_{s} (i; \vec{\mu}^{(j)})$ exceed the inflow~$\arrival$.
    Additionally, we observe that $\bar{\nu}_{Q_j}  = \sum_{i \in I_j} \nu_i$ and, thus,
    \begin{equation}\label{eq:prf:convexity:2}
        \bar{\nu}_{Q_{j_1}} - \bar{\nu}_{Q_{j_2}}  = \sum_{i \in I_{j_1} \setminus I_{j_2}} \nu_i - \sum_{i \in I_{j_2} \setminus I_{j_1}} \nu_i
        .
    \end{equation}

    \begin{claim}\label{claim:prf:convexity}
        For $j \in {1, 2}$ we have
        \begin{align*}
            A_j \coloneqq
            \arrival T + T \cdot \big[\bar{\nu}_{Q_3} - \arrival\big]^-
            - \sum_{i \in I_3} \nu_{i} \exitTime_{i, s}(\vec{\mu}^{(j)})
            + \exitTime_{\sigma_{Q_3} (\numContributingLinks*_{Q_3}), s}(\vec{\mu}^{(j)}) \cdot \big[\bar{\nu}_{Q_3} - \arrival\big]^+
            \leq
            F_s (\vec{\mu}^{(j)})
            .
        \end{align*}
    \end{claim}
    \begin{proof}[Proof of~\Cref{claim:prf:convexity}]
        With the formula from~\Cref{lem:flow-value:function}, we compute
        \begin{align*}
            A_j &=
            F_s (\vec{\mu}^{(j)})
            + T \bigl( \big[\bar{\nu}_{Q_3} - \arrival\big]^- - \big[\bar{\nu}_{Q_j} - \arrival\big]^- \bigr)
            + \sum_{i \in I_j \setminus I_3} \nu_{i} \exitTime_{i, s}(\vec{\mu}^{(j)}) - \sum_{i \in I_3 \setminus I_j} \nu_{i} \exitTime_{i, s}(\vec{\mu}^{(j)})
            \\
            &\qquad
            + \exitTime_{\sigma_{Q_3} (\numContributingLinks*_{Q_3}), s}(\vec{\mu}^{(j)})  \cdot  \big[\bar{\nu}_{Q_3} - 
            \arrival\big]^+ 
            - 
            \exitTime_{\sigma_{Q_j} (\numContributingLinks*_{Q_j}), s}(\vec{\mu}^{(j)})  \cdot 
            \big[\bar{\nu}_{Q_j} - \arrival\big]^+
        \end{align*}
        There are four possible cases.

        \noindent\textbf{Case 1:} $\bar{\nu}_{Q_j} < \arrival$ and $\bar{\nu}_{Q_3} < \arrival$. Then,
        \begin{align*}
            A_j &= F_s (\vec{\mu}^{(j)}) +
            T \bigl( \bar{\nu}_{Q_3} - \bar{\nu}_{Q_j} \bigr)
            + \sum_{i \in I_j \setminus I_3} \nu_{i} \exitTime_{i, s}(\vec{\mu}^{(j)}) - \sum_{i \in I_3 \setminus I_j} \nu_{i} \exitTime_{i, s}(\vec{\mu}^{(j)}) \\
            &\stackrel{\mathclap{\eqref{eq:prf:convexity:2}}}{=}
            F_s (\vec{\mu}^{(j)})
            - \sum_{i \in I_j \setminus I_3} \nu_{i} \big(T - \exitTime_{i, s}(\vec{\mu}^{(j)})\big)
            + \sum_{i \in I_3 \setminus I_j} \nu_{i} \big(T - \exitTime_{i, s}(\vec{\mu}^{(j)})\big)
            \;\stackrel{\mathclap{\eqref{eq:prf:convexity:11},\eqref{eq:prf:convexity:13}}}{\leq}\;
            F_s (\vec{\mu}^{(j)})
            .
        \end{align*}

        \noindent\textbf{Case 2:} $\bar{\nu}_{Q_j} < \arrival$ and $\bar{\nu}_{Q_3} \geq \arrival$. Then,
        \begin{align*}
            A_j &= F_s (\vec{\mu}^{(j)}) 
            +
            \exitTime_{\sigma_{Q_3} (\numContributingLinks*_{Q_3}), s}(\vec{\mu}^{(j)}) \big( \bar{\nu}_{Q_3} - \arrival \big)
            - T \big( \bar{\nu}_{Q_j} - \arrival \big)
            + \sum_{i \in I_j \setminus I_3} \nu_{i} \exitTime_{i, s}(\vec{\mu}^{(j)}) - \sum_{i \in I_3 \setminus I_j} \nu_{i} \exitTime_{i, s}(\vec{\mu}^{(j)}) \\
            &\stackrel{\mathclap{\eqref{eq:prf:convexity:2}}}{=}
            F_s (\vec{\mu}^{(j)}) 
            +
            \exitTime_{\sigma_{Q_3} (\numContributingLinks*_{Q_3}), s}(\vec{\mu}^{(j)}) \big( \bar{\nu}_{Q_3} - \arrival \big)
            - T \bigg(\bar{\nu}_{Q_3} + \sum_{i \in I_{j} \setminus I_{3}} \nu_i - \sum_{i \in I_{3} \setminus I_{j}} \nu_i - \arrival \bigg) \\
            &\qquad
            + \sum_{i \in I_j \setminus I_3} \nu_{i} \exitTime_{i, s}(\vec{\mu}^{(j)}) - \sum_{i \in I_3 \setminus I_j} \nu_{i} \exitTime_{i, s}(\vec{\mu}^{(j)})
            \\
            &=
            F_s (\vec{\mu}^{(j)})
            - \sum_{i \in I_j \setminus I_3} \nu_{i} \big(T - \exitTime_{i, s}(\vec{\mu}^{(j)})\big)
            + \sum_{i \in I_3 \setminus I_j} \nu_{i} \big(T - \exitTime_{i, s}(\vec{\mu}^{(j)})\big)
            +
            \big( \exitTime_{\sigma_{Q_3} (\numContributingLinks*_{Q_3}), s}(\vec{\mu}^{(j)}) - T \big) \big( \bar{\nu}_{Q_3} - \arrival \big)
            .
        \end{align*}
        If $\exitTime_{\sigma_{Q_3} (\numContributingLinks*_{Q_3}), s}(\vec{\mu}^{(j)}) - T \leq 0$, we obtain
        \begin{align*}
            A_j &\leq 
            F_s (\vec{\mu}^{(j)})
            - \sum_{i \in I_j \setminus I_3} \nu_{i} \big(T - \exitTime_{i, s}(\vec{\mu}^{(j)})\big)
            + \sum_{i \in I_3 \setminus I_j} \nu_{i} \big(T - \exitTime_{i, s}(\vec{\mu}^{(j)})\big)
            \;\stackrel{\mathclap{\eqref{eq:prf:convexity:11},\eqref{eq:prf:convexity:13}}}{\leq}\;
            F_s (\vec{\mu}^{(j)})
            .
        \end{align*}
        Otherwise, if $\exitTime_{\sigma_{Q_3} (\numContributingLinks*_{Q_3}), s}(\vec{\mu}^{(j)}) - T > 0$, then $\sigma_{Q_3} (\numContributingLinks*_{Q_3}) \notin I_j$ and, thus, $\sigma_{Q_3} (\numContributingLinks*_{Q_3}) \in I_3 \setminus I_j$. Then,
        \begin{align*}
            A_j &=
            F_s (\vec{\mu}^{(j)})
            - \sum_{i \in I_j \setminus I_3} \nu_{i} \big(T - \exitTime_{i, s}(\vec{\mu}^{(j)})\big)
            + \sum_{\mathclap{\substack{i \in I_3 \setminus I_j\\ i \neq \sigma_{Q_3} (\numContributingLinks*_{Q_3})}}} \nu_{i} \big(T - \exitTime_{i, s}(\vec{\mu}^{(j)})\big)
            +
            \nu_{\sigma_{Q_3} (\numContributingLinks*_{Q_3})} \big(T - \exitTime_{\sigma_{Q_3} (\numContributingLinks*_{Q_3}), s}(\vec{\mu}^{(j)})\big)
            \\
            &\qquad
            +
            \big( \exitTime_{\sigma_{Q_3} (\numContributingLinks*_{Q_3}), s}(\vec{\mu}^{(j)}) - T \big) \big( \bar{\nu}_{Q_3} - \arrival \big)
            \\
            &=
            F_s (\vec{\mu}^{(j)})
            - \sum_{i \in I_j \setminus I_3} \nu_{i} \big(T - \exitTime_{i, s}(\vec{\mu}^{(j)})\big)
            + \sum_{\mathclap{\substack{i \in I_3 \setminus I_j\\ i \neq \sigma_{Q_3} (\numContributingLinks*_{Q_3})}}} \nu_{i} \big(T - \exitTime_{i, s}(\vec{\mu}^{(j)})\big)
            \\
            &\qquad
            +
            \big( \exitTime_{\sigma_{Q_3} (\numContributingLinks*_{Q_3}), s}(\vec{\mu}^{(j)}) - T \big) \underbrace{ \big( \bar{\nu}_{Q_3} - \arrival - \nu_{\sigma_{Q_3} (\numContributingLinks*_{Q_3})} \big) }_{\leq 0, \text{ by \eqref{eq:prf:convexity:3}}}
            \;\stackrel{\mathclap{\eqref{eq:prf:convexity:11},\eqref{eq:prf:convexity:13}}}{\leq}
            F_s (\vec{\mu}^{(j)})
        \end{align*}

        \noindent\textbf{Case 3:} $\bar{\nu}_{Q_j} \geq \arrival$ and $\bar{\nu}_{Q_3} < \arrival$. Then,
        \begin{align*}
            A_j &= F_s (\vec{\mu}^{(j)}) +
            T \big( \bar{\nu}_{Q_3} - \arrival \big) - \exitTime_{\sigma_{Q_j} (\numContributingLinks*_{Q_j}), s}(\vec{\mu}^{(j)}) \big( \bar{\nu}_{Q_j} - \arrival \big)
            + \sum_{i \in I_j \setminus I_3} \nu_{i} \exitTime_{i, s}(\vec{\mu}^{(j)}) - \sum_{i \in I_3 \setminus I_j} \nu_{i} \exitTime_{i, s}(\vec{\mu}^{(j)}) \\
            &\stackrel{\mathclap{\eqref{eq:prf:convexity:2}}}{=}
            F_s (\vec{\mu}^{(j)}) +
            T \big( \bar{\nu}_{Q_3} - \arrival \big)
            - \exitTime_{\sigma_{Q_j} (\numContributingLinks*_{Q_j}), s}(\vec{\mu}^{(j)}) \bigg(\bar{\nu}_{Q_3} + \sum_{i \in I_{j} \setminus I_{3}} \nu_i - \sum_{i \in I_{3} \setminus I_{j}} \nu_i - \arrival \bigg) 
            \\
            &\qquad
            + \sum_{i \in I_j \setminus I_3} \nu_{i} \exitTime_{i, s}(\vec{\mu}^{(j)}) - \sum_{i \in I_3 \setminus I_j} \nu_{i} \exitTime_{i, s}(\vec{\mu}^{(j)})
            \\
            &=
            F_s (\vec{\mu}^{(j)})
            - \sum_{i \in I_j \setminus I_3} \nu_{i} \big(\exitTime_{\sigma_{Q_j} (\numContributingLinks*_{Q_j}), s}(\vec{\mu}^{(j)}) - \exitTime_{i, s}(\vec{\mu}^{(j)})\big)
            + \sum_{i \in I_3 \setminus I_j} \nu_{i} \big(\exitTime_{\sigma_{Q_j} (\numContributingLinks*_{Q_j}), s}(\vec{\mu}^{(j)}) - \exitTime_{i, s}(\vec{\mu}^{(j)})\big)
            \\
            &\qquad
            + \big(T - \exitTime_{\sigma_{Q_j} (\numContributingLinks*_{Q_j}), s}(\vec{\mu}^{(j)}) \big) \big( \bar{\nu}_{Q_3} - \arrival \big)
            \\
            \;&\stackrel{\mathclap{\eqref{eq:prf:convexity:11},\eqref{eq:prf:convexity:12}}}{\leq}\;
            F_s (\vec{\mu}^{(j)})
            +
            \big(T - \exitTime_{\sigma_{Q_j} (\numContributingLinks*_{Q_j}), s}(\vec{\mu}^{(j)}) \big) \underbrace{\big( \bar{\nu}_{Q_3} - \arrival \big)}_{< 0} <  F_s (\vec{\mu}^{(j)})
            .
        \end{align*}

        \noindent\textbf{Case 4:} $\bar{\nu}_{Q_j} \geq \arrival$ and $\bar{\nu}_{Q_3} \geq \arrival$. Then,
        \begin{align*}
            A_j &= F_s (\vec{\mu}^{(j)}) +
            \exitTime_{\sigma_{Q_3} (\numContributingLinks*_{Q_3}), s}(\vec{\mu}^{(j)})  \big( \bar{\nu}_{Q_3} - \arrival \big)
            - \exitTime_{\sigma_{Q_j} (\numContributingLinks*_{Q_j}), s}(\vec{\mu}^{(j)}) \big( \bar{\nu}_{Q_j} - \arrival \big)
            \\
            &\qquad
            + \sum_{i \in I_j \setminus I_3} \nu_{i} \exitTime_{i, s}(\vec{\mu}^{(j)}) - \sum_{i \in I_3 \setminus I_j} \nu_{i} \exitTime_{i, s}(\vec{\mu}^{(j)}) \\
            &\stackrel{\mathclap{\eqref{eq:prf:convexity:2}}}{=} F_s (\vec{\mu}^{(j)}) +
            \exitTime_{\sigma_{Q_3} (\numContributingLinks*_{Q_3}), s}(\vec{\mu}^{(j)})  \big( \bar{\nu}_{Q_3} - \arrival \big)
            - \exitTime_{\sigma_{Q_j} (\numContributingLinks*_{Q_j}), s}(\vec{\mu}^{(j)}) \bigg(\bar{\nu}_{Q_3} + \sum_{i \in I_{j} \setminus I_{3}} \nu_i - \sum_{i \in I_{3} \setminus I_{j}} \nu_i - \arrival \bigg)
            \\
            &\qquad
            + \sum_{i \in I_j \setminus I_3} \nu_{i} \exitTime_{i, s}(\vec{\mu}^{(j)}) - \sum_{i \in I_3 \setminus I_j} \nu_{i} \exitTime_{i, s}(\vec{\mu}^{(j)}) \\
            &=
            F_s (\vec{\mu}^{(j)})
            - \sum_{i \in I_j \setminus I_3} \nu_{i} \big(\exitTime_{\sigma_{Q_j} (\numContributingLinks*_{Q_j}), s}(\vec{\mu}^{(j)}) - \exitTime_{i, s}(\vec{\mu}^{(j)})\big)
            + \sum_{i \in I_3 \setminus I_j} \nu_{i} \big(\exitTime_{\sigma_{Q_j} (\numContributingLinks*_{Q_j}), s}(\vec{\mu}^{(j)}) - \exitTime_{i, s}(\vec{\mu}^{(j)})\big)
            \\
            &\qquad
            + \Big( \exitTime_{\sigma_{Q_3} (\numContributingLinks*_{Q_3}), s} (\vec{\mu}^{(j)}) - \exitTime_{\sigma_{Q_j} (\numContributingLinks*_{Q_j}), s}(\vec{\mu}^{(j)}) \Big) \big( \bar{\nu}_{Q_3} - \arrival \big)
        \end{align*}
        If $\exitTime_{\sigma_{Q_3} (\numContributingLinks*_{Q_3}), s} (\vec{\mu}^{(j)}) - \exitTime_{\sigma_{Q_j} (\numContributingLinks*_{Q_j}), s}(\vec{\mu}^{(j)}) \leq 0$, then we obtain
        \begin{align*}
            A_j &\leq 
            F_s (\vec{\mu}^{(j)})
            - \sum_{i \in I_j \setminus I_3} \nu_{i} \big(\exitTime_{\sigma_{Q_j} (\numContributingLinks*_{Q_j}), s}(\vec{\mu}^{(j)}) - \exitTime_{i, s}(\vec{\mu}^{(j)})\big)
            + \sum_{i \in I_3 \setminus I_j} \nu_{i} \big(\exitTime_{\sigma_{Q_j} (\numContributingLinks*_{Q_j}), s}(\vec{\mu}^{(j)}) - \exitTime_{i, s}(\vec{\mu}^{(j)})\big)
            \\
            \;&\stackrel{\mathclap{\eqref{eq:prf:convexity:11},\eqref{eq:prf:convexity:12}}}{\leq}\;
            F_s (\vec{\mu}^{(j)})
            .
        \end{align*}
        If $\exitTime_{\sigma_{Q_3} (\numContributingLinks*_{Q_3}), s} (\vec{\mu}^{(j)}) - \exitTime_{\sigma_{Q_j} (\numContributingLinks*_{Q_j}), s}(\vec{\mu}^{(j)}) > 0$, then $\sigma_{Q_3} (\numContributingLinks*_{Q_3}) \notin I_j$ and, thus, $\sigma_{Q_3} (\numContributingLinks*_{Q_3}) \in I_3 \setminus I_1$.
        Therefore, we get
        \begin{align*}
            A_j &\leq 
            F_s (\vec{\mu}^{(j)})
            - \sum_{i \in I_j \setminus I_3} \nu_{i} \big(\exitTime_{\sigma_{Q_j} (\numContributingLinks*_{Q_j}), s}(\vec{\mu}^{(j)}) - \exitTime_{i, s}(\vec{\mu}^{(j)})\big)
            + \sum_{\mathclap{\substack{i \in (I_3 \setminus I_j) \\ i \neq \sigma_{Q_3} (\numContributingLinks*_{Q_3})}}} \nu_{i} \big(\exitTime_{\sigma_{Q_j} (\numContributingLinks*_{Q_j}), s}(\vec{\mu}^{(j)}) - \exitTime_{i, s}(\vec{\mu}^{(j)})\big) \\
            &\qquad
            +  \nu_{\sigma_{Q_3} (\numContributingLinks*_{Q_3})} \big(\exitTime_{\sigma_{Q_j} (\numContributingLinks*_{Q_j}), s}(\vec{\mu}^{(j)}) - \exitTime_{\sigma_{Q_3} (\numContributingLinks*_{Q_3}), s}(\vec{\mu}^{(j)})\big)
            \\
            &\qquad
            +
            \Big( \exitTime_{\sigma_{Q_3} (\numContributingLinks*_{Q_3}), s} (\vec{\mu}^{(j)}) - \exitTime_{\sigma_{Q_j} (\numContributingLinks*_{Q_j}), s}(\vec{\mu}^{(j)}) \Big) \big( \bar{\nu}_{Q_3} - \arrival \big) \\
            &=
            F_s (\vec{\mu}^{(j)})
            - \sum_{i \in I_j \setminus I_3} \nu_{i} \big(\exitTime_{\sigma_{Q_j} (\numContributingLinks*_{Q_j}), s}(\vec{\mu}^{(j)}) - \exitTime_{i, s}(\vec{\mu}^{(j)})\big)
            + \sum_{\mathclap{\substack{i \in (I_3 \setminus I_j) \\ i \neq \sigma_{Q_3} (\numContributingLinks*_{Q_3})}}} \nu_{i} \big(\exitTime_{\sigma_{Q_j} (\numContributingLinks*_{Q_j}), s}(\vec{\mu}^{(j)}) - \exitTime_{i, s}(\vec{\mu}^{(j)})\big) \\
            &\qquad
            +
            \Big( \exitTime_{\sigma_{Q_3} (\numContributingLinks*_{Q_3}), s} (\vec{\mu}^{(j)}) - \exitTime_{\sigma_{Q_j} (\numContributingLinks*_{Q_j}), s}(\vec{\mu}^{(j)}) \Big) 
            \underbrace{\big( \bar{\nu}_{Q_3} - \arrival - \nu_{\sigma_{Q_3} (\numContributingLinks*_{Q_3})} \big)}_{\leq 0, \text{ by \eqref{eq:prf:convexity:3}}}
            \\
            &\leq 
            F_s (\vec{\mu}^{(j)})
            - \sum_{i \in I_j \setminus I_3} \nu_{i} \big(\exitTime_{\sigma_{Q_j} (\numContributingLinks*_{Q_j}), s}(\vec{\mu}^{(j)}) - \exitTime_{i, s}(\vec{\mu}^{(j)})\big)
            + \sum_{i \in I_3 \setminus I_j} \nu_{i} \big(\exitTime_{\sigma_{Q_j} (\numContributingLinks*_{Q_j}), s}(\vec{\mu}^{(j)}) - \exitTime_{i, s}(\vec{\mu}^{(j)})\big)
            \\
            \;&\stackrel{\mathclap{\eqref{eq:prf:convexity:11},\eqref{eq:prf:convexity:12}}}{\leq}\;
            F_s (\vec{\mu}^{(j)})
            . \qedhere
        \end{align*}
    \end{proof}

    Finally, we can prove convexity. Using the linearity of the exit times on $P$, we compute
    \begin{align*}
        F_s (\vec{\mu}^{(3)}) &=
        \arrival T + T \cdot \big[\bar{\nu}_{Q_3} - \arrival\big]^-
        - \sum_{i \in I_3} \nu_{i} \exitTime_{i, s}(\vec{\mu}^{(3)})
        + \exitTime_{\sigma_{Q_3} (\numContributingLinks*_{Q_3}), s}(\vec{\mu}^{(3)}) \cdot \big[\bar{\nu}_{Q_3} - \arrival\big]^+
        \\
        &=
        \arrival T + T \cdot \big[\bar{\nu}_{Q_3} - \arrival\big]^-
        - \sum_{i \in I_3} \nu_i \bigl( \xi \exitTime_{i, s} (\vec{\mu}^{(1)}) + (1-\xi) \exitTime_{i, s} (\vec{\mu}^{(2)}) \bigr)
        \\
        &\qquad
        +
        \bigl( \xi \exitTime_{\sigma_{Q_3} (\numContributingLinks*_{Q_3}), s} (\vec{\mu}^{(1)}) + (1-\xi) \exitTime_{\sigma_{Q_3} (\numContributingLinks*_{Q_3}), s} (\vec{\mu}^{(2)}) \bigr) \cdot \big[\bar{\nu}_{Q_3} - \arrival\big]^+
        \\
        &=
        \xi \biggl(
        \arrival T + T \cdot \big[\bar{\nu}_{Q_3} - \arrival\big]^-
            - \sum_{i \in I_3} \nu_{i} \exitTime_{i, s}(\vec{\mu}^{(1)})
            +
            \exitTime_{\sigma_{Q_3} (\numContributingLinks*_{Q_3}), s}(\vec{\mu}^{(1)})
            \cdot \big[\bar{\nu}_{Q_3} - \arrival\big]^+
        \biggr)
        \\
        &\qquad +
        (1 - \xi) \biggl(
        \arrival T + T \cdot \big[\bar{\nu}_{Q_3} - \arrival\big]^-
            - \sum_{i \in I_3} \nu_{i} \exitTime_{i, s}(\vec{\mu}^{(2)})
            +
            \exitTime_{\sigma_{Q_3} (\numContributingLinks*_{Q_3}), s}(\vec{\mu}^{(2)})
            \cdot \big[\bar{\nu}_{Q_3} - \arrival\big]^+
        \biggr) \\
        &\stackrel{\mathclap{(*)}}{\leq}
        \xi F_s (\vec{\mu}^{(1)}) + (1- \xi) F_s (\vec{\mu}^{(2)})
        ,
    \end{align*}
    where we used \Cref{claim:prf:convexity} in the inequality marked with $(*)$.
\end{proof}

\bibliographystyle{abbrvnat}
\bibliography{soda}

\end{document}

%% file: fig-flowvalue.tex
\begin{figure}
    \hfill
    \begin{subfigure}[t]{0.47\textwidth}
        \centering
        \begin{tikzpicture}
            \fill[color1!10]
                (.75, 0) -- (.75, .5) 
                -| ++(1, .5)
                -| ++(.8, .6)
                -| ++(1, .5)
                -- (4.5, 2.1)
                -- (4.5,0) -- cycle;

            \fill[color2!10] 
                (0,0) -| ++(.75, .5) 
                -| ++(1, .5)
                -| ++(.8, .6)
                -| ++(1, .5) -| cycle;

            \draw[color2]
                (0,0) rectangle ++(.75, .5)
                (0,.5) rectangle ++(1.75, .5)
                (0,1) rectangle ++(2.55, .6)
                (0,1.6) rectangle ++(3.55, .5);

            \draw[thick, -latex] (-.5, 0) -- (6, 0) node[right] {$\theta$};
            \draw[thick, -latex] (0, -.5) -- (0, 3);

            \foreach \x/\i in {.75/1,1.75/2,2.55/3,3.55/4,5.3/5} {
                \draw[thick] (\x, .1) -- ++(0, -.2) node[below] {
                    \footnotesize $\exitTime_{\sigma(\i)}$
                };
            }
            \node[below] at (3.55, -.5) {\footnotesize $= \exitTime_{\sigma(K)}$};

            \foreach \y/\i in {.5/1,1/2,1.6/3,2.1/4} {
                \draw[thick] (.1, \y) -- ++(-.2, 0) node[left] {
                    \footnotesize $\bar{\nu}_{s} (\i)$
                };
            }

            \draw[thick, color1] 
                (0,0) -| ++(.75, .5) 
                -| ++(1, .5)
                -| ++(.8, .6)
                -| ++(1, .5)
                -| ++(1.75, .6) -- (6, 2.7);

            \draw[thick, dashed]
                (4.5, 3) -- (4.5, -.1) node[below] {\footnotesize $T$};

            \draw[thick, dashed]
                (6, 2.5) -- (-.1, 2.5) node[left] {\footnotesize $\arrival$};

            \draw[color3, thick]
                (0,0) rectangle (4.5, 2.1);
        \end{tikzpicture}
        \caption{Case $\bar{\nu} (\numContributingLinks*) < \arrival$: The flow value (blue area) is the difference between the green rectangle and the red rectangles.}
    \end{subfigure}\hfill%
    \begin{subfigure}[t]{0.47\textwidth}
        \centering
        \begin{tikzpicture}
            \fill[color1!10]
                (.75, 0) -- (.75, .5) 
                -| ++(1, .5)
                -| ++(.8, .6)
                -| ++(1, .5)
                -- (4.5, 2.1)
                -- (4.5,0) -- cycle;

            \fill[color2!10] 
                (0,0) -| ++(.75, .5) 
                -| ++(1, .5)
                -| ++(.8, .6)
                 -| cycle;

            \draw[color2]
                (0,0) rectangle ++(.75, .5)
                (0,.5) rectangle ++(1.75, .5)
                (0,1) rectangle ++(2.55, .6);

            \draw[pattern=north east lines, pattern color=color2]
                (0,1.3) rectangle ++(2.55, .3);

            \draw[thick, -latex] (-.5, 0) -- (6, 0) node[right] {$\theta$};
            \draw[thick, -latex] (0, -.5) -- (0, 3);

            \foreach \x/\i in {.75/1,1.75/2,2.55/3,3.55/4,5.3/5} {
                \draw[thick] (\x, .1) -- ++(0, -.2) node[below] {
                    \footnotesize $\exitTime_{\sigma(\i)}$
                };
            }
            \node[below] at (2.55, -.5) {\footnotesize $= \exitTime_{\sigma(K)}$};

            \foreach \y/\i in {.5/1,1/2,1.6/3,2.1/4} {
                \draw[thick] (.1, \y) -- ++(-.2, 0) node[left] {
                    \footnotesize $\bar{\nu}_{s} (\i)$
                };
            }

            \draw[thick, color1] 
                (0,0) -| ++(.75, .5) 
                -| ++(1, .5)
                -| ++(.8, .6)
                -| ++(1, .5)
                -| ++(1.75, .6) -- (6, 2.7);

            \draw[thick, dashed]
                (4.5, 3) -- (4.5, -.1) node[below] {\footnotesize $T$};

            \draw[thick, dashed]
                (6, 1.3) -- (-.1, 1.3) node[left] {\footnotesize $\arrival$};

            \draw[color3, thick]
                (0,0) rectangle (4.5, 1.3);
        \end{tikzpicture}
        \caption{Case $\bar{\nu} (\numContributingLinks*) \geq \arrival$: The flow value (blue area) is the difference between the green rectangle and the red rectangles. In this case, the red shared area has to be added again resulting in the extra term.}
    \end{subfigure}
    \hfill
    \caption{Visualization of the flow value $F_s (\vec{\mu})$ in scenario~$s$ for given belief~$\vec{\mu}$. The blue function is the total inflow into the sink, the blue area depicts the flow value.}
\end{figure}

%% file: fig_throughput_example.tex
\begin{figure}[bt]%
\hfill%
\begin{subfigure}{.48\textwidth}
    \centering
    \begin{tikzpicture}[scale=4]
        \node[node] (s) at (0,0) {\clap{\footnotesize s}};
        \node[node] (t) at (1,0) {\footnotesize \clap{\smash{t}}\llap{\phantom{s}}};

        \draw[-{Stealth}, very thick, mpgray]
            (s) edge[bend left=25] 
                node[pos=.5, above] {\footnotesize $\nu_1 = \nicefrac{1}{3}$ \quad $\vec\tau_1 = \sctau{1}{5}$}
            (t)
            (s) edge[bend right=25] 
                node[pos=.5, below] {\footnotesize $\nu_3 = \nicefrac{2}{3}$ \quad $\vec\tau_2 = \sctau{4}{3}$}
            (t);

        \draw[very thick, {Stealth}-, mppetrol]
            ($(s)+(-.075,0)$) -- ++(-.35, 0) node[midway, above] {\small $u=1$};
    \end{tikzpicture}
    \caption{The instance.}
    \label{subfig:throughput:instance}
\end{subfigure}\hfill%
\begin{subfigure}{.48\textwidth} \def\sy{6}
    \centering
    \begin{tikzpicture}[xscale=\sx, yscale=\sy]
        \def\FVoffset{-1.1}
        \def\muStar{7/16}

        \draw[thick, -latex, thick] (-.5/\sx, 0) -- (1+.75/\sx, 0) node[right] {\small $\mu$};
        \draw[thick, -latex, thick] (0, -.5/\sy) -- (0, .5) node[above] {\small $\FV$};

        \draw[thick] (1, .1/\sy) -- (1, -.1/\sy) node[below] {\footnotesize $1$};
        \draw[thick] (\muStar, .1/\sy) -- (\muStar, -.1/\sy) node[below] {\footnotesize $\lambda^*$};
        \draw[thick] (.6, .1/\sy) -- (.6, -.1/\sy) node[below] {\footnotesize $\nicefrac{3}{5}$};

        \draw[objectivefunction, domain=0:1/5] plot (\x, {1/3*(4+\x*(-3+5*\x))  + \FVoffset});
        \draw[objectivefunction, domain=1/5:3/5] plot (\x, {(1 + \x) + \FVoffset});
        \draw[objectivefunction, domain=3/5:1] plot (\x, {(4 + 2/3 * \x * (-9+5*\x)) + \FVoffset});
        
        \coordinate (truth1) at (0, 4/3 + \FVoffset);
        \coordinate (truth2) at (1, 4/3 + \FVoffset);
        \draw[myblue, very thick, dashed] (truth1) -- (truth2);
        \fill[myblue] (truth1) ellipse (2/\sx pt and 2/\sy pt) node[below left] {\footnotesize $\tilde{\mu}^{1}$};
        \fill[myblue] (truth2) ellipse (2/\sx pt and 2/\sy pt) node[below right] {\footnotesize $\tilde{\mu}^{2}$};

        \coordinate (opt1) at (truth1);
        \coordinate (opt2) at (3/5, 8/5 + \FVoffset);
        \coordinate (opt3) at (truth2);
        \draw[mpgreen, very thick, dashed] (opt1) -- (opt2) -- (opt3);
        \fill[mpgreen] (opt1) ellipse (2/\sx pt and 2/\sy pt) node[above left] {\footnotesize $\mu^{1}$};
        \fill[mpgreen] (opt2) ellipse (2/\sx pt and 2/\sy pt) node[above right] {\footnotesize $\mu^{2}$};
        \fill[mpgreen] (opt3) ellipse (2/\sx pt and 2/\sy pt) node[above right] {\footnotesize $\mu^{3}$};

        \coordinate (muStarTruth) at (\muStar, 4/3 + \FVoffset);
        \coordinate (muStarNoInfo) at (\muStar, \muStar + 1 + \FVoffset);
        \coordinate (muStarOpt) at (\muStar, 4/3 + 4/9 * \muStar + \FVoffset);
        \draw[thick, dashed] (muStarOpt) -- (\muStar, 0);

        \fill[myblue] (muStarTruth) ellipse (2/\sx pt and 2/\sy pt);
        \fill[mppetrol] (muStarNoInfo) ellipse (2/\sx pt and 2/\sy pt);
        \fill[mpgreen] (muStarOpt) ellipse (2/\sx pt and 2/\sy pt);

    \end{tikzpicture}\hfill%

    \caption{The throughput function with signaling schemes.}
    \label{subfig:throughput:function}
\end{subfigure}
\caption{The example for the throughput objective from~\Cref{sec:throughput-example}.
(a) The instance with $n=3$~edges and $m=2$~scenarios (\textcolor{mpblue}{blue} and \textcolor{mpred}{red}).
(b) The throughput for $T = 5$ depending on the probability \textcolor{mpred}{$\mu$} for the \textcolor{mpred}{red} scenario. \textcolor{myblue}{Full information revelation with conditional beliefs $\tilde{\mu}^1$ and $\tilde{\mu}^2$} is \emph{not optimal}, a \textcolor{mpgreen}{signaling scheme} incorporating the indifference point at $\mu^2 = \nicefrac{3}{5}$ is optimal.}
\label{fig:throughput}
\end{figure}

%% file: fig_makespan_example.tex
\begin{figure}%
    \hfill%
    \begin{subfigure}{.48\textwidth}
    \centering
        \begin{tikzpicture}[scale=4]
            \node[node] (s) at (0,0) {\clap{\footnotesize s}};
            \node[node] (t) at (1,0) {\footnotesize \clap{\smash{t}}\llap{\phantom{s}}};
    
            \draw[-{Stealth}, very thick, mpgray]
                (s) edge[bend left=60] 
                    node[pos=.5, above] {\footnotesize $\nu_1 = \nicefrac{1}{2}$ \quad $\vec\tau_1 = \sctau{0}{5}$}
                (t)
                (s) edge 
                    node[pos=.5, above] {\footnotesize $\nu_2 = \nicefrac{1}{3}$ \quad $\vec\tau_2 = \sctau{1}{1}$}
                (t)
                (s) edge[bend right=60] 
                    node[pos=.5, below] {\footnotesize $\nu_3 = \nicefrac{1}{2}$ \quad $\vec\tau_3 = \sctau{4}{0}$}
                (t);
    
            \draw[very thick, {Stealth}-, mppetrol]
                ($(s)+(-.075,0)$) -- ++(-.35, 0) node[midway, above] {\small $u=1$};
        \end{tikzpicture}
        \caption{The instance.}
        \label{subfig:makespan:instance}
    \end{subfigure}\hfill%
    \begin{subfigure}{.48\textwidth}
    \centering
        \begin{tikzpicture}[xscale=\sx, yscale=\sy]
            \draw[thick, -latex, thick] (-.5/\sx, 0) -- (1+.75/\sx, 0) node[right] {\small $\mu$};
            \draw[thick, -latex, thick] (0, -.5/\sy) -- (0, 4) node[above] {\small $\MS$};
    
            \draw[thick] (1, .1/\sy) -- (1, -.1/\sy) node[below] {\footnotesize $1$};
            \draw[thick] (.45, .1/\sy) -- (.45, -.1/\sy) node[below] {\footnotesize $\lambda^*$};
    
            \draw[mppetrol, thick, dotted] (0.1, 1.5) -- (0.1, 39/20);
            \draw[mppetrol, thick, dotted] (2/5, 37/10) -- (2/5, 5/2);
            \draw[mppetrol, thick, dotted] (1/2, 5/2) -- (1/2, 7/2);
            \draw[mppetrol, thick, dotted] (7/8, 31/16) -- (7/8, 3/2);
    
            \draw[objectivefunction, domain=0:0.1] plot (\x, 1+5*\x);
            \draw[objectivefunction, domain=0.1:0.2] plot (\x, {7/5 + (6-5*\x)*\x});
            \draw[objectivefunction, domain=0.2:0.4] plot (\x, {7/10 + 19/2*\x-5*\x*\x});
            \draw[objectivefunction] (2/5, 5/2) -- (1/2, 5/2);
            \draw[objectivefunction, domain=0.5:.75] plot (\x, {43/10+2/5*\x-4*\x*\x});
            \draw[objectivefunction, domain=.75:7/8] plot (\x, {11/5+4/5*(4-5*\x)*\x});
            \draw[objectivefunction, domain=7/8:1] plot (\x, {5-4*\x});
    
            \draw[mpgreen, very thick, dashed] (0,1) -- (1,1);
            \fill[mpgreen] (0, 1) ellipse (2/\sx pt and 2/\sy pt) node[left] {\footnotesize $\mu^{1}$};
            \fill[mpgreen] (1, 1) ellipse (2/\sx pt and 2/\sy pt) node[right] {\footnotesize $\mu^{2}$};
    
            \draw[myblue, very thick, dashed] (7/40, 147/64) -- (5/8, 239/80);
            \fill[myblue] (7/40, 147/64) ellipse (2/\sx pt and 2/\sy pt) node[above left] {\footnotesize $\tilde{\mu}^{1}$};
            \fill[myblue] (5/8, 239/80) ellipse (2/\sx pt and 2/\sy pt) node[above right] {\footnotesize $\tilde{\mu}^{2}$};
    
            \draw[thick, dashed] (0.45,15661/5760) -- (0.45,0);
    
            \fill[myblue] (0.45,15661/5760) ellipse (2/\sx pt and 2/\sy pt);
            \fill[mppetrol] (0.45,5/2) ellipse (2/\sx pt and 2/\sy pt);
            \fill[mpgreen] (0.45,1) ellipse (2/\sx pt and 2/\sy pt);
    
        \end{tikzpicture}
        \caption{The makespan function with signaling schemes.}
        \label{subfig:makespan:function}
    \end{subfigure}\hfill

    \caption{The example for the makespan objective from~\Cref{sec:makespan-example}.
    (a)
    The example instance with $n=3$~edges and $m=2$~scenarios (\textcolor{mpblue}{blue} and \textcolor{mpred}{red}).
    (b) 
    The makespan function for $T = \nicefrac{1}{2}$ depending on the probability \textcolor{mpred}{$\mu$} for the \textcolor{mpred}{red} scenario. \textcolor{mpgreen}{Full information revelation with conditional beliefs $\mu^1$ and $\mu^2$} is an optimal signaling scheme, \textcolor{myblue}{the other signaling scheme with conditional beliefs $\tilde{\mu}^1$ and $\tilde{\mu^2}$} is suboptimal.
    }
    \label{fig:makespan}
\end{figure}

%% file: arxiv.bbl
\begin{thebibliography}{30}
\providecommand{\natexlab}[1]{#1}
\providecommand{\url}[1]{\texttt{#1}}
\expandafter\ifx\csname urlstyle\endcsname\relax
  \providecommand{\doi}[1]{doi: #1}\else
  \providecommand{\doi}{doi: \begingroup \urlstyle{rm}\Url}\fi

\bibitem[Acemoglu et~al.(2018)Acemoglu, Makhdoumi, Malekian, and
  Ozdaglar]{Acemoglu18}
D.~Acemoglu, A.~Makhdoumi, A.~Malekian, and A.~Ozdaglar.
\newblock Informational {Braess}' paradox: {T}he effect of information on
  traffic congestion.
\newblock \emph{Oper.\ Res.}, 66\penalty0 (4):\penalty0 893--917, 2018.

\bibitem[Avis and Fukuda(1996)]{Avis96}
D.~Avis and K.~Fukuda.
\newblock Reverse search for enumeration.
\newblock \emph{Discret.\ Appl.\ Math.}, 65:\penalty0 21--46, 1996.

\bibitem[Bhaskar et~al.(2015)Bhaskar, Fleischer, and Anshelevich]{Bhaskar15}
U.~Bhaskar, L.~Fleischer, and E.~Anshelevich.
\newblock A {Stackelberg} strategy for routing flow over time.
\newblock \emph{Games Econ. Behav.}, 92:\penalty0 232--247, 2015.

\bibitem[Bhaskar et~al.(2016)Bhaskar, Cheng, Ko, and Swamy]{Bhaskar16}
U.~Bhaskar, Y.~Cheng, Y.~K. Ko, and C.~Swamy.
\newblock Hardness results for signaling in {Bayesian} zero-sum and network
  routing games.
\newblock In \emph{Proc.\ 17th Conf.\ Econ.\ Comput.\ (EC)}, pages 479--496,
  2016.

\bibitem[Boyd and Vandenberghe(2004)]{Boyd04}
S.~Boyd and L.~Vandenberghe.
\newblock \emph{Convex Optimization}.
\newblock Cambridge University Press, 2004.

\bibitem[Buck(1943)]{Buck43}
R.~C. Buck.
\newblock Partition of space.
\newblock \emph{Amer. Math. Monthly}, 50\penalty0 (9):\penalty0 541--544, 1943.

\bibitem[Castiglioni et~al.(2021)Castiglioni, Celli, Marchesi, and
  Gatti]{Castiglioni21}
M.~Castiglioni, A.~Celli, A.~Marchesi, and N.~Gatti.
\newblock Signaling in {Bayesian} network congestion games: the subtle power of
  symmetry.
\newblock In \emph{Proc.\ 35th Conf.\ Artif.\ Intell.\ (AAAI)}, pages
  5252--5259, 2021.

\bibitem[Cheng et~al.(2015)Cheng, Cheung, Dughmi, Emamjomeh{-}Zadeh, Han, and
  Teng]{Cheng15}
Y.~Cheng, H.~Y. Cheung, S.~Dughmi, E.~Emamjomeh{-}Zadeh, L.~Han, and S.~Teng.
\newblock Mixture selection, mechanism design, and signaling.
\newblock In V.~Guruswami, editor, \emph{Proc.\ 56th Symp.\ Found.\ Comput.\
  Sci.\ (FOCS)}, pages 1426--1445, 2015.

\bibitem[Cominetti et~al.(2015)Cominetti, Correa, and
  Larr{\'e}]{cominetti2015existence}
R.~Cominetti, J.~Correa, and O.~Larr{\'e}.
\newblock Dynamic equilibria in fluid queueing networks.
\newblock \emph{Oper. Res.}, 63\penalty0 (1):\penalty0 21--34, 2015.

\bibitem[Cominetti et~al.(2022)Cominetti, Correa, and Olver]{Cominetti22}
R.~Cominetti, J.~R. Correa, and N.~Olver.
\newblock Long-term behavior of dynamic equilibria in fluid queuing networks.
\newblock \emph{Oper. Res.}, 70\penalty0 (1):\penalty0 516--526, 2022.

\bibitem[Correa et~al.(2022)Correa, Cristi, and Oosterwijk]{Correa22}
J.~R. Correa, A.~Cristi, and T.~Oosterwijk.
\newblock On the price of anarchy for flows over time.
\newblock \emph{Math. Oper. Res.}, 47\penalty0 (2):\penalty0 1394--1411, 2022.

\bibitem[Dacorogna(2008)]{Dacorogna08}
B.~Dacorogna.
\newblock \emph{Direct Methods in the Calculus of Variations}.
\newblock Springer, New York, NY, 2nd edition, 2008.

\bibitem[Das et~al.(2017)Das, Kamenica, and Mirka]{Das17}
S.~Das, E.~Kamenica, and R.~Mirka.
\newblock Reducing congestion through information design.
\newblock In \emph{Proc. 55th Annual Allerton Conference on Communication,
  Control, and Computing}, 2017.

\bibitem[Dughmi(2014)]{Dughmi14}
S.~Dughmi.
\newblock On the hardness of signaling.
\newblock In \emph{Proc.\ 55th Symp.\ Found.\ Comput.\ Sci.\ (FOCS)}, pages
  354--363, 2014.

\bibitem[Graf et~al.(2020)Graf, Harks, and Sering]{Graf20}
L.~Graf, T.~Harks, and L.~Sering.
\newblock Dynamic flows with adaptive route choice.
\newblock \emph{Math. Program.}, 183\penalty0 (1):\penalty0 309--335, 2020.

\bibitem[Griesbach et~al.(2022)Griesbach, Hoefer, Klimm, and
  Koglin]{Griesbach22}
S.~M. Griesbach, M.~Hoefer, M.~Klimm, and T.~Koglin.
\newblock Public signals in network congestion games.
\newblock In D.~M. Pennock, I.~Segal, and S.~Seuken, editors, \emph{Proc.\ 23rd
  Conf.\ Econ.\ Comput.\ (EC)}, page 736, 2022.

\bibitem[Grötschel et~al.(1988)Grötschel, Lovász, and
  Schrijver]{Groetschel88}
M.~Grötschel, L.~Lovász, and A.~Schrijver.
\newblock \emph{Geometric Algorithms and Combinatorial Optimization}.
\newblock Springer, Berlin, Heidelberg, Germany, 1988.

\bibitem[Israel and Sering(2020)]{Israel20}
J.~Israel and L.~Sering.
\newblock The impact of spillback on the price of anarchy for flows over time.
\newblock In T.~Harks and M.~Klimm, editors, \emph{Proc.\ 13th Symp.\
  Algorithmic Game Theory (SAGT)}, pages 114--129. Springer, 2020.

\bibitem[Koch and Skutella(2011)]{Koch11}
R.~Koch and M.~Skutella.
\newblock Nash equilibria and the price of anarchy for flows over time.
\newblock \emph{Theory Comput. Syst.}, 49\penalty0 (1):\penalty0 71--97, 2011.

\bibitem[Lucchetti(2006)]{lucchetti2006convexity}
R.~Lucchetti.
\newblock \emph{Convexity and well-posed problems}.
\newblock Springer Science \& Business Media, 2006.

\bibitem[Massicot and Langbort(2019)]{Massicot19}
O.~Massicot and C.~Langbort.
\newblock Public signals and persuasion for road network congestion games under
  vagaries.
\newblock \emph{IFAC-PapersOnLine}, 51\penalty0 (34):\penalty0 124--130, 2019.

\bibitem[Nachbar and Xu(2021)]{NachbarX21}
J.~Nachbar and H.~Xu.
\newblock The power of signaling and its intrinsic connection to the price of
  anarchy.
\newblock In \emph{Proc.\ 3rd Intl.\ Conf.\ Distrib.\ Artif.\ Intell.\ (DAI)},
  pages 1--20, 2021.

\bibitem[Olver et~al.(2021)Olver, Sering, and {Vargas Koch}]{Olver21}
N.~Olver, L.~Sering, and L.~{Vargas Koch}.
\newblock Continuity, uniqueness and long-term behavior of {Nash} flows over
  time.
\newblock In \emph{Proc.\ 62nd Symp.\ Found.\ Comput.\ Sci.\ (FOCS)}, pages
  851--860, 2021.

\bibitem[Oosterwijk et~al.(2022)Oosterwijk, Schmand, and
  Schr{\"{o}}der]{Oosterwijk022}
T.~Oosterwijk, D.~Schmand, and M.~Schr{\"{o}}der.
\newblock Bicriteria {Nash} flows over time.
\newblock In K.~A. Hansen, T.~X. Liu, and A.~Malekian, editors, \emph{Proc.\
  18th Conf.\ Web and Internet Econ.\ (WINE)}, page 368, 2022.

\bibitem[Sering and Skutella(2018)]{Sering18}
L.~Sering and M.~Skutella.
\newblock Multi-source multi-sink {Nash} flows over time.
\newblock In R.~Bornd{\"{o}}rfer and S.~Storandt, editors, \emph{Proc.\ 18th
  Workshop on Algorithmic Approaches for Transportation Modelling (ATMOS)},
  volume~65, pages 12:1--12:20, 2018.

\bibitem[Sering and {Vargas Koch}(2019)]{Sering19}
L.~Sering and L.~{Vargas Koch}.
\newblock Nash flows over time with spillback.
\newblock In T.~M. Chan, editor, \emph{Proc.\ 30th Symp.\ Discret.\ Algorithms
  (SODA)}, pages 935--945, 2019.

\bibitem[Vasserman et~al.(2015)Vasserman, Feldman, and Hassidim]{Vasserman15}
S.~Vasserman, M.~Feldman, and A.~Hassidim.
\newblock Implementing the wisdom of {Waze}.
\newblock In \emph{Proc.\ 24th Int.\ Joint Conf.\ Artif.\ Intell.\ (IJCAI)},
  pages 660--660, 2015.

\bibitem[Vickrey(1969)]{Vickrey69}
W.~S. Vickrey.
\newblock Congestion theory and transport investment.
\newblock \emph{Am. Econ. Rev.}, 59\penalty0 (2):\penalty0 161--179, 1969.

\bibitem[Wu et~al.(2021)Wu, Amin, and Ozdaglar]{Wu21}
M.~Wu, S.~Amin, and A.~Ozdaglar.
\newblock Value of information in {Bayesian} routing games.
\newblock \emph{Oper.\ Res.}, 69\penalty0 (1):\penalty0 148--163, 2021.

\bibitem[Zhou et~al.(2022)Zhou, Nguyen, and Xu]{Zhou22}
C.~Zhou, T.~H. Nguyen, and H.~Xu.
\newblock Algorithmic information design in multi-player games: Possibility and
  limits in singleton congestion.
\newblock In \emph{Proc.\ 23rd Conf.\ Econ.\ Comput.\ (EC)}, page 869, 2022.

\end{thebibliography}
